\documentclass[runningheads,11pt]{llncs}
\pdfoutput=1
\usepackage{style}

\begin{document}

\sloppy

\title{Single-Round Proofs of Quantumness from Knowledge Assumptions}

\author{}
\institute{}

\author{Petia Arabadjieva\inst{1} \and
Alexandru Gheorghiu\inst{2} \and
Victor Gitton\inst{1} \and
Tony Metger\inst{1}}

\pagestyle{plain}

\institute{%
    ETH Zurich \\
    \email{\{petiaa,vgitton,tmetger\}@ethz.ch}\and
    Chalmers University of Technology \\
    \email{aleghe@chalmers.se}
}

\maketitle              

\begin{abstract}
A \emph{proof of quantumness} is an efficiently verifiable interactive test that an efficient quantum computer can pass, but all efficient classical computers cannot (under some cryptographic assumption). Such protocols play a crucial role in the certification of quantum devices. Existing single-round protocols (like asking the quantum computer to factor a large number) require large quantum circuits, whereas multi-round ones use smaller circuits but require experimentally challenging mid-circuit measurements. As such, current proofs of quantumness are out of reach for near-term devices.

In this work, we construct efficient single-round proofs of quantumness based on existing \emph{knowledge assumptions}. 
While knowledge assumptions have not been previously considered in this context, we show that they provide a natural basis for separating classical and quantum computation.
Specifically, we show that multi-round protocols based on Decisional Diffie-Hellman (DDH) or Learning With Errors (LWE) can be ``compiled'' into single-round protocols using a knowledge-of-exponent assumption~\cite{tkea} or knowledge-of-lattice-point assumption~\cite{k_assumption}, respectively.
We also prove an \emph{adaptive hardcore-bit} statement for a family of claw-free functions based on DDH, which might be of independent interest. 

Previous approaches to constructing single-round protocols relied on the random oracle model and thus incurred the overhead associated with instantiating the oracle with a cryptographic hash function. In contrast, our protocols have the same resource requirements as their multi-round counterparts without necessitating mid-circuit measurements, making them, arguably, the most efficient single-round proofs of quantumness to date.
Our work also helps in understanding the interplay between black-box/white-box reductions and cryptographic assumptions in the design of proofs of quantumness.
\end{abstract}

\section{Introduction}
Demonstrating quantum advantage, the point where a quantum computer can solve a problem that no existing classical computer can, is both a theoretical and technological challenge. It requires a problem that is plausibly intractable for classical algorithms and which admits an efficient quantum algorithm, ideally one that can be performed with noisy intermediate-scale quantum (NISQ) devices~\cite{preskill2018quantum}. 
Additionally, the problem's solution should be efficiently verifiable by a classical computer. This is necessary if one wishes to have a convincing and scalable way of proving quantum advantage. Currently, there are three main paradigms for demonstrating quantum advantage. 

The most straightforward one is to solve a problem which is believed to be classically hard, such as integer factorization. A quantum computer could efficiently solve this task by running Shor's algorithm~\cite{shor} and the solution can be efficiently verified by multiplying the factors reported by the prover. But while Shor's algorithm is efficient in the sense of only requiring a polynomial-size quantum circuit, the actual circuit for any reasonably-sized integer to be factored is too large to implement with NISQ devices~\cite{gidney2021factor,gouzien2021factoring}. Another approach is based on the classical hardness of sampling problems such as \emph{random circuit sampling}~\cite{bouland2019complexity,Arute2019,wu2021strong} or \emph{boson sampling}~\cite{aaronson2011computational,zhong2020quantum,madsen2022quantum}. While these experiments can be implemented with current hardware, they are not efficiently verifiable.

The work of Brakerski et al.~\cite{randomness} introduced a new approach towards testing for quantum advantage, referred to as a \emph{proof of quantumness} protocol. Akin to the cryptographic notions of proof or argument systems~\cite{goldwasser1989knowledge,brassard1988minimum,goldreich1998complexity}, a proof of quantumness is an interactive protocol between a polynomial-time classical \emph{verifier} and an ostensibly quantum polynomial-time \emph{prover}. The verifier issues challenges to the prover and checks the correctness of the prover's responses. The key feature of such a protocol is that there should exist an efficient quantum strategy that allows the prover to correctly answer the verifier's challenges with high probability, whereas any efficient classical strategy can only succeed with low probability (under some plausible cryptographic hardness assumption such as the classical intractability of factoring, Decisional Diffie-Hellman (DDH), or the Learning With Errors (LWE) problem).

In contrast to other paradigms for proving quantum advantage, the cryptographic proofs of quantumness of Brakerski et al.~and follow-up works do not require the quantum prover to break the underlying computational hardness assumption. Instead, they leverage the fact that the restriction imposed on the prover through mid-protocol interaction limits a classical prover's capacity for correctly responding to subsequent challenges from the verifier more strongly than it limits a quantum prover. 

The advantage of these protocols over an integer factorization-based test is that they have much smaller circuits than Shor's algorithm, making them potentially more suitable for implementation on near-term devices~\cite{kahanamoku2024fast,poqbell,zhu2023interactive,hirahara2021test,liu2022depth,alnawakhtha2022lattice}. However, due to their interactive nature, they require the honest quantum prover to perform mid-circuit measurements for each of the verifier's challenges. Mid-circuit measurements on a subset of qubits are difficult to implement on existing quantum devices without disturbing neighboring qubits and can thus degrade the quality of the remaining computation. An experimental implementation of the two-round protocol by Brakerski et al.~\cite{zhu2023interactive} directly compared the performance of an ion-trap quantum computer running the protocol with and without mid-circuit measurements, revealing a significant difference. Thus, implementing this proof of quantumness protocol at the scale required for quantum advantage seems especially challenging.\footnote{Note that the instance sizes of the underlying cryptographic problem in~\cite{zhu2023interactive} are so small that they can easily be broken by a laptop. For a convincing demonstration of quantum advantage, one would have to use instance sizes that cannot be broken even by the fastest classical supercomputers.}
It would therefore be desirable to have proofs of quantumness with relatively small quantum circuits and which involve only one round\footnote{Throughout the paper, we use the convention that a one-round protocol refers to a protocol with two messages, one message from the verifier to the prover and one message back.} of interaction between the verifier and the prover. Beyond the practical motivation, this would also give us a better understanding of the structure required for demonstrating quantum advantage in a way that is efficiently verifiable.

The recent breakthrough work of Yamakawa and Zhandry made progress in this direction by giving a single-round proof of quantumness protocol in the \emph{random oracle model} (ROM)~\cite{yamakawa2022verifiable}. Prior to their work, Brakerski et al.~constructed single-round proofs of quantumness in the ROM that additionally required a structured computational assumption, such as factoring or LWE~\cite{simplerpoq}. However, with both of these approaches the circuits that the prover would have to perform are larger than those in existing multi-round proofs of quantumness~\cite{randomness,poqbell}.

Our main result is a single-round proof of quantumness that only requires the same small circuits as existing multi-round interactive proofs of quantumness.
We achieve this by starting from the two-round protocol of~\cite{randomness} and removing one of the rounds of interaction through the use of a \emph{knowledge assumption}~\cite{koe,naor2003cryptographic,canetti2009towards,k_assumption,comp_with_kas}. 
As explained in~\cite{comp_with_kas}, a knowledge assumption is a statement of the form: 

\begin{displayquote}
``If an algorithm $\cA$ outputs an object of type $X$, it must know a corresponding witness of the type $W$, such that the output and the witness are in some relation $\mathcal{R} \subseteq X \times W$.''
\end{displayquote}

The rationale behind knowledge assumptions is that certain computational tasks, performed by some probabilistic algorithm $\cA$, can only be performed efficiently by following a specific sequence of steps, thus obtaining a series of intermediate values. Informally, we say that $\cA$ must have ``known'' the intermediate values for its specific output. This is made more precise by saying that there exists an efficient \emph{extractor} $\cA^*$ that receives as input $\cA$'s random coins and outputs (or \emph{extracts}) the relevant intermediate values of $\cA$.

Knowledge assumptions have traditionally been used for the design of protocols that require both extractability (like in proof/argument of knowledge protocols) and succinctness~\cite{bitansky2013succinct,comp_with_kas}. In our case, we are able to leverage knowledge assumptions to construct efficient, single-round proofs of quantumness. As far as we are aware, this is the first time knowledge assumptions are used in this context. 
We argue that their application here is natural and in some sense necessary, if one wishes to avoid multi-round interaction or working in the ROM. Intuitively, this is because some aspect of the proof of quantumness protocol has to differentiate between classical and quantum provers.
In multi-round protocols, this distinction comes from the fact that classical provers can be \emph{rewound}, but quantum provers generally cannot.
In a ROM-based protocol, the distinction arises from the ability to \emph{record} classical queries to a random oracle in a way that is not possible quantumly.

In our single-round protocols, the classical-quantum distinction is due to the knowledge assumption: for example, if $f: \cX \to \cY$ is a one-way function with a sparse range (i.e.~most values in $\cY$ are not in the range of $f$), it is plausible that a \emph{classical} prover that produces a value in the range of $f$ must have done so by evalating the function on some $x \in \cX$, and must therefore also ``know'' the corresponding preimage; this can be captured formally as a knowledge assumption~\cite{canetti2009towards,bitansky2014existence}.
However, a \emph{quantum} prover can do the following: first prepare the state $\sum_x \ket{x} \ket{f(x)}$, then measure the second register in the computational basis to obtain an image $y \in {\rm range}(f)$, and measure the first register in the Hadamard basis to ``erase'' any knowledge of the preimage $x$.
Such a prover can be said to produce a value in the range of $f$ without knowledge of a preimage.
It is this distinction between classical and quantum computation that our single-round proofs of quantumness exploit.
We note that a related idea of oblivious LWE sampling was recently explored in~\cite{debris2024quantum}; since a proof of quantumness only requires soundness against classical provers, we do not need to analyse the quantum prover in detail, but it might be interesting to explore the relation between knowledge assumptions used in our work and the (strong) notion of oblivious LWE sampling proposed in~\cite{debris2024quantum}.
We discuss the classical-quantum distinction and the impossibility of single-round proofs of quantumness with black-box security reductions in more detail in \cref{sec:intro_impossibility}.

The rest of this introduction is structured as follows: in \cref{sec:intro_bcmvv}, we give a brief overview of the two-round proof of quantumness from~\cite{randomness}.
This will form the basis for our single-round proof of quantumness.
In \cref{sec:intro_knowledge}, we recall two existing knowledge assumptions from~\cite{tkea,k_assumption} that we use in our protocols.
In \cref{sec:intro_results}, we explain how to make use of these knowledge assumptions to obtain different single-round proofs of quantumness.
In \cref{sec:intro_impossibility} we argue that a white-box assumption (like a knowledge assumption) seems to be necessary for single-round proofs of quantumness.
Finally, in \cref{sec:intro_relatedwork} and \cref{sec:intro_discussion} we discuss additional related works and open questions.

\subsection{A two-round proof of quantumness} \label{sec:intro_bcmvv}

To explain our results, we first need to outline the two-round (four-message) proof of quantumness protocol from~\cite{randomness}. At the heart of this protocol is a collection of functions known as Trapdoor Claw-free Functions (TCF). TCFs are a type of collision-resistant hash function---they are 2-to-1 functions for which it should be intractable to find a colliding pair of inputs (known as a \emph{claw}), given the description of the function. Additionally, the functions are generated with a trapdoor that allows for efficient inversion. The specific TCFs used in~\cite{randomness} require an additional property known as the \emph{adaptive hardcore-bit} (AHCB) property. Informally, this states that it is not only intractable to find collisions, but given any particular input $x_0$ and corresponding image under the TCF, denoted $y = f(x_0)$, it should be intractable to recover even a single bit of $x_1$, the other input with which $x_0$ forms a claw (i.e.~$f(x_1)=f(x_0)=y$).
Prior to our work, constructing TCFs with an AHCB had only been achieved from LWE and (non-standard) hardness assumptions of isogeny-based group actions~\cite{alamati2022candidate}. One of our results shows an AHCB property for a TCF based on DDH.

\begin{myprotocol} 
\caption{The Proof of Quantumness of \cite{randomness} (informal)}
\label{prot:bcmvv_informal}
\begin{myalgo}
    \State The verifier generates a description of a TCF $f$, together with its trapdoor, $t$. It then sends $f$ to the prover.
    \State The prover sends the verifier a point $y$ in the image of $f$. Denote as $x_0$ and $x_1$ the associated preimages, so $y=f(x_0)=f(x_1).$
    \State
    With probability $1/2$, the verifier sends the prover one of the following two challenges:
    \begin{enumerate}
        \item[(a)] \textbf{Preimage test.} The verifier asks the prover for a valid preimage of $y$. Denoting the prover's response as $x$, the verifier accepts iff $f(x)=y$.
        \item[(b)] \textbf{Equation test.} The verifier asks the prover for a non-zero string $d$, such that $d \cdot (x_0 \oplus x_1) = 0$. The verifier uses the trapdoor, $t$, to recover $x_0$ and $x_1$ from $y$ and then checks whether the equation is satisfied, accepting if it is and rejecting otherwise.
    \end{enumerate}
\end{myalgo}
\end{myprotocol}

Given a family of TCFs, with the AHCB property, the Brakerski et al.~protocol from ~\cite{randomness} is described informally in \protref{prot:bcmvv_informal}. Brakerski et al.~showed that, assuming the AHCB property, no polynomial-time classical prover makes the verifier accept with probability non-negligibly larger than $3/4$ (this is referred to as the \emph{soundness} of the protocol). At the same time, they gave a simple quantum strategy which would allow the prover to succeed with probability $1$ (referred to as the \emph{completeness} of the protocol). This honest prover first creates an equal superposition over evaluations of $f$,
\[
\sum_x \ket{x}\ket{f(x)}\,.
\]
The prover then measures the second register, resulting in a value $y = f(x_0) = f(x_1),$ and collapsing the state in the first register to
\begin{equation} \label{eq:preimage_superpos}
\frac{\ket{x_0} + \ket{x_1}}{\sqrt{2}}\,.    
\end{equation}
We can see that if this state is measured in the computational basis, it results in one of the two preimages of $y$, thus providing a valid response to the preimage test. Conversely, if the state is measured in the Hadamard basis (i.e.~applying Hadamard gates to all qubits and measuring in the computational basis), the result will be a string $d$ such that $d \cdot (x_0 \oplus x_1) = 0$, yielding a valid response to the equation test. Thus, a quantum prover implementing this strategy would make the verifier accept with probability 1.

The intuition for why it is classically intractable to succeed in this protocol is that a classical prover that answers both the preimage and equation tests correctly can be \emph{rewound} so as to obtain both a valid preimage and a valid equation. However, having both would contradict the AHCB property. In contrast, a quantum prover cannot be rewound and can use the state from~\cref{eq:preimage_superpos} to answer \emph{either of the two challenges} correctly, but not both at the same time. Communication between the two rounds is crucial for the security proof, as a reduction to the AHCB property requires that the prover holds a preimage and equation \emph{for the same} $y$. This can only be guaranteed if the prover commits to a fixed choice of $y$ before receiving the second challenge. 

A natural question is whether the equation test on its own is already classically intractable. 
Unfortunately, simply removing the preimage test breaks the security proof of~\cite{randomness}: the AHCB property says that it is hard to produce an image, an equation, \emph{and a preimage} together.
Therefore, without the preimage test one cannot use a successful classical prover in the proof of quantumness to break the AHCB property, as one does not have access to a preimage\footnote{We remark here that this is actually not the only problem that arises when removing the preimage test. In fact, Urmila Mahadev observed that for both constructions we consider, there exists a classical winning strategy in the equation test which involves evaluating $f$ on an extended domain. In the original Brakerski et al. protocol~\cite{randomness}, this strategy is excluded by the preimage test. In our protocols, this strategy is excluded by adding another type of test which does not increase the number of rounds of the protocol and is explained in \cref{sec:intro_results}.\label{footn:extdom}}.
Indeed, it can be shown that no \emph{black-box} reduction can reduce the security of  this ``equation-test only'' proof of quantumness to LWE (\cite{morimae2023quantum}, see also \cref{sec:intro_impossibility}).

Our results show that a variation of the "equation-test only" proof of quantumness can be made to work, provided we use a knowledge assumption to replace the role of the preimage test. Intuitively, the knowledge assumption can be used to ``extract'' a preimage from a successful classical prover without the need for an explicit preimage test.
Since a knowledge assumption deals with the inner workings of a prover, our result can be interpreted as a \emph{white-box} reduction from a version of the equation test to the AHCB property.

\subsection{Knowledge assumptions} \label{sec:intro_knowledge}

Knowledge assumptions posit the existence of an efficient extractor that is able to produce certain intermediate values that are in some relation with the output of an algorithm. A canonical example is the so-called \emph{knowledge of exponent assumption} (KEA), introduced by Damgård in~\cite{koe}. Informally, this says that given a generator $g$ of some multiplicative group $\mathbb{G}$, as well $g^\alpha$ for some random power $\alpha$, the only way to produce a new pair of the form $(h, h^\alpha)$ is by \emph{exponentiating} the original pair. In other words, any efficient (randomized) algorithm $\mathcal{A}$ which outputs $(h, h^\alpha)$ given $(g, g^\alpha)$ must ``know'' an exponent $k$ such that $h=g^k$. 
This is formalized by saying that for every such algorithm $\cA$, there exists an efficient extractor $\mathcal{A}^*$ that, given $(g, g^\alpha)$ and the random coins of $\mathcal{A}$ for which $\cA$ outputs $(h, h^\alpha)$, will output (or \emph{extract}) the exponent $k$ such that $h = g^k$. Variations of this assumption have also been considered, such as the $t$-KEA, in~\cite{tkea}, in which the input to $\cA$ is instead of the form $(g^\mathbf{r}, g^{\alpha \mathbf{r}})$, where $\mathbf{r}$ is a $t$-dimensional vector and exponentiation is element-wise.

The rationale behind knowledge of exponent assumptions is that the function that maps $k$ to ($g^k,g^{\alpha k}$) has a sparse image (considered as a subset of $\mathbb{G} \times \mathbb{G}$). Thus, it seems implausible that $\cA$ could come up with a valid image simply by cleverly sampling its output from $\mathbb{G} \times \mathbb{G}$ without exponentiating the input.\footnote{An alternative way of coming up with an image would be if $\cA$ could determine $\alpha$, given $g$ and $g^{\alpha}$. 
However, this would entail solving the discrete logarithm problem, which is believed to be classically intractable.}

More recently, similar knowledge assumptions were proposed for lattices. One example, introduced in~\cite{k_assumption}, is known as the \emph{knowledge of lattice point assumption} (denoted as LK-$\epsilon$). This says the following: suppose there is an efficient randomized classical algorithm $\cA$ which takes as input a lattice basis $\mathbf{A}$ and outputs a point $\mathbf{c}$ that is $\epsilon$-close to the lattice $\cL(\mathbf{A})$. Then $\cA$ must ``know'' the lattice point closest to $\mathbf{c}$. The motivation for this assumption is similar to that of KEA---the set of points that are close to the lattice (for a suitable choice of $\epsilon$) is sparse in the set of all points, and the only apparent efficient strategy to sample from this sparse set is to pick a random lattice point and perturb it. This lattice knowledge assumption and variants of it have been used to construct efficient FHE and SNARK schemes~\cite{gennaro2018zksnarks,ishai2021zksnarks}.

For our results, we will use $t$-KEA and LK-$\epsilon$, which we state formally as~\cref{assump:tkea} and~\cref{assumpt:lke}.

\subsection{Main results} \label{sec:intro_results}

Our main result is that combining knowledge assumptions with standard cryptographic assumptions, like LWE or DDH, leads to efficient single-round proof of quantumness protocols.
To make our results modular, we first show how to construct a general single-round proof of quantumness from a cryptographic primitive that we call Doubly Extended Extractable Noisy Trapdoor Claw-free Functions (abbreviated \eiiintcf).
Second, we give two constructions of \eiiintcf: one from the DDH and $t$-KEA assumptions, and one from the LWE and LK-$\eps$ assumptions.

\subsubsection{Single-round proof of quantumness from \eiiintcf.}
Our starting point is the protocol from~\cite{randomness}, which we explained in \cref{sec:intro_bcmvv}.
Recall that there, one uses a TCF with the AHCB property, and argues that if a classical prover could succeed in the preimage and equation tests, by rewinding we could construct a tuple $(x, y, d)$ of a preimage, image, and equation that would contradict the AHCB property.

Now suppose that the TCF family $\cF = \{f: \cX \to \cY\}$ used in this protocol had an additional extractability property: for any classical prover that produces an image $y \in {\rm range}(f)$, there exists an extractor that produces a corresponding preimage $x$.
This is, in essence, a knowledge assumption.
With such an extractable TCF, we could simply remove the preimage test from~\protref{prot:bcmvv_informal}: then, if we had a classical prover that succeeded in this modified protocol, we could use that prover to find an image $y \in {\rm range}(f)$ and equation $d$, and use the extractor to find a corresponding preimage $x$, such that $(x, y, d)$ break the AHCB property.

Unfortunately, we do not know how to construct such an extractable TCF with AHCB from existing knowledge assumptions such as $t$-KEA or LK-$\eps$.
The key difficulty here is that the extractability and AHCB properties have to hold simultaneously.

One way to circumvent this issue would be to introduce an additional function family $\cH = \{h: \cX \to \cY\}$ that is indistinguishable from the TCF family $\cF$ (i.e.~given a description of a random choice of either kind of function, it is hard to tell which kind of function it is).
This function family $\cH$ can be constructed such that it has an extractability property, i.e.~if a classical algorithm produces a value $y$ in the image of $h$, an extractor can find a preimage $x$ (under a standard knowledge assumption such as $t$-KEA).
However, $\cH$ itself is not a TCF and has no AHCB property.

We now want to combine the AHCB property of $\cF$ with the extractability property of $\cH$.
For this we leverage that the two function families are computationally hard to distinguish\footnote{One may ask here why we do not use computational indistinguishability to directly transfer the extractability property of $\cH$ to $\cF$, given that the extractor is an efficient algorithm. This has to do with the exact form of the extractability property: successful extraction is only guaranteed under the condition $y \in {\rm range}(h)$, which cannot be checked efficiently without the trapdoor. Thus, while we can exploit the computational indistinguishability of $\cF$ and $\cH$ in our protocol, it is not possible to use computational indistinguishability to infer an analogous extractability property for $\cF$.}:
if we send a description of either $f \in \cF$ or $h \in \cH$ to a classical prover, the prover cannot tell which kind of function it received.
This suggests the following protocol.

\begin{myenumi}
\item The verifier generates a description of a TCF $f \in \cF$ (with a trapdoor $t$) or an extractable function $h \in \cH$ and sends this function description to the prover.
\item The verifier receives $(y,d)$ from the prover.
\begin{itemize}
\item If the verifier sent a TCF $f$, it uses the trapdoor to recover $x_0$ and $x_1$ such that $y = f(x_0) = f(x_1)$ and accepts iff $d \cdot (x_0 \oplus x_1) = 0$.
\item If the verifier sent an extractable function $h$, it accepts iff $y \in {\rm range}(h)$.
\end{itemize}
\end{myenumi}

Suppose that a classical prover succeeded in this protocol with high probability.
We can use this prover (and the corresponding extractor for the extractable function family) to break the AHCB property of $f$ as follows:
given a description of $f \in \cF$, run the prover to generate $(y,d)$.
Then use the extractor on this prover on input $f \in \cF$ to generate a preimage $x$.
We claim that $(x, y, d)$ violates the AHCB property.

Note that this reduction runs the extractor for the function family $\cH$ on an input $f \in \cF$, i.e.~it runs the extractor on an input for which it was not designed.
However, recall that descriptions of $f$ and $h$ are computationally indistinguishable.
On the other hand, the extractor as well as the function that checks whether the extractor produced a correct preimage are efficient.
Therefore, since the protocol ensures that on input $h$, the classical prover produces $y \in {\rm range}(h)$, it follows that when we run the extractor for this classical prover on an input $f \in \cF$, it will still produce a valid preimage of $y = f(x_0) = f(x_1)$; otherwise, we could distinguish $f$ from $h$.

When trying to construct such a pair $(\cF, \cH)$ of function families from DDH and $t$-KEA, we are faced with one additional technical challenge: 
the only evident construction of an extractable function family $\cH$ works by extending the functions $f \in \cF$ to a larger domain $\cX'$ (on which $\cF$ no longer satisfies the AHCB property) and constructing extractable functions $h: \cX' \to \cY$, i.e. the extractability property of $h$ only holds with respect to the extended domain $\cX'$. This means that the extractor may, on input $y \in f(\cX)$, produce a $x \in \cX' \setminus \cX$ that also maps to $y$.
While this is a valid preimage to $y$, it is not useful for breaking the AHCB property: for that we need a preimage $x \in \cX$.

We therefore introduce a third function family, $\cG = \{g: \cX' \to \cY\}$, that is computationally indistinguishable from both $\cH$ and the extension of $\cF$.
The functions in $\cG$ are injective even on the larger domain $\cX'$.
In our proof of quantumness, the verifier will check that when sent a function $g$, the prover returns an image $y \in g(\cX)$, i.e.~the image $y$ must be in the range of the restricted domain $\cX$.
This essentially forces the prover to evaluate any function it receives only on the restricted domain\footnote{This also serves to exclude the ``extended-domain'' attack mentioned in \cref{footn:extdom}.} $\cX$, since if it evaluated the \emph{injective} function $g$ on an input $x \in \cX' \setminus \cX$, the verifier would reject the resulting image $y$.

In summary, our single-round proof of quantumness relies on a triplet of function families $(\cF, \cG, \cH)$ with the following properties:\footnote{As usual, this primitive depends on a security parameter $\lambda$, which we suppress for readability.}
\begin{definition}[\eiiintcf, informal]
A tuple of function families $(\cF, \cG, \cH)$ is called a \emph{Doubly Extended Extractable Noisy Trapdoor Claw-free Functions} (abbreviated \eiiintcf) if
\begin{myenumi}
    \item $\cF$ is a TCF family with an AHCB property.
    \item $\cG$ is an injective trapdoor one-way function family.
    \item $\cH$ is an extractable one-way function family, in the sense that for any algorithm that takes the description of $h \in \cH$ as input and outputs $y \in {\rm range}(h)$, there exists an extractor which outputs a preimage $x$ such that $h(x) = y$.
    \item  The functions are computationally indistinguishable from each other. In other words, given a description of a function from one of the three families, no polynomial-time classical algorithm can determine the function's type with probability non-negligibly greater than $1/3$.
\end{myenumi}
\end{definition}
For the formal statement, see~\cref{def:e^2ntcf}.

The notion of an \eiiintcf{} function family is an extension of Injective Invariant Noisy TCFs, which were introduced in~\cite{mahadev2018classical} to derive a protocol for verifying general quantum computations, and which only use the first two types of families $\cF$ and $\cG$.

The preceding discussion naturally leads to the single-round~\protref{prot:informal 3test} (see~\protref{poq:3test} for the formal protocol), based on an \eiiintcf. Our main result is that this protocol is a proof of quantumness, i.e.~that an efficient quantum prover can succeed with high probability, but no efficient classical prover can.

\begin{myprotocol} 
    \caption{Single-Round Proof of Quantumness based on \eiiintcf{} (informal)}
    \label{prot:informal 3test}
    \begin{myalgo}
        \State The verifier samples a challenge type $a \leftarrow_U \{\mathtt{Eq}, \mathtt{sIm},\mathtt{wIm}\}$. 
        \MyIf{$a = \mathtt{Eq}$}
            \MyState Verifier samples $f \in \cF$ with trapdoor $t_f$ and sends a description of $f$ to the prover.
            \State Verifier receives tuple $(y, d)$ and accepts iff $d \cdot (x_0 \oplus x_1) = 0$ (and $d \neq 0$). Here, $(x_0, x_1)$ 
            \StateNoNumber are the two  preimages of $y$, which the verifier can efficiently compute with $t_f$.
        \MyElsIf{$a = \mathtt{sIm}$}
            \MyState Verifier samples $g \in \cG$ with trapdoor $t_g$ and sends a description of $g$ to the prover.
            \MyState Verifier receives tuple $(y, d)$ and accepts iff $y \in g(\cX)$. This check is efficient using $t_g$.
        \MyElsIf {$a = \mathtt{wIm}$}
            \MyState Verifier samples $h \in \cH$ with trapdoor $t_h$ and sends a description of $h$ to the prover.
            \State Verifier receives tuple $(y, d)$ and accepts iff $y \in h(\cX')$. This check is efficient using $t_h$.
        \MyEndIf
    \end{myalgo}
\end{myprotocol}

\begin{theorem}[informal] \label{thm:pos_informal}
    Suppose that $(\cF, \cG, \cH)$ is an \eiiintcf.
    Then \protref{poq:3test} is a proof of quantumness, i.e.,
    \begin{myenumi}
    \item \textbf{Completeness.} There exists an efficient quantum prover that succeeds with probability $1 - \negl(\lambda)$.
    \item \textbf{Soundness.} Every efficient classical prover succeeds with probability at most $5/6+\negl(\lambda)$.
    \end{myenumi}
\end{theorem}
In the theorem, $\lambda$ is the security parameter.
As usual, the families $(\cF, \cG, \cH)$ implicitly depend on the security parameter and ``efficient'' classical or quantum provers are those whose runtime scales polynomially in $\lambda$. For the formal statement, see~\cref{thm:poq2}.

We have already sketched the security proof when we explained why we need to introduce the three different function families $\cF, \cG$, and $\cH$.
We briefly summarize the role that each of these function families and the associated challenge types $\mathtt{Eq}, \mathtt{sIm}$, and $\mathtt{wIm}$ play.\
\begin{itemize}
\item[\textbullet] The \emph{weak image test} (challenge type $\mathtt{wIm}$) uses the extractable function family $\cH$. This test ensures that for any classical prover in the protocol, there exists an extractor that extracts a preimage $x \in \cX'$ to the output $y = h(x)$ produced by the prover.
Note that as discussed above, this preimage might in principle be from the extended domain $\cX'$, not just $\cX$.
\item[\textbullet] The \emph{strong image test} (challenge type $\mathtt{sIm}$) uses the injective function family $\cG$.
This test, combined with the indistinguishability of the function families $\cF, \cG$, and $\cH$, ensures that the prover evaluates any function $f$, $g$, or $h$ only on inputs in the restricted domain $\cX$.
Furthermore, from this we can prove that for any $y \in \cY$ that has a preimage in $\cX$ under a given \emph{TCF function} $f$, the extractor will output such a preimage (rather than a different preimage in $\cX' \setminus \cX$).
\item[\textbullet] The \emph{equation test} (challenge type $\mathtt{Eq}$) uses the TCF function family $\cF$ that has the AHCB property.
The verifier checks that the prover's output $(y, d)$ satisfies the equation $d \cdot (x_0 \oplus x_1) = 0$, with $(x_0, x_1)$ the preimages of $y$.
From the strong image test, we also know that the extractor for a successful classical prover will produce a preimage $x \in \cX$ to $y.$
Together, $(x, y, d)$ would break the AHCB property of $f$, so no classical prover can win with very high probability.
\end{itemize}

A key objective in the design of our single-round proofs of quantumness was that the required circuit sizes should not be larger than for multi-round proofs of quantumness.
\protref{poq:3test} achieves this: to pass in the protocol, a quantum prover can simply prepare the state $\sum \ket{x} \ket{p(x)}$ for a given function $p \in \cF \cup \cG \cup \cH$, measure the first register in the Hadamard basis to get a string $d$, and measure the second register in the computational basis to get an image $y$.
These are exactly the same actions as those of an honest prover in the equation test in \protref{prot:bcmvv_informal},\footnote{In our constructions of \eiiintcf, the functions $g \in \cG$ and $h \in \cH$ are essentially as costly to evaluate as $f \in \cF$ in \protref{prot:bcmvv_informal}, so introducing these additional function families does not increase the demands on an honest prover.} except that now the honest prover can perform the entire measurement in one step without experimentally challenging mid-circuit measurements.

\subsubsection{Instantiating \eiiintcf{} families from DDH and  LWE.}
We show that \eiiintcf{} can be instantiated either from DDH and $t$-KEA, or from LWE and LK-$\epsilon$. 
Here, we only state the main results and defer a more detailed discussion of the construction to \cref{section:ddhntcf,section:lwentcf}, respectively.

For the LWE-based construction, we already know that the LWE-based TCF from~\cite{randomness} has an AHCB property. We can combine this with suitable injective and extractable one-way functions.
In fact, it turns out that for the LWE-based construction, the roles of the injective and extractable functions in \protref{poq:3test} can be played by the same function family, i.e.~we only require two distinct families $(\cF, \cG)$.
This simplifies the analysis somewhat, as we explain in \cref{section:2test}.
Combined with \cref{thm:pos_informal}, we obtain the following:
\begin{theorem}[Informal]
    Assuming the classical intractability of LWE and the lattice knowledge assumption LK-$\epsilon$, there exists a single-round proof of quantumness with completeness $1-\negl(\lambda)$ and soundness $3/4 + \negl(\lambda)$.
    The circuit sizes for an honest prover in this protocol are identical to the circuit sizes in the 2-round protocol from~\cite{randomness}.
\end{theorem}

While a family of TCFs based on DDH was considered in~\cite{poqbell} to construct a 3-round proof of quantumness, it had not been shown that these functions have an AHCB property. We prove this through a lossy sampling technique similar to the proof in~\cite{randomness}, showing a reduction from the AHCB property to the \emph{matrix $d$-linear assumption}~\cite{matrixdlin} (which is implied by standard DDH).
We then again augment this TCF family with an injective and an extractable family to get the following result.
\begin{theorem}[Informal]
    Assuming the classical intractability of DDH and the knowledge of exponent assumption $t$-KEA, there exists a single-round proof of quantumness with completeness\footnote{The reason completeness here is $0.99$ instead of $1$ is inherited from~\cite{poqbell}. There, the quantum prover does not create a superposition over the exact range of the TCF but over a superset. However, it can be shown that the exact range contains at least a $0.99$ fraction of the elements in the superset.} $0.99 - \negl(\lambda)$ and soundness $5/6 + \negl(\lambda)$.  The circuit sizes for an honest prover in this protocol are identical to the circuit sizes in the 3-round protocol from~\cite{poqbell}.
\end{theorem}

\subsection{Impossibility of black-box reductions}
\label{sec:intro_impossibility}

Knowledge assumptions are non-falsifiable cryptographic assumptions~\cite{naor2003cryptographic,gentry2011separating}, which makes their use somewhat undesirable.
A natural question is whether our main results, single-round proofs of quantumness, can also be achieved using only standard falsifiable (also called game-based) assumptions such as DDH or LWE, or even weaker assumptions like the existence of one-way functions.

As was observed in \cite{morimae2023quantum}, the security of a single-round proof of quantumness cannot be reduced to a \emph{quantumly hard} problem like LWE in a black-box manner.
The reason for this is simple: if there existed a black-box reduction from a successful classical prover to an algorithm for breaking LWE, we could also apply that reduction to the honest quantum prover in the proof of quantumness.
Since the honest quantum prover is required to succeed with high probability, this would give a quantum algorithm for LWE. This rules out the existence of such a black-box reduction.\footnote{Note that this argument does not apply to interactive protocols, ROM-based protocols, or protocols that use whitebox assumptions like ours: in those cases, one cannot simply run the reduction on the quantum prover as the reduction may perform operations that only work for classical provers, e.g.~rewinding or using knowledge extractors.}

An extension of this reasoning suggests that the use of knowledge assumptions is still justified even for single-round proofs of quantumness based on computational problems that are quantumly easy: now it is of course not a contradiction to have a quantum algorithm that breaks the underlying cryptographic assumption, but we can argue that implementing an honest quantum prover in such a protocol is no easier than breaking the assumption outright: Suppose we had a single-round proof of quantumness with a black-box reduction to factoring.
Then we could again use that reduction with the honest quantum prover to construct a factoring algorithm whose only quantum component is repeatedly (and independently) running the honest prover.
In this sense, implementing the honest prover is as hard as breaking factoring: if we had an honest prover that was much easier to implement than Shor's algorithm, we could use that prover and the black-box reduction to construct a much more practical quantum factoring algorithm. In contrast, in the DDH-based proof of quantumness from~\cite{poqbell} and also in our single-round version thereof, implementing the honest quantum prover does not require computing discrete logarithms. 

\subsection{Related work} \label{sec:intro_relatedwork}

As mentioned in~\cref{sec:intro_bcmvv}, cryptographic proofs of quantumness were introduced in~\cite{randomness}. Since then, there have been a number of follow-up works aiming to understand the types of cryptographic constructions on which these protocols can be based, as well as how to make them more efficient~\cite{kahanamoku2024fast,simplerpoq,liu2022depth,hirahara2021test,alnawakhtha2022lattice,kalai2023quantum,morimae2023quantum}. So far, only small-scale demonstrations of these protocols have been performed~\cite{zhu2023interactive,lewis2024experimental}, though the hope is that by further reducing the sizes of the quantum circuits and the amount of interaction between the verifier and the prover, these protocols can be performed with NISQ devices.

The only existing constructions of single-round proofs of quantumness (apart from trivial ones like asking the prover to factor a large number) use the random oracle model~\cite{simplerpoq,yamakawa2022verifiable}. In both of these prior works, the quantum circuits required to succeed in these protocols do not have to perform mid-circuit measurements. However, the circuits are prohibitively larger than the ones in multi-round proofs of quantumness. This is partly due to the fact that one would have to instantiate the random oracle with plausible cryptographic hash functions, like SHA-2 or SHA-3. Running these functions as quantum circuits in superposition seems to be out of reach for near-term quantum computers~\cite{jang2024sha}.

Since our construction of \eiiintcf{}s makes use of extractable one-way functions (EOWFs), it is natural to ask whether we could have used existing results concerning these functions~\cite{canetti2009towards,bitansky2014existence,bitansky2020weakly}. In particular, in~\cite{bitansky2014existence}, Bitansky et al. prove the existence of a class of EOWFs from the subexponential hardness of LWE against adversaries with bounded auxilary input, without having to rely on a knowledge assumption. There are at least two obstacles towards applying those results to our setting. First, the extractable one-way functions we require have to be computationally indistinguishable from TCFs. It is unclear whether the EOWFs of Bitansky et al. can be suitably adapted to satisfy this requirement. Second, since one of the motivations for our work is to devise more efficient single-round proofs of quantumness, we would like to ensure that the honest quantum prover's circuits are at most as large as in the multi-round protocols. However, constructing a family of \eiiintcf{}s with the EOWFs in~\cite{bitansky2014existence} would likely require larger circuits.

Knowledge assumptions in the quantum setting have also been considered in~\cite{liu2023another} to derive \emph{quantum money} and \emph{quantum lightning} schemes. There, the assumptions are required to be sound against quantum adversaries. This introduces some challenges in formalizing the appropriate notion of a quantum knowledge assumption. The subsequent work of~\cite{zhandry2023quantum} also discusses this point. Interestingly, this latter work points out how certain knowledge assumptions can be generically broken by a quantum adversary. In some sense, this is exactly what we exploit to arrive at our single-round proofs of quantumness.

\subsection{Discussion and open problems} \label{sec:intro_discussion}

We have shown how existing knowledge assumptions, together with standard cryptographic assumptions, lead to efficient single-round proofs of quantumness. Our work opens up several avenues for future research.

One such avenue is the possibility of a \emph{white-box proof} for a single-round proof of quantumness based only on the classical intractability of, say, LWE.
This would circumvent the impossibility discussed in \cref{sec:intro_impossibility}.
As mentioned in the previous section, one possible approach would be to use the extractable one-way functions based on subexponential LWE from~\cite{bitansky2014existence}, with the caveat that this would provide soundness against classical provers with bounded auxiliary input. 
We leave this as an interesting but challenging open problem.

Currently, our construction necessitates the use of a TCF which satisfies the AHCB property. This is a fairly strong requirement which so far has only been achieved from LWE~\cite{randomness}, isogeny-based group actions~\cite{alamati2022candidate}, and, in this work, the DDH assumption. It would be desirable to use TCFs that are not known to have an AHCB property, such as those based on Ring-LWE or Rabin's function ($x^2~{\rm mod}~N)$, since those require smaller circuits than the TCFs based on LWE or DDH, as outlined in~\cite{poqbell,kahanamoku2024fast}. However, it is unclear if our construction can be modified so as to not require the AHCB property; in the equation test, the prover is allowed to output \emph{any} valid equation, hence the need for an \emph{adaptive} hardcore bit. The protocol in~\cite{poqbell} is able to circumvent this and use TCFs without an AHCB by introducing an additional round of interaction between the verifier and the prover, leading to a 3-round protocol. Unfortunately, there does not seem to be an obvious way of employing knowledge assumptions to reduce the round complexity of the protocol from~\cite{poqbell}. Thus, a fundamentally different approach would be needed to achieve a single-round protocol without relying on the AHCB property. Of course, for the purpose of having more practical protocols, proving an AHCB statement for the TCFs based on Ring-LWE or Rabin's function (even relying on knowledge assumptions) would be equally useful.

We remark that the LWE-based approach is expected to be more efficient than the DDH one. This is because, as is also discussed in~\cite{poqbell}, the DDH approach requires performing a large number of group exponentiations coherently. It would therefore seem that the DDH approach is less efficient than performing Shor's algorithm to solve the discrete logarithm problem. However, as~\cite{poqbell} points out, the proof of quantumness protocol can benefit from certain optimizations that are not known to be possible for Shor's algorithm. As an example, the group exponentiation circuits for the DDH-based proof of quantumness need not be reversible, thereby halving the depth and number of gates compared to Shor's algorithm (see section III.D in~\cite{poqbell} for more details). 

From existing knowledge assumptions one is only able to prove soundness against classical adversaries, i.e.~one is able to bound the success probability of a classical prover, but one cannot make a statement about the internal operations of a successful quantum prover. This is sufficient for a proof of quantumness, but other protocols, such as certifiable randomness generation~\cite{randomness}, remote state preparation~\cite{gheorghiu2019computationally}, or verification of quantum computations~\cite{mahadev2018classical}, require soundness against quantum adversaries, i.e.~one needs to have at least a partial characterization of any quantum adversary in the protocol.
As mentioned, one can derive single-round proofs of quantumness in the random oracle model, and these can be extended to single-round versions of quantum-sound protocols by proving security in the quantum random oracle model~\cite{boneh2011random,simplerpoq,alagic2020non,zhang2022classical}. However, with knowledge assumptions, it is unclear what the right quantum-sound analogue would be to derive these more advanced functionalities. As mentioned in the previous section, quantum knowledge assumptions have been considered in~\cite{liu2023another,zhandry2023quantum} to construct quantum money and quantum lightning schemes. These could serve as a useful starting point for constructing single-round quantum protocols for functionalities like single-device randomness expansion. We leave formalizing this for future work. 

\paragraph{Organization.} Our paper is organized as follows. We have collected commonly used notation as well as useful definitions from prior works in \cref{app:prelim}. In \cref{section:3test}, we provide the definition of an \eiiintcf{} and present the associated single-round protocol, comprised of three possible challenges. In \cref{section:ddhntcf}, we present a realization of an \eiiintcf{} based on the DDH assumption and a knowledge of exponent assumption. In \cref{section:2test} we provide the definition of an \eiintcf{}, a function family which consists of only two types of functions, rather than three, and present an associated single-round protocol with two challenges. Finally, \cref{section:lwentcf} contains a construction of an \eiintcf{} based on the LWE problem and a lattice knowledge assumption. This also implies the existence of an \eiiintcf{} based on LWE and the lattice knowledge assumption.

\paragraph{Acknowledgements.}
We thank Alex Lombardi, Urmila Mahadev, Greg Kahanamoku-Meyer, Umesh Vazirani, John Wright, and Tina Zhang for helpful discussions.
We are especially grateful to Vinod Vaikuntanathan to suggesting many of these ideas in the early stages of the project.
PA, VG and TM acknowledge support from the ETH Zurich Quantum Center and the Air Force Office for Scientific Research, grant No.~FA9550-19-1-0202.
TM is supported by an ETH Doc.Mobility Fellowship. AG is supported by the Knut and Alice Wallenberg Foundation through the Wallenberg Centre for Quantum Technology (WACQT).

\section{Notation and basic definitions} \label{app:prelim}

\subsection{Notation}

We denote the set of integers as $\mathbb{Z}$, and by $\mathbb{N}$ the set of natural numbers without zero. 
For any $q \in \mathbb{N}$ with $q\geq 2$ we let $\mathbb{Z}_q$ denote the ring of integers modulo $q$.
Vectors are assumed to be in column form and are written using bold lower-case letters, e.g., $\mathbf{x}$, and we denote the $i$-th component of $\mathbf{x}$ as $x_i$. 
Matrices are written as bold capital letters, e.g., \textbf{A}, and we denote the $i$-th entry in the $j$-th column by $A_{i,j}$ and the $i$-th row of \textbf{A} by $\textbf{a}_i$. 
We denote by $\rankedmats_d(\mathbb{Z}_q^{m\times n})$ the set of $m\times n$ matrices over $\mathbb{Z}_q$ of rank d, and by $\rank(\mathbf{A})$ the rank of the matrix $\mathbf{A}$.
Given a multiplicative group $\mathbb{G}$ of order $q$, an element $g \in \mathbb{G}$ and a vector $\mathbf{x} \in \mathbb{Z}_q^n$, we denote by $g^\mathbf{x}$ the vector $(g^{x_i})_i \in \mathbb{G}^n$, and for a matrix $\mathbf{A} \in \mathbb{Z}_q^{m\times n}$ we denote by $g^\mathbf{A}$ the matrix $(g^{A_{i,j}})_{i,j} \in \mathbb{G}^{m\times n}$.
We will consider lattices $\cL \subset \mathbb{Z}_q^m$, and for every $\mathbf{A} \in \mathbb{Z}_q^{m\times n}$ we define
$$
\cL(\mathbf{A}) \deq \mathbf{A}\cdot\mathbb{Z}_q^n = \{ \mathbf{Ax} \middlebar \mathbf{x} \in \mathbb{Z}_q^n \}\,.
$$
A \emph{negligible} function $f$, denoted by $f = \negl(\lambda)$, is a function $f : \mathbb{N} \to \mathbb{R}$ such that $f(\lambda) = o(\lambda^{-c})$ for every constant $c \in \mathbb{R}$.
For $\lambda \in \mathbb{N}$, we denote with $1^\lambda$ the string of $\lambda$ many $1$'s.

\paragraph{Distributions.}

The set of all distributions (which we identify with their density) over a set $X$ is denoted $\cD_X$.
The \emph{statistical distance} between two distributions $D_0$, $D_1$ over a countable domain $X$ is defined to be 
$$
    d(D_0,D_1) = \frac{1}{2}\sum_{x \in X}|D_0(x)-D_1(x)|\,.
$$
We let $H^2$ be the Hellinger distance, which is defined as follows for two densities $D_1$ and $D_2$ over the same finite domain $X$:
\begin{equation}
\label{eq:bhatt}
    H^2(D_0, D_1) = 1 - \sum_{x\in X} \sqrt{D_0(x) D_1(x)}\,.
\end{equation}
Two families of distributions $D_0 = \{D_{0,\lambda}\}_{\lambda\in\mathbb{N}}$, $D_1 = \{D_{1,\lambda}\}_{\lambda\in\mathbb{N}}$ on the same finite sets $\{X_\lambda\}_{\lambda\in\mathbb{N}}$ are called \emph{statistically indistinguishable} if $d(D_{0,\lambda},D_{1,\lambda}) = \negl(\lambda)$.
The support of a distribution $D \in \cD_X$ is defined as
$$
    \supp(D) = \{ x \in X \middlebar D(x) > 0 \} \,.
$$
\subsection{Computational Indistinguishability}

Throughout the paper, we will deal with efficient classical and quantum adversaries.
As usual, we abbreviate classical probabilistic polynomial time as PPT, and quantum polynomial time as QPT.

The notion of computational indistinguishability, which gives an equivalence relation between distributions as perceived by a PPT or QPT adversary, will be of importance in the soundess proofs of our protocols.
In our proofs, we only require computational indistinguishability with respect to PPT adversaries, so we use the following definition.
\begin{definition}[Computational Indistinguishability]
    \label{comp_ind}
    We say that two families of distributions $D_0 = \{D_{0,\lambda}\}_{\lambda\in\mathbb{N}}$, $D_1 = \{D_{1,\lambda}\}_{\lambda\in\mathbb{N}}$ on the same finite sets $\{X_\lambda\}_{\lambda\in\mathbb{N}}$ are \emph{computationally indistinguishable} if for every family of PPT algorithms (meaning, polynomial time in $\lambda$) $\mathcal{A} = \big\{A_\lambda : X_\lambda \rightarrow \{0,1\}\big\}_{\lambda\in\mathbb{N}}$ it holds that:
    \begin{equation*}
        \left| 
        \prss{x\leftarrow D_{0,\lambda}}{A_\lambda(x) = 0} -
        \prss{x\leftarrow D_{1,\lambda}}{A_\lambda(x)=0}
        \right| = \negl(\lambda)\,.
    \end{equation*}
\end{definition}

It follows immediately from the definition that statistically indistinguishable distributions are also computationally indistinguishable.

\subsection{Computational Hardness Assumptions}

The existence of the cryptographic primitives required for our single-round proofs of quantumness is based the DDH assumption (for the \eiiintcf{} in \cref{section:ddhntcf}) and the LWE assumption (for the \eiintcf{} in \cref{section:lwentcf}), which are defined as follows.

\begin{assumption}[Decisional Diffie-Hellman (DDH), see \cite{matrixdlin}]
\label{asmp:ddh} Let $\groupgen$ be a PPT algorithm that takes as input a security parameter $\lambda\in\N$ and outputs a tuple $(\mathbb{G},q,g)$, where $q$ is a $\lambda$-bit prime, $\mathbb{G}$ a cyclic group of order $q$ and $g$ a generator of $\mathbb{G}$. The \emph{decisional Diffie-Hellman (DDH) assumption} is that the ensembles
$$
        \big\{(\mathbb{G},q,g_1,g_2,g_1^r,g_2^r) \big\}_{\lambda\in\mathbb{N}} \text{ and } \big\{(\mathbb{G},q,g_1,g_2,g_1^{r_1},g_2^{r_2})\big\}_{\lambda\in\mathbb{N}}
    $$
    are computationally indistinguishable, where $(\mathbb{G},q,g) \leftarrow \groupgen(1^\lambda)$ and the elements $g_1,g_2 \in \mathbb{G}$ and $r,r_1,r_2\in \mathbb{Z_q}$ are chosen independently and uniformly at random.  
\end{assumption}

The DDH assumption implies the following \emph{matrix $d$-linear assumption} (\cref{asmp:matrixdlin}).
In other words, \cref{asmp:matrixdlin} is not an additional assumption, but rather a corollary of \cref{asmp:ddh}; we phrase it as a separate assumption only for convenience.

\begin{assumption}[Matrix $d$-Linear Assumption, from \cite{matrixdlin}]
\label{asmp:matrixdlin}
    Let $\groupgen$ be a PPT algorithm that takes as input a security parameter $\lambda\in\N$ and outputs a tuple $(\mathbb{G},q,g)$, where $q$ is a $\lambda$-bit prime, $\mathbb{G}$ a cyclic group of order $q$ and $g$ a generator of $\mathbb{G}$. 
    The \emph{matrix $d$-Linear assumption} is that for any $a,b,i,j \in \mathbb{N}$ such that $d \leq i < j \leq \min\{a,b\}$ the ensembles 
    $$
        \big\{(\mathbb{G},q,g,g^\bR) \bigmiddlebar \bR\in \rankedmats_i(\mathbb{Z}_q^{a\times b}) \big\}_{\lambda\in\mathbb{N}} \text{ and }
        \big\{(\mathbb{G},q,g,g^\bR) \bigmiddlebar \bR\in \rankedmats_j(\mathbb{Z}_q^{a\times b})\big\}_{\lambda\in\mathbb{N}}
    $$
    are computationally indistinguishable, where $(\mathbb{G},q,g) \leftarrow \groupgen(1^\lambda)$.
\end{assumption}

Indeed, the 1-linear assumption is precisely DDH, and for any $d\geq 1$, the $d$-linear assumption implies the $d+1$-linear assumption (Lemma 3 in \cite{dlin}).
As shown in Lemma A.1 from \cite{matrixdlin}, the $d$-linear assumption in turn implies the matrix $d$-linear assumption. 
Thus, the matrix $d$-linear assumption holds for any group generator $\groupgen$ for which we assume the DDH assumption to hold, a fact that we make use of in our proofs.

We will also use the decision variant of the Learning with Errors assumption~\cite{regev2009lattices}, which is defined as follows.
\begin{assumption}[LWE$_{n,q,\chi}$ Assumption, from \cite{regev2009lattices,randomness}]
\label{assump:lwe}
    Let $\lambda\in\N$ be a security parameter, let $n,m,q \in \mathbb{N}$ be integer functions of $\lambda$. 
    Let $\chi = \chi(\lambda)$ be a distribution over $\mathbb{Z}$. 
    The \emph{LWE$_{n,m,q,\chi}$ problem} is to distinguish between the distributions $(\mathbf{A}, (\mathbf{As}+\mathbf{e})\mod q)$ and $(\mathbf{A},\mathbf{u})$ where $\mathbf{A}\unifsample\mathbb{Z}_q^{n \times m},\, \mathbf{s} \unifsample\mathbb{Z}_q^n,\, \mathbf{e}\unifsample\chi^m$ and $\mathbf{u}\unifsample\mathbb{Z}_q^m$. 
    When $m$ grows at most polynomially with $n\log q$, we denote the problem as \emph{LWE}$_{n,q,\chi}$. 
    The \emph{LWE$_{n,q,\chi}$ assumption} states that no QPT procedure can solve the \emph{LWE}$_{n,q,\chi}$ problem with more than negligible advantage in $\lambda$, even when given access to a quantum-polynomial size advice state depending on the parameters $n,m,q$ and $\chi$ of the problem. 
\end{assumption}

\subsection{Knowledge Assumptions}

In addition to the above computational hardness assumptions, we use knowledge assumptions to show that the function family construction we provide fulfill the required extractability properties of the \eiiintcf{} and \eiintcf{} definition.

To show the existence of a DDH-based \eiiintcf, we additionally make the $t(\lambda)$-Knowledge-of-Exponents assumption (\tlkea{} for short) from~\cite{tkea}.\footnote{Note that in~\cite{tkea}, this assumption is phrased for a fixed value of $t$, but their proofs require $t$ to depend on $\lambda$. For clarity, we make this dependence explicit in the assumption.}
\begin{assumption}[\tlkea, based on \cite{tkea}] \label{assump:tkea}
    Let $t : \mathbb{N} \to \mathbb{N}$ be an integer function that grows at most polynomially, and let $\groupgen$ be a PPT algorithm that takes as input $\lambda \in \mathbb{N}$ and outputs a tuple $(\mathbb G, q, g)$, where $\mathbb{G}$ is a group of prime order $q \in \left[2^{\lambda-1},2^\lambda\right)$ and $g\in \mathbb{G}$ is a generator of $\mathbb{G}$.
    The \emph{$t(\lambda)$-Knowledge-of-Exponents assumption} is that for any PPT adversary $\cA$ with input in $\mathbb{G}^{t}\times\mathbb{G}^{t}\times\bits^{\rm poly(\lambda)}$ and output $(f,f') \in \mathbb{G}^2$, there exists a PPT extractor $\cA^*$ that takes the same input as $\cA$ (including the same random coins $r$) and outputs $\mathbf x \in \mathbb{Z}_q^{t(\lambda)}$ such that for any auxiliary input $z \in \bits^{\rm poly(\lambda)}$ and all random coins $r$,
    \begin{equation}
    \label{eq:tkea}
     \prss{
        (\mathbb{G}, q, g) \leftarrow \groupgen(1^\lambda) \\
        (\alpha, \mathbf{r}) \unifsample \mathbb{Z}_q \times \mathbb{Z}_q^{t(\lambda)}
     }{
        \begin{array}{cc}
            \left(f, f^{\prime}\right) \leftarrow \mathcal{A}\left(g^{\mathbf{r}}, g^{\alpha \mathbf{r}}, z,r\right) \\
            f^\alpha = f'
        \end{array}
        \land
        \begin{array}{cc}
            \mathbf{x} \leftarrow \mathcal{A}^*\left(g^{\mathbf{r}}, g^{\alpha \mathbf{r}}, z,r\right) \\
             g^{\langle\mathbf{x}, \mathbf{r}\rangle} \neq f
        \end{array}
    }
    = \negl(\lambda)\,,
    \end{equation}
    where $\langle \cdot, \cdot \rangle$ denotes the inner product on $\mathbb{Z}_q^{t(\lambda)}$.
\end{assumption}
The \tlkea{} essentially states that for any group-generating PPT algorithm $\groupgen$ for which the computation of discrete logarithms is assumed to be hard\footnote{For groups for which computing discrete logarithms can be done efficiently, there trivially exists an efficient extractor to find a suitable exponent and thus, the \tlkea{} holds trivially.}, the only way for a PPT adversary $\cA$ to generate a pair of elements $f,f' \in \mathbb{G}$ such that $f^\alpha = f'$ when given as input $g^{\mathbf r}, g^{\alpha \mathbf r}$ is to first pick a $\mathbf{x} \in \mathbb{Z}_q^{t(\lambda)}$ and to then output 
\begin{align*}
    f = (g^{r_1})^{x_1} \cdots (g^{r_t})^{x_t} = g^{\langle \mathbf x, \mathbf r \rangle} \tand
    f' = (g^{\alpha r_1})^{x_1} \cdots (g^{\alpha r_t})^{x_t} = f^\alpha \,.
\end{align*}
The extractor $\cA^*$ outputs precisely this $\mathbf x$, which formalises the idea that $\cA$ had to ``know'' $\mathbf x$ to generate its output.
Note that if the adversary $\cA$ outputs  $f = f' = \groupid$, where $\groupid$ denotes the identity element of $\mathbb G$, then the extractor still succeeds by outputting $\mathbf x = \mathbf 0$.

To show the existence of an LWE-based \eiintcf, we additionally make the LK-$\epsilon$ assumption from \cite{k_assumption}. This knowledge assumption states (informally) that any classical algorithm which can find a point ``suitably close'' to a lattice point must have known this corresponding lattice point in the first place.

\begin{assumption}[LK-$\epsilon$, from \cite{k_assumption}] \label{assumpt:lke}
Let $\epsilon$ be a fixed constant in the interval (0,1/2). Let $\algofont{Gen}_L$ denote an algorithm which on input of a security parameter $1^\lambda$ outputs a lattice $\cL$ given by a basis $\mathbf{A}$ of dimension $n = n(\lambda)$ and volume $\Delta = \Delta(\lambda)$. Denote by $\lambda_\infty(\cL)$ the $\infty$-norm of a shortest vector (w.r.t.~the $\infty$-norm) in a lattice $\cL$. Let $\mathcal{A}$ be an algorithm that takes a lattice basis $\mathbf{A}$ as input, has access to an oracle $\mathcal{O}$, and returns nothing. Let $\mathcal{A}^*$ denote an algorithm which takes as input a vector $\mathbf{y} \in \mathbb{R}^n$ and some state information, and returns another vector $\mathbf{p}\in \mathbb{R}^n$ and a new state. Consider the experiment in \cref{fig:LKA}. The \emph{LK-$\epsilon$} advantage of $\mathcal{A}$ relative to $\mathcal{A}^*$ is defined by 
\begin{equation*}
    \textup{Adv}^{\textup{LK-$\epsilon$}}_{\algofont{Gen}_\cL,\mathcal{A},\mathcal{A}^*}(\lambda) = \pr{ \textup{Exp}^{\textup{LK-$\epsilon$}}_{\algofont{Gen}_\cL,\mathcal{A},\mathcal{A}^*}(\lambda) = 1}\,.
\end{equation*}
We say that $\algofont{Gen}_\cL$ satisfies the LK-$\epsilon$ assumption, for a fixed $\epsilon$, if for every PPT algorithm $\mathcal{A}$ there exists a PPT algorithm $\mathcal{A}^*$ such that $\textup{Adv}^{\textup{LK-$\epsilon$}}_{\algofont{Gen}_\cL,\mathcal{A},\mathcal{A}^*}(\lambda) = \negl(\lambda)\,.$
\end{assumption}

\begin{algorithm}[ht]
\caption{Experiment Exp$_{\algofont{Gen}_\cL,\mathcal{A},\mathcal{A}^*}^{\text{LK-$\epsilon$}}(\lambda)$, from \cite{k_assumption}}
\label{fig:LKA}
\begin{myalgo}
    \State Sample $\mathbf{A} \leftarrow \algofont{Gen}_\cL(1^\lambda)$ and let $\cL$ be the lattice associated to the basis $\mathbf{A}$
    \State Sample $\mathrm{coins}\left[\mathcal{A}\right]$ (resp.~$\mathrm{coins}\left[\cA^*\right]$) from the random coin distribution of $\cA$ (resp.~$\cA^*$)
    \State Set $\mathrm{St} \leftarrow (\mathbf{A}, \mathrm{coins}[\mathcal{A}])$
    \State Run $\mathcal{A}^\mathcal{O}(\mathbf{A}, \mathrm{coins}[\mathcal{A}])$ until it halts, replying to the oracle queries $\mathcal{O}(\mathbf{y})$ as follows:
    \begin{itemize}
        \item $(\mathbf{p},\text{St}) \leftarrow \mathcal{A}^*(\mathbf{y}, \text{St},\mathrm{coins}[\cA^*])$
        \item If $\mathbf{y}$ is not within $\epsilon\cdot\lambda_\infty(\cL)$ of a point in $\cL$, return 0
        \item If $\mathbf{p} \notin \cL$, return 1
        \item If $\lVert \mathbf{p} - \mathbf{y} \rVert_\infty > \epsilon\cdot\lambda_\infty(\cL)$, return 1
        \item Return $\mathbf{p}$ to $\mathcal{A}$
    \end{itemize}
    \State Return 0    
\end{myalgo}
\end{algorithm}

\subsection{Noisy Trapdoor Claw-Free Function Family} 
\label{sec:ntcf}

The cryptographic primitive that the proof of quantumness protocol from \cite{randomness} relies on is a \emph{noisy trapdoor claw-free function family} (NTCF). In their work, they show how such a function family can be constructed from the LWE assumption. We show in \cref{section:ddhntcf} how an NTCF can also be realized from the DDH problem. 
Below, we provide a slightly relaxed version of the definition from~\cite{randomness} to accommodate that our DDH-based construction does not entirely fulfill the original NTCF definition, but stress that this relaxation does not impact the utility of this primitive for proofs of quantumness.

\begin{definition}[NTCF family, based on \cite{randomness}]\label{def:ntcf}
    Let $\lambda \in \mathbb{N}$ be a security parameter; all the following objects implicitly depend on $\lambda$.
    Let $\cX$ and $\cY$ be finite sets.
    Let $\mathcal{K}_{\cF}$ be a finite set of keys and $\cT_\cF$ a finite set of trapdoors.
    A family of functions 
    $$
        \mathcal{F} = \big\{f_{k,b} : \cX\rightarrow \mathcal{D}_{\cY} \big\}_{k\in \mathcal{K}_{\mathcal{F}},b\in\{0,1\}}
    $$
    is called a \emph{noisy trapdoor claw free (NTCF) family} if the following conditions hold.
    \begin{myenumi}
        \item \textbf{Efficient Function Generation.} \label{prop:gen}
        There exists a PPT algorithm $\genalg_\cF$ which generates a key $k\in \mathcal{K}_{\mathcal{F}}$ together with a trapdoor $t_k \in \mathcal{T}_{\mathcal{F}}$ on input of $1^\lambda$.
        \item \textbf{Trapdoor Injective Pair.} \label{prop:trapdoor_inj}
        There exist subsets $\cX_0, \cX_1 \subseteq \cX$ and a set $\cI_\cF\subset \cK_\cF\times\cT_\cF$ such that\footnote{The original definition in~\cite{randomness} ignored the detail that the inversion function in their construction only works for most keys, not all keys. Here, we include this technical detail by introducing a set $\cI_\cF$ of ``good'' key-trapdoor pairs on which the inversion function works, and require that most key-trapdoor pairs are good.}
        \begin{equation}
        \label{eq:ntcf bad keys}
            \prss{(k,t_k) \leftarrow\genalg_\cF(1^\lambda)}{(k,t_k)\notin \cI_\cF} = \negl(\lambda)\,,
        \end{equation}
        and for all pairs $(k,t_k) \in \cI_\cF$ the following conditions hold.
        \begin{myenumii}
            \item \label{item:trapdoor ntcf} \textit{Trapdoor}:
            there exists a poly-time deterministic algorithm $\invalg_\cF$ such that for all $b\in \{0,1\}$,  $x\in \cX$ and $y\in \supp(f_{k,b}(x))$, $\invalg_\cF(t_k,b,y) = x$.%
            \footnote{Note that this implies that for all $b\in\{0,1\}$ and $x\neq x' \in \cX$, $\supp(f_{k,b}(x))\cap \supp(f_{k,b}(x')) = \emptyset$.}
            \item \textit{Injective pair}: 
            consider the set $\cR_k$ of all tuples $(x_0,x_1)$ such that $f_{k,0}(x_0) = f_{k,1}(x_1)$. For all $x\in\cX_0$, there exists exactly one $x' \in \cX$ such that $(x,x') \in \cR_k$, and for all $x\in\cX_1$, there exists exactly one $x' \in \cX$ such that $(x',x) \in \cR_k$. Furthermore, there exists $\injconstant{\cF} \in \left[0,1\right]$ and $\lambda_\injconstant{\cF} \in \mathbb{N}$ such that for all $\lambda > \lambda_{\injconstant{\cF}}$, $|\cX_{0} \cap \cX_{1}|/|\cX| \geq \injconstant{\cF}$ and for any $(x_0,x_1) \in \cR_k$, $x_b \in \cX_{0} \cap \cX_{1}$ implies $x_{b\oplus1} \in \cX_{b\oplus1}$. Finally, membership in the set $\cR_k$ should be efficiently checkable.
        \end{myenumii}

        \item{\textbf{Efficient Range Superposition.} \label{prop:range_superpos}}
        For all keys $k\in \mathcal{K}_{\mathcal{F}}$ and $b\in \{0,1\}$ there exists a function $f'_{k,b}:\cX\to \mathcal{D}_{\cY}$ such that the following hold for all $b \in \bits$.
        \begin{myenumii}
            \item For all $(x_0,x_1)\in \mathcal{R}_k$, for all $(k,t_k) \in \cI_\cF$ and and $y\in \supp(f'_{k,b}(x_b))$, $\invalg_{\mathcal{F}}(t_k,b,y) = x_b$ and $\invalg_{\mathcal{F}}(t_k,b\oplus 1,y) = x_{b\oplus 1}$. 
            \item There exists a poly-time deterministic procedure $\chkalg_{\mathcal{F}}$ that, on input $k\in \cK_\cF$, $b$, $x\in \cX$ and $y\in \cY$, returns $1$ if  $y\in \supp(f'_{k,b}(x))$ and $0$ otherwise. 
            Note that $\chkalg_{\mathcal{F}}$ is not provided the trapdoor $t_k$. 
            \item \label{item:hellinger ntcf} It holds that
            \begin{equation}
            \label{eq:hellinger ntcf}
                \Ess{x \unifsample \cX}{ H^2\big(f_{k,b}(x), f'_{k,b}(x)\big) } = \negl(\lambda)\,,
            \end{equation}
            where $H^2$ is the Hellinger distance, see \cref{eq:bhatt}.
            Moreover, there exists a QPT procedure $\sampalg_{\mathcal{F}}$ that on input $k$ and $b\in\{0,1\}$ prepares the state
            $$
                \frac{1}{\sqrt{|\cX|}}\sum_{x\in \cX,y\in \cY}\sqrt{(f'_{k,b}(x))(y)}\ket{x}_\sX\ket{y}_\sY\,.
            $$
    
        \end{myenumii}

        \item{\textbf{Adaptive hardcore bit property (AHCB property).}} \label{item:ahcb}
        For all keys $k\in \mathcal{K}_{\mathcal{F}}$ the following conditions hold, for some integer $w$ that is a polynomially bounded function of $\lambda$. 
        \begin{myenumii}
            \item For all $b\in \{0,1\}$, $x\in \cX$, there exists a set $\dset_{k,b,x}\subseteq \{0,1\}^{w}$ such that
            $$
                \prss{d\unifsample \{0,1\}^w}{d\notin \dset_{k,b,x}} = \negl(\lambda)\,,
            $$
            and moreover there exists a poly-time deterministic algorithm that checks for membership in $\dset_{k,b,x}$ given $k,b,x$ and the trapdoor $t_k$. 
            \item There is a poly-time computable injection $\inj:\cX\to \{0,1\}^w$, such that $\inj$ can be inverted in poly-time on its range, and such that the following holds. 
            Letting
            \begin{align*} \label{eq:defsetsH}
                H_k &= 
                \big\{\big(b,\, x_b, d, d \cdot (\inj(x_0)\oplus \inj(x_1))\big) \bigmiddlebar
                b \in \{0,1\}, (x_0,x_1) \in \mathcal{R}_k, x_b \in \cX_b, d\in \dset_{k,0,x_0}\cap \dset_{k,1,x_1} \big\}\,,\text{\footnotemark} \\
                \overline{H}_k &= \big\{(b,x_b,d,c) \bigmiddlebar (b,x,d,c\oplus 1) \in H_k\big\}\,,
            \end{align*}
            it must hold that for any PPT procedure $\mathcal{A}$ that outputs a $4$-tuple $(b,x,d,c)$,
            \footnotetext{Note that although both $x_0$ and $x_1$ are referred to define the set $H_k$, only one of them, $x_b$, is explicitly specified in any $4$-tuple that lies in $H_k$.}
            \begin{equation*}
                \left|
                    \prss{(k,t_k)\leftarrow \genalg_\cF(1^{\lambda})}{\cA(k) \in H_k} - 
                    \prss{(k,t_k)\leftarrow \genalg_\cF(1^{\lambda})}{\cA(k) \in\overline{H}_k}
                \right| = \negl(\lambda)\,.
            \end{equation*}
        \end{myenumii}
    \end{myenumi}
\end{definition}

Compared to the NTCF definition of \cite{randomness}, the ``Injective Pair'' property is relaxed to only require that a fraction $\injconstant\cF$ of preimages $x \in \cX$ are part of a claw pair.
This is only relevant in the context of completeness, and $\injconstant\cF$ can be ensured to be high enough to obtain a sufficiently large completeness-soundness gap for the two constructions that we present here.
In fact, for the LWE construction from~\cite{randomness} $\injconstant\cF$ is 1, and for the DDH-based constructions $\injconstant\cF$ can be brought arbitrarily close to 1 at the cost of an increasing $\lambda_{\injconstant\cF}$.

The second relaxation we make is to only require the adaptive hardcore bit (AHCB) property to hold for probabilistic classical adversaries (whereas~\cite{randomness} required security against quantum adversaries, too). 
This is necessary for a DDH-based construction because DDH is not a quantum-hard problem. 
Since a proof of quantumness is only required to be hard for PPT adversaries, it suffices to require that the AHCB property also only holds for PPT adversaries.\footnote{We note that in \cite{randomness}, the authors also gave a protocol for generating certifiable randomness from a single device. For this, the soundness analysis also has to partially characterize QPT adversaries in the protocol, which requires that the AHCB property holds against such adversaries. For proofs of quantumness, neither~\cite{randomness} nor our work requires a quantum-sound AHCB property.}

For our single-round proof of quantumness protocol, we require the extension of $\cF$ by a \emph{trapdoor injective function family} as first defined in \cite{mahadev2018classical}. 
By itself, the trapdoor injective function family is similar to an NTCF in many regards, except that the function family branches $g_{k,0}$, $g_{k,1}$ do not have overlapping supports anymore, i.e., for each $y$ in the joint image of $\{g_{k,b}\}_{b\in \bits}$, $y$ has only one preimage $(b,x)$. 
Consequently, the inversion function does not need $b$ as input, but can recover $b$ by itself given an input $y$. 
There is also no adaptive hardcore bit property. 
We extend the original definition of a trapdoor injective function family to allow an extension to a larger domain $\cB_\cG\times\cX_\cG$, which will become useful for our \eiiintcf{}-based protocol. Note that while the check function also takes values in $\cB_\cG\times\cX_\cG$, we do not require the existence of an efficient inversion function for values $y$ with preimage in $\cB_\cG\times\cX_\cG\setminus \bits\times\cX$.

\begin{definition}[Trapdoor Injective Function Family, based on Definition 4.2 in \cite{mahadev2018classical}]
\label{def:trapdoor injective function family}
    Let $\lambda \in \mathbb{N}$ be a security parameter. 
    Let $\cB_\cG, \cX_\cG, \cY$ be finite sets such that $\bits \subseteq \cB_\cG$ and $\cX \subseteq \cX_\cG$. 
    Let $\mathcal{K}_\mathcal{G}$ be a finite set of keys, $\mathcal{T}_\mathcal{G}$ a finite set of trapdoors. 
    A family of functions
    \begin{equation*}
        \mathcal{G} = \{g_{k,b}:\cX_\cG\rightarrow \mathcal{D}_\cY\}_{b\in\cB_\cG, k\in\mathcal{K}_\mathcal{G}}
    \end{equation*}
    is called a \emph{trapdoor injective family} if the following conditions hold.
    \begin{myenumi}
        \item \textbf{Efficient Function Generation.} 
        There exists a PPT algorithm $\genalg_\cG$ that generates a key $k\in \mathcal{K}_{\mathcal{G}}$ together with a trapdoor $t_k \in \mathcal{T}_\mathcal{G}$.
        \item \textbf{Disjoint Trapdoor Injective Pair.} \label{item:disjoint trapdoor injective pair}
        There exists a set $\mathcal{I}_\mathcal{G}\subset \mathcal{K}_\mathcal{G}\times\mathcal{T}_\mathcal{G}$ such that
        \begin{align}
        \label{eq:trapdoor injective bad keys}
            \prss{(k,t_k) \leftarrow \genalg_\cG(1^\lambda)}{ (k,t_k)\notin \cI_\cG} = \negl(\lambda)\,,
        \end{align}
        and for all  $(k,t_k) \in \mathcal{I}_\mathcal{G}$, the following holds.
        \begin{myenumii}
            \item For all $b,b' \in \cB_\cG$ and $x,x' \in \cX_\cG$, if $(b,x) \neq (b',x')$, then $\supp(g_{k,b}(x))\cap \supp(g_{k,b'}(x')) = \emptyset$. 
            \item There exists a poly-time deterministic algorithm $\invalg_\cG$ with output in $\bits\times\cX$ such that for all $b \in \bits$, $x \in \cX$ and $y \in \supp(g_{k,b}(x))$, $\invalg_\cG(t_k,y) = (b,x)$.
        \end{myenumii}
        \item \textbf{Efficient Range Superposition.} \label{item:trapdoor injective function family, efficient range superposition}
        For all keys $k \in \mathcal{K}_\mathcal{G}$ and $b \in \cB_\cG$:
        \begin{myenumii}
            \item There exists an efficient deterministic procedure $\chkalg_\cG$ that, on input $k$, $y \in \cY$, $b \in \cB_\cG$, $x\in\cX_\cG$, outputs 1 if $y \in \supp(g_{k,b}(x))$ and 0 otherwise. Note that $\chkalg_\mathcal{G}$ is not provided the trapdoor $t_k$.
            \item There exists a QPT procedure $\sampalg_\cG$ that on input $k$, $b$, returns the state 
            \begin{equation*}
                \frac{1}{\sqrt{|\cX|}}\sum_{x\in \cX,y\in \cY}\sqrt{(g_{k,b}(x))(y)}\ket{x}_\sX \ket{y}_\sY\,.
            \end{equation*}
        \end{myenumii}
    \end{myenumi} 
\end{definition}

In the work of \cite{mahadev2018classical}, and also in our protocol, the map from a key $k \in \cK_{\cF} \cup \cK_{\cG}$ and an input $x$ to an image $y$ is the same for both $\cF$ and $\cG$. 
The difference lies in the use of different key distributions which enforce different properties on the resulting function pair $f_{k,b}$, such as one-to-oneness instead of two-to-oneness. 
A particularly useful instance of such a function family pair $(\cF,\cG)$ in the context of testing an adversary is when additionally, the key distributions of $\cF$, $\cG$ are computationally indistinguishable. 
Such a function family $\cF$ is then called \emph{injective invariant}.

\begin{definition}[Injective Invariance, based on Definition 4.3 in \cite{mahadev2018classical}] 
\label{def:injinv}
    A noisy trapdoor claw-free family $\mathcal{F}$ is \emph{injective invariant} if there exists a trapdoor injective family $\mathcal{G}$ such that:
    \begin{myenumi}
        \item The algorithms $\chkalg_\mathcal{F}$ and $\sampalg_\mathcal{F}$ are the same as the algorithms $\chkalg_\mathcal{G}$ and $\sampalg_\mathcal{G}$.
        \item \label{item:inj inv key indist}
        The family of marginal distributions over the keys, $\big\{ k \leftarrow \genalg_\cF(1^\lambda) \big\}_{\lambda\in\mathbb{N}}$ and $\big\{ k \leftarrow \genalg_\cG(1^\lambda) \}_{\lambda\in\mathbb{N}}$, are computationally indistinguishable (\cref{comp_ind}).
    \end{myenumi}
\end{definition}

\section{A Single-Round Proof of Quantumness based on Doubly Extended Extractable Trapdoor Claw-free Function Families}
\label{section:3test}

In this section, we present our first (and more general) single-round proof of quantumness. We begin by giving a formal definition of an \eiiintcf, which is the cryptographic primitive required for this test, in \cref{subsection:defe^3ntcf}. In \cref{subsection:3testpoq} we present the formal version of this protocol and prove its completeness and soundness. Later, in \cref{section:ddhntcf}, we will give an explicit construction of an \eiiintcf{} based on the DDH assumption and the $t(\lambda)$-KEA. 

\subsection{Doubly Extended Extractable Trapdoor Claw-free Function Family} 
\label{subsection:defe^3ntcf}
As explained in \cref{sec:intro_results}, an \eiiintcf{} consists of three functions families $(\cF, \cG, \cH)$.
The families $\cF$ and $\cG$ are the same as in the injective invariant NTCF family from~\cite{mahadev2018classical}, which we describe in detail in \cref{sec:ntcf}. 
We therefore focus on the third function family, $\cH$, which we call a \emph{weak extension} of a TCF family $\cF$ and which needs to satisfy a \emph{weak extractability property}.\footnote{The reason for why we refer to this property as ``weak'' extractability is to differentiate it from the extractability property of the \eiintcf{}, see \cref{def:extractability}, and to capture that extractability implies weak extractability, but the other way only works under additional conditions.}

\begin{definition}[Weak Extension of an NTCF]
\label{def:weakinv}
    Let $\lambda \in \mathbb{N}$ be a security parameter. Let $\cB_\cH, \cX_\cH, \cY$ be finite sets such that $\bits \subseteq \cB_\cH$ and $\cX \subseteq \cX_\cH$. Let $\mathcal{K}_\mathcal{H}$ be a finite set of keys, $\mathcal{T}_\mathcal{H}$ a finite set of trapdoors, and let $\cF$ be an NTCF (\cref{def:ntcf}). 
    A family of functions
    \begin{equation*}
        \mathcal{H} = \{h_{k,b}:\cX_\cH\rightarrow \mathcal{D}_\cY\}_{b\in\cB_\cH, k\in\mathcal{K}_\mathcal{H}}
    \end{equation*}
    is a \emph{weak extension of $\cF$} if the following holds.
    \begin{myenumi}
        \item \textbf{Efficient Function Generation.} There exists a PPT algorithm $\genalg_{\mathcal{H}}$ which generates a key $k\in \mathcal{K}_{\mathcal{H}}$ and trapdoor $t_k \in \mathcal{T}_\mathcal{H}$.
        \item \textbf{Efficient Range Superposition.} 
        \begin{myenumii}
            \item There exist an extension of the check function $\chkalg_\mathcal{F}$ associated to $\cF$ to inputs $b \in \cB_\cH$ and it holds that on input $k \in \cK_\cH$, $b \in \cB_\cH$, $x \in \cX_\cH, y \in \cY$ it outputs 1 if $y \in \supp(h_{k,b}(x))$ and 0 otherwise.\footnote{Intuitively, this statement just says that the check function ``works for $\cH$ as well''. Same for the $\sampalg_\mathcal{F}$ function.}
            \item The state preparation function $\sampalg_\mathcal{F}$ associated to $\cF$, on input $k \in \cK_\cH$ and $b \in \bits$, returns the state 
            \begin{equation*}
            \frac{1}{\sqrt{|\cX|}}\sum_{x\in \cX,y\in \cY}\sqrt{(h_{k,b}(x))(y)}\ket{x}_\sX \ket{y}_\sY\,.
            \end{equation*}
        \end{myenumii}
        \item \label{item:weak inv indkey}\textbf{Indistinguishability of keys.} 
        The families of marginal distributions over the keys, $\big\{k \leftarrow \genalg_\cH(1^\lambda)\}_{\lambda\in\mathbb{N}}$ and $\big\{ k \leftarrow \genalg_\cF(1^\lambda) \big\}_{\lambda\in\N}$, are computationally indistinguishable (\cref{comp_ind}).
    \end{myenumi}
\end{definition}

\begin{definition}[Weak Extractability Property]
\label{def:weak extractability}
    Let $\cF$ be an NTCF. 
    Let $\mathcal{H}$ be a weak extension of an NTCF. 
    We say that $\mathcal{H}$ \emph{satisfies the weak extractability property} if the following holds.
    \begin{myenumi}
        \item \textbf{Image Check.} \label{item:weak extractability image check}
        There exists a poly-time deterministic procedure $\imchkalg_\mathcal{H}$ that takes as input a trapdoor $t_k$ and an image $y \in \cY$ and outputs a bit, such that for all $k \in \cK_\cH$ and associated trapdoor $t_k$, $x \in \cX_\cH, b \in \cB_\cH, y \in \supp(h_{k,b}(x))$, it holds that $\imchkalg_\cH(t_k,y) = 1$.
        \item \textbf{Existence of Extractor.} 
        For every PPT algorithm $\cA$ that takes $k \in \mathcal{K}_\mathcal{H}$ as input and outputs a point $y \in \cY$, there exists a poly-time deterministic algorithm $\cA^*$ with inputs $k \in \mathcal{K}_\mathcal{H}$ and the random coins $r$ (sampled from a distribution $R$) of $\cA$ that outputs $b \in \cB_\cH$, $x\in \cX_\cH$ such that
        \begin{equation*}
        \prss{\kt \leftarrow \genalg_{\mathcal{H}}(1^\lambda), r \leftarrow R}{\chkalg_\cF(k,y,b,x) = 0 \land \imchkalg_\mathcal{H}(t_k,y) = 1} = \negl(\lambda)\,.
        \end{equation*} 
    \end{myenumi}
\end{definition}

In \cref{def:weak extractability}, the $\imchkalg_\cH$ procedure should be thought of as defining the set of images on which the extractor is required to work.
The first condition in the definition requires $\imchkalg_\cH$ to output $1$ whenever $y \in \supp(h_{k,b}(x))$, i.e.~when $y$ is actually in the range of $h_{k,b}$.
However, we also allow the $\imchkalg_\cH$ to output 1 on other $y$ as long as the extractor also works on such $y$.
In other words, Condition~\ref{item:weak extractability image check} in the definition specifies the minimum set on which the extractor is required to work, but the $\imchkalg_\cH$ procedure can output 1 on a larger set than just $\supp(h_{k,b})$ as long as the extractor works on this larger set, too, which is ensured by Condition 2.

The combination of an injective invariant NTCF $\cF$ with associated trapdoor injective function family $\cG)$ as in \cref{sec:ntcf} and a weak extension $\cH$ then forms an \eiiintcf.
\begin{definition}
\label{def:e^3ntcf}
A tuple of function families $(\cF, \cG, \cH)$ is called a \emph{doubly extended extractable noisy trapdoor claw-free family tuple} (\eiiintcf) if the following conditions holds:
\begin{myenumi}
\item $\cF$ is a noisy trapdoor claw free family with an adaptive hardcore bit (see \cref{sec:ntcf});
\item $\cF$ is injective invariant, with $\cG$ being a corresponding trapdoor injective family (see \cref{sec:ntcf});
\item $\cH$ is a weak extension of $\cF$ with  $\cB_\cH = \cB_\cG$ and $\cX_\cH = \cX_\cG $, and $\cH$ satisfies the weak extractability property.
\end{myenumi}
\end{definition}

\subsection{Single-Round Proof of Quantumness}
\label{subsection:3testpoq}

In the following, we present our first (and more general) protocol for a single-round proof of quantumness, which is based on the use of a \eiiintcf{} family introduced above. An explicit construction of an \eiiintcf{} based on the DDH assumption and the $t(\lambda)$-KEA is then given in \cref{section:ddhntcf}.

The protocol consists of three possible tests. 
Crucially, due to the computational indistinguishability of the function families in an \eiiintcf, a classical prover cannot tell which test is being performed\footnote{Note that while a quantum prover could in principle distinguish the three distributions for an \eiiintcf{} based on DDH, this is not required to break soundness. There exists a quantum prover that can pass all three tests with probability 1 with the same strategy, i.e. is not required to distinguish the tests.}.
Thus, a prover with a high probability of success has to use a strategy that will work well for all three tests, on average.

\begin{myprotocol}
\caption{Single-Round Proof of Quantumness based on \eiiintcf{}}
\label{poq:3test}
\begin{myalgo}   
    \Require Let $(\cF,\cG, \cH)$ be an \eiiintcf{} tuple and $\lambda \in \mathbb{N}$ a security parameter. 
    \State The verifier samples $a \leftarrow_U \{\mathtt{Eq},\mathtt{wIm},\mathtt{sIm}\}$. 
    \MyIf{$a = \mathtt{Eq}$}
        \MyState The verifier samples a key $(k,t_k) \leftarrow \genalg_{\mathcal{F}}(1^\lambda)$ and sends $k$ to the prover.
        \State The verifier receives tuple $(y, d, c)$, computes $x_b = \invalg_{\mathcal{F}}(t_k,b,y)$ for $b \in \bits$ using 
        \StateNoNumber the trapdoor, and checks
        \begin{itemize}
            \item if $\chkalg_\cF(k,y,0,x_0) = \chkalg_\cF(k,y,1,x_1) = 1\,,$
            \item if $(x_0,x_1) \in \cR_k\,,$
            \item if $x_0 \in \cX_0$ and $x_1 \in \cX_1\,,$
            \item if $d \in G_{k,0,x_0}\cap G_{k,1,x_1}\,,$
            \item if $(J(x_0) \oplus J(x_1))\cdot d = c\,.$
        \end{itemize} 
        \StateNoNumber If all of these conditions hold, the verifier accepts.
    \MyElsIf{$a = \mathtt{wIm}$}
        \MyState The verifier samples $(k,t_k) \leftarrow \genalg_{\mathcal{H}}(1^\lambda)$ and sends $k$ to the prover.
        \State The verifier receives tuple $(y, d, c)$ and checks whether $\imchkalg_\mathcal{H}(t_k,y) = 1$. If this holds,
        \StateNoNumber the verifier accepts. 
    \MyElsIf{$a = \mathtt{sIm}$}
        \MyState The verifier samples $(k,t_k) \leftarrow \genalg_{\mathcal{G}}(1^\lambda)$ and sends $k$ to the prover.
        \State The verifier receives tuple $(y, d, c)$, computes $(b,x) = \invalg_\mathcal{G}(t_k,y)$ using the trapdoor, and 
        \StateNoNumber checks whether $\chkalg_\mathcal{G}(k,y,b,x) = 1$. If this holds, the verifier accepts. 
    \EndIf
\end{myalgo}
\end{myprotocol}

The role of the equation test $\eqtest$ is to test the ability of the prover to find the value $c$ of the equation $d \cdot (J(x_0) \oplus J(x_1))$, where $d$ is chosen by the prover and $x_0,x_1$ are the preimages of the value $y$, which is also chosen by the prover. 
The role of the image tests $\wimtest$ and $\simtest$ is to constrain the strategies which a ``highly successful'' classical prover can perform in such a way that we can conclude that the prover must have ``known'' a preimage $(b,x) \in \bits\times\cX$ under $f_{k,b}$ to the image $y$ it returned.
Therefore, if a classical prover had a sufficiently high success probability, we could use that prover and the extractor from the \eiiintcf~to break the AHCB property of $f$.

We will now explain in more detail the role of the image tests $\simtest$ and $\wimtest$ and how they allow us to conclude that a (highly successful) prover $\cA$ must know a preimage. 
This notion of ``knowing a preimage'' is captured by the existence of an extractor that outputs a preimage under certain conditions.  

\paragraph{Weak Image Test:} The weak extractability property of $\cH$ implies the existence of an extractor $\cA^*$ such that whenever the image $y$ output by the prover $\cA$ fulfills certain conditions (captured by the image check function $\imchkalg_\cH$) the probability that the extractor fails to produce a valid preimage $(b,x)$ of $y$ under the functions $h_{k,b}$ (for fixed $k$, and $b$ in the set $\cB_\cH$) is negligible. 
Thus, the success probability of the prover in the weak image test $\wimtest$ allows us to lower-bound the probability that the extractor $\cA^*$ outputs a valid preimage to $y$ under $\genalg_\cH$ (see \cref{claim:successprobextr} in the proof of \cref{thm:poq2}). Since checking whether the extractor has returned a valid preimage is efficient, this probability is the same (up to a negligible difference) for all three key distributions. 
Note, however, that the preimage that the extractor outputs is a priori a preimage $(b,x)$ in the extended domain $\cB_\cH \times \cX_\cH$. 
To show a contradiction with the adaptive hardcore bit property, however, we need to show that the extractor produces a preimage in the restricted domain $\bits \times \cX$.
\paragraph{Strong Image Test:} This is where the strong image test $\simtest$ comes in. For the trapdoor injective function family $\cG$, the supports of the distributions $g_{k,b}(x)$ (for the extended set $\cB_\cH \times \cX_\cH$ of preimages) are all pairwise disjoint. 
Thus, under the function family $\cG$, the preimage $(b,x)$ of $y$ is unique even on the extended domain. The prover passes the strong image test $\simtest$ if and only if $y$ is in $g_{k,b}(\cX)$ for binary values of $b$, i.e., when the unique preimage of $y$ is in $\bits\times\cX$. In other words, the injectivity condition allows us to use the strong image test $\simtest$ to restrict the strategy of the prover to generating images of the form $f_{k,b}(x)$ for $(b,x) \in \bits\times\cX$ instead of, potentially, an extension of $f_{k,b}$ to the larger domain $\cB_\cH \times \cX_\cH$. By the computational indistinguishability of the key distributions of $\cG$ and $\cH$, it can be concluded that the ``probability of successful extraction'' for the prover-extractor pair $(\cA,\cA^*)$ is the same under the key distribution $\cG$ as under the key distribution of $\cH$.  This probability, together with the probability that the prover passes the strong image test, can be used to lower bound the probability that the extractor $\cA^*$ produces a preimage $(b,x)$ of $y$ with $b \in \bits$, $x \in \cX$, see \cref{claim:imtestbound}. By the computational indistinguishability of the key distributions of $\cG$ and $\cF$, it can be concluded that the ``probability of successful extraction'' for the ``restricted-domain-output'' extractor $\cE^*$ (which returns the output of $\cA^*$ if it lies in the restricted domain, else a random element) is at most negligibly different under the key distribution of $\cF$. 
\paragraph{Equation Test:} Finally, we can use the probability that $\cE^*$ successfully extracts a preimage together with the success probability of the prover in the equation test $\eqtest$ to lower-bound the probability that the prover-extractor pair $(\cA,\cE^*)$(which we combine to a single algorithm $\cB$) holds both a preimage and a valid equation, see \cref{claim:advsuccessbound}. We will then use this result to show that the existence of a classical prover that succeeds with a probability non-negligibly higher than $5/6$ would violate the adaptive hardcore bit property of the NTCF.

\begin{theorem}
\label{thm:poq2}
Let $\lambda \in \mathbb{N}$, 
and let $(\cF,\cG,\cH)$ be an \eiiintcf{} pair with paremeter $\injconstant\cF$.
We consider \protref{poq:3test} with inputs $\lambda$ and $(\cF, \cG, \cH)$.
\begin{myenumi}
    \item \textbf{Completeness.} There is a QPT prover which succeeds in the protocol with probability 
    $$
        \frac{2+\injconstant\cF}{3} - \negl(\lambda) \,.
    $$
    \item \textbf{Soundness.} Any PPT adversary succeeds in the protocol with probability at most
    $$
        \frac{5}{6} + \negl(\lambda).
    $$
\end{myenumi}
\end{theorem}

\begin{myproof}
~

\textbf{Completeness.}
Consider a quantum prover $P$ that, upon receiving the key $k$ from the verifier, proceeds as indicated in \cref{alg:successfulqpt}.
Note that all operations that $P$ applies are efficient, so $P$ is a QPT prover. \hfill
\begin{algorithm}
\caption{Successful Quantum Prover}
\label{alg:successfulqpt}
\begin{myalgo}
    \Require \eiiintcf{} tuple $(\cF,\cG,\cH)$, security parameter $\lambda \in \mathbb{N}$, and a key $k \in \cK_\cF \cup \cK_\cG \cup \cK_\cH$.
    \State The prover uses $\sampalg_\cF$ to prepare a quantum state, which is if $k\in\cK_\cF$
    \begin{equation}
    \label{eq:successful quantum prover state}
        \frac{1}{\sqrt{2|\cX|}} \sum_{\substack{ b \in \bits \\ x\in \cX,y\in \cY}} \sqrt{f'_{k,b}(x)(y)}\ket{b}_\sB \ket{x}_\sX \ket{y}_\sY.
    \end{equation}
    \State The prover measures the $\sY$-register in the computational basis and the $\sX$- and $\sB$-register in the Hadamard basis with outcomes $y,d,c$ respectively. 
    \State The prover outputs $(y,d,c)$. 
\end{myalgo}
\end{algorithm}

\begin{myenumii}
\item Case $a = \mathtt{sIm}$: 
Since $\mathcal{F}$ is injective invariant (\cref{def:injinv}), SAMP$_\mathcal{F} =$ SAMP$_\mathcal{G}$ and thus the state that $P$ prepares in \cref{eq:successful quantum prover state} is
$$
    \frac{1}{\sqrt{2|\cX|}} \sum_{\substack{b \in \bits \\ x\in \cX, y\in \cY}} \sqrt{(g_{k,b}(x))(y)}\ket{b}_\sB \ket{x}_\sX \ket{y}_\sY\,.
$$
Thus, for any $y$ that $P$ gets from the $\sY$ measurement, there exist $(b',x')\in \bits\times\cX$ such that $y \in \supp(g_{k,b'}(x'))$. 
Using Item~\ref{item:disjoint trapdoor injective pair} of \cref{def:trapdoor injective function family} of $\mathcal{G}$, it follows that if $(k,t_k) \in \mathcal{I}_G$, 
then the verifier will obtain $\invalg_\cG(t_k, y) = (b',x')$, and then $\chkalg_\cG(k,y,b',x') = 1$, so that $P$ succeeds in the protocol.
The failure probability of $P$ in the case $a = s\imtest$ is thus at most
$$
    \prss{\kt \leftarrow \genalg_\cG(1^\lambda)}{ (k,t_k)\notin \cI_\cG} = \negl(\lambda)\,.
$$

\item Case $a = \mathtt{wIm}$: It follows directly from the definition of $P$ and $\cH$ that the success probability of $P$ in the case $a = \mathtt{wIm}$ is 1.
\item Case $a = \mathtt{Eq}$: 
First, note that the statistics obtained on the state $\ket{\psi'}$ of \cref{eq:successful quantum prover state} and the same state but with the replacement $f'_{k,b} \to f_{k,b}$, denoted $\ket{\psi}$, differ at most by the trace norm $\lVert \ket{\psi'} - \ket{\psi} \lVert_{\mathrm{tr}}$.
Using Lemma 2.1 of \cite{randomness}, this trace norm is at most
$$
    \Ess{\substack{b \unifsample \bits, x \unifsample \cX}}{H^2\big(f_{k,b}(x),f'_{k,b}(x)\big)} = \negl(\lambda)\,,
$$
where we used \cref{eq:hellinger ntcf}.
On the state $\ket\psi$, the probability that $P$ obtains a value $y$ such that there exists a $(b,x) \in \bits\times(\cX_0 \cap \cX_1)$ such that $y \in \supp(f_{k,b}(x))$ is at least $\injconstant{\cF}$, using Item~\ref{prop:trapdoor_inj} of \cref{def:ntcf}. Together with the other statements of \cref{def:ntcf}, we can conclude that with probability at least $\injconstant{\cF}$ there exists $(x_0,x_1) \in \cR_k$ with $y \in \supp(f_{k,0}(x_0)) = \supp(f_{k,1}(x_1))$ and such that $x_0 \in \cX_0$, $x_1 \in \cX_1$ and there is no $(b',x') \in \bits\times\cX\setminus\{(0,x_0),(1,x_1)\}$ such that $y \in \supp(f_{k,b'}(x'))$.

It then follows from the same argument as for the case of the equation test in \cite{randomness} that the prover succeeds in obtaining a correct equation with probability $1 -\negl(\lambda)$, conditioned on the existence of $(b,x) \in \bits\times(\cX_0 \cap \cX_1)$ such that $y \in \supp(f_{k,b}(x))$.
Thus, the success probability of $P$ in the equation test is lower bounded by $\injconstant\cF - \negl(\lambda)$.
\end{myenumii}
Therefore, the total success probability of $P$ in the protocol is $\frac{2 + c_\cF}{3} - \negl(\lambda)$.

\textbf{Soundness.} Assume that there exists a PPT adversary $\mathcal{A}$ that succeeds with probability at least $\frac{5}{6} + \frac{1}{q(\lambda)}$ for some polynomial $q: \mathbb{N} \rightarrow \mathbb{R}_+$. 
We will use the weak extractability property of $\mathcal{H}$ together with the success probability of $\cA$ in the protocol to construct a PPT algorithm $\mathcal{B}$ which contradicts the AHCB property of $\mathcal{F}$.

Let the distribution of the random coins of $\cA$ be $R$. 
For ease of notation, we refrain from writing the coin distribution $r\leftarrow R$ of $\mathcal{A}$ explicitly in the proof, but note that all probabilities in the proof are defined on average over this random coin distribution.
Let $\mathcal{A}^*$ be the extractor corresponding to $\cA$ from the weak extractability property of $\mathcal{H}$ (\cref{def:weak extractability}).
Denote by $S_\mathcal{A}^{\hspace{1pt}a}$ the event that $\mathcal{A}$ produces an output which passes the test of case $a \in \{\mathtt{sIm},\mathtt{wIm},\mathtt{Eq}\}$. 
We begin by relating the probability of $\cA,\cA^*$ producing an image-preimage pair $y$, $(b,x)$ and the success probability of $\cA$ in the weak image test.
\begin{claim} 
\label{claim:successprobextr}
It holds that 
\begin{equation*}
 \prss{\kt \leftarrow \genalg_{\mathcal{H}}(1^\lambda)}{\chkalg_\cF(k,y,b,x) = 1}
 \geq
 \prss{\kt \leftarrow \genalg_{\mathcal{H}}(1^\lambda)}{S_\cA^{w\imtest}} - \negl(\lambda)\,.
\end{equation*}
where $y$ is obtained from $(y,d,c) = \mathcal{A}(k,r)$, $(b,x) = \mathcal{A}^*(k,r)$ and $r$ are the random coins of $\mathcal{A}\,$.
\end{claim}

\begin{claimproof}
Using that $\mathcal{A}^*$ is the extractor associated to $\mathcal{A}$, it follows by definition that
\begin{equation}
\label{eq:weakext}
    \prss{\kt \leftarrow \genalg_{\mathcal{H}}(1^\lambda)}{\chkalg_\mathcal{F}(k,y,b,x) = 0 \land S_\mathcal{A}^\mathtt{wIm}} = \text{negl}(\lambda)\,.
\end{equation}
Thus, we can derive that:
\begin{align*}
    &\prss{\kt \leftarrow \genalg_{\mathcal{H}}(1^\lambda)}{\chkalg_\cF(k,y,b,x) = 1}
    \geq \prss{\kt \leftarrow \genalg_{\mathcal{H}}(1^\lambda)}{\chkalg_\cF(k,y,b,x) = 1 \land S_\cA^{w\imtest}} \nonumber\\
    = &\prss{\kt \leftarrow \genalg_{\mathcal{H}}(1^\lambda)}{S_\cA^{w\imtest}} - \prss{\kt \leftarrow \genalg_{\mathcal{H}}(1^\lambda)}{\chkalg_\cF(k,y,b,x) = 0 \land S_\cA^{w\imtest}} \nonumber\\
    = &\prss{\kt \leftarrow \genalg_{\mathcal{H}}(1^\lambda)}{S_\cA^{w\imtest}} - \negl(\lambda)\,,
\end{align*}
where we used \cref{eq:weakext} in the last line.
\end{claimproof}

We now recall that $\chkalg_\cF$ does not use a trapdoor and is poly-time, that $\cA$, $\cA^*$ are poly-time, and thus the concatenation with $\chkalg_\mathcal{F}$ is poly-time.
It follows from Item~\ref{item:inj inv key indist} of \cref{def:injinv} (key distributions of $\cF$ and $\cG$ are computationally indistinguishable) and item \ref{item:weak inv indkey} of \cref{def:weakinv} (key distributions of $\cF$ and $\cH$ are computationally indistinguishable) that the key distributions of $\cG$, $\cH$ are computationally indistinguishable, and thus: 
\begin{multline*}
    \prss{\kt \leftarrow \genalg_{\mathcal{G}}(1^\lambda)}{\chkalg_\cF(k,y,b,x) = 1} 
    \geq \prss{\kt \leftarrow \genalg_{\mathcal{H}}(1^\lambda)}{\chkalg_\cF(k,y,b,x) = 1} - \negl(\lambda)\,,
\end{multline*}
and thus by \cref{claim:successprobextr},
\begin{equation}
\label{eq:bridge2}
    \prss{\kt \leftarrow \genalg_{\mathcal{G}}(1^\lambda)}{\chkalg_\cF(k,y,b,x) = 1} \geq \prss{\kt \leftarrow \genalg_{\mathcal{H}}(1^\lambda)}{S_\cA^{w\imtest}} - \text{negl}(\lambda)\,.
\end{equation}

Now, we show the existence of an extractor $\cE^*$ whose output has, loosely speaking, ``the right form'' to be used to construct an adversary $\cB$ to the adaptive hardcore bit property, and relate its probability of successful extraction to the success probabilities of $\cA$ in the two image tests.
\begin{claim}
\label{claim:imtestbound}
There exists a PPT extractor $\cE^*$ that takes as input a key $k$ and the random coins $r$ of $\cA$ and with output in $\bits\times\cX$ such that: 
\begin{multline*}
\prss{\kt \leftarrow \genalg_{\mathcal{H}}(1^\lambda)}{S_\mathcal{A}^\mathtt{wIm}} 
+ \prss{\kt \leftarrow \genalg_{\mathcal{G}}(1^\lambda)}{S_\mathcal{A}^\mathtt{sIm}} \\
\leq 1 + \prss{\kt \leftarrow \genalg_{\mathcal{F}}(1^\lambda)}{\chkalg_\mathcal{F}(k,y,b'',x'') = 1} + \negl(\lambda)\,,
\end{multline*}
where $y = \mathcal{A}(k,r)$, $(b'',x'') = \mathcal{E}^*(k,r)$ and $r$ are the random coins of $\mathcal{A}$.
\end{claim}
\begin{claimproof}
First, note that: 
\begin{align}
\label{eq:aux1}
&\prss{\kt \leftarrow \genalg_\cH(1^\lambda)}{ S_\cA^\wimtest } 
+ \prss{\kt \leftarrow \genalg_\cG(1^\lambda)}{ S_\cA^\simtest } \nonumber\\
\annotatesign{Eq. \eqref{eq:bridge2}}{\leq}
&\prss{\kt \leftarrow \genalg_\cG(1^\lambda)}{\chkalg_\cF(k,y,b,x) = 1} 
+ \prss{\kt \leftarrow \genalg_\cG(1^\lambda)}{S_\cA^\simtest} + \negl(\lambda) \nonumber\\
\leq & \hspace{10pt} 1 + \prss{\kt \leftarrow \genalg_\cG(1^\lambda)}{\chkalg_\cF(k,y,b,x) = 1 \land S_\cA^\simtest} + \negl(\lambda) \nonumber\\
\annotatesign{Def. $S_\cA^\simtest$}{\leq}&
\hspace{10pt} 1 + \prss{\kt \leftarrow \genalg_\cG(1^\lambda)}{\chkalg_\cF(k,y,b,x) = 1 \land \begin{gathered}
    \chkalg_\cF(k,y,b',x') = 1 \\
    (b',x') = \invalg_\cG(t_k,y) 
\end{gathered}
} + \negl(\lambda)\,
\end{align}
Now, define $\cE$ to be the algorithm which takes as input $(b,x) \in \cB_\cG \times \cX_\cH$ and outputs $(b,x)$ whenever $(b,x) \in \bits\times\cX$ and else\footnote{This choice is arbitrary, we just need to replace $b$ by a bit and $x$ by some element in $\cX$.} $(0,z)$ for some arbitrary $z \in \cX$ and define $\cE^* = \cE \circ \cA^*$. 
In the following calculation, we denote by $(b'',x'')$ the output of $\cE^*$ upon input $(k,r)$. 
Note that whenever $(k,t_k) \in \cI_\cG$, we have that $\chkalg_\mathcal{F}(k,y,b,x) = 1 \land \chkalg_\mathcal{F}(k,y,b',x') = 1$ implies that $(b',x') = (b,x)$ due to Item~\ref{item:disjoint trapdoor injective pair} of \cref{def:trapdoor injective function family}, and using that $\chkalg_\cF = \chkalg_\cG$. 
Furthermore, since $(b',x') \in \bits\times\cX$ (recall that the inversion function $\invalg_\cG$ only returns values of $b',x'$ in $\bits\times\cX$ by definition) it holds in particular that $\chkalg_\mathcal{F}(k,y,b,x) = 1 \land \chkalg_\mathcal{F}(k,y,b',x') = 1$  
 implies
 $\chkalg_\mathcal{F}(k,y,b'',x'') = 1$.
Using this, it follows that: 
\begin{align}
\label{eq:aux2}
    &\prss{\kt \leftarrow \genalg_{\mathcal{G}}(1^\lambda)}{\chkalg_\mathcal{F}(k,y,b,x) = 1 \land \chkalg_\mathcal{F}(k,y,b',x') = 1} \nonumber\\
    \annotatesign{Eq. \eqref{eq:trapdoor injective bad keys}}{\leq} &\prss{\kt \leftarrow \genalg_{\mathcal{G}}(1^\lambda)}{\chkalg_\mathcal{F}(k,y,b,x) = 1 \land \chkalg_\mathcal{F}(k,y,b',x') = 1 \land (k,t_k) \in \cI_\cG} + \negl(\lambda) \nonumber\\
    \leq &\prss{\kt \leftarrow \genalg_{\mathcal{G}}(1^\lambda)}{\chkalg_\mathcal{F}(k,y,b'',x'') = 1 \land (k,t_k) \in \cI_\cG} + \negl(\lambda) \nonumber\\
    \leq &\prss{\kt \leftarrow \genalg_{\mathcal{G}}(1^\lambda)}{\chkalg_\mathcal{F}(k,y,b'',x'') = 1} + \negl(\lambda)\,.
\end{align}
Combining \cref{eq:aux1,eq:aux2} and using that $\cA$, $\cE^*$ are PPT algorithms that do not use the trapdoor, the computational indistinguishability of the keys distributed by $\genalg_\cF$ and $\genalg_\cG$ (see \cref{def:injinv}) yields the claim.
\end{claimproof}

\begin{claim} 
\label{claim:advsuccessbound}
There exists a PPT algorithm $\cB$ that takes a key $k\in \cK_\cF$ as input and produces outputs in $\bits\times\cX\times\bits^w\times\bits$ such that
\begin{multline*}
\prss{\kt \leftarrow \genalg_\cH(1^\lambda)}{S_\mathcal{A}^\mathtt{wIm}} + \prss{\kt \leftarrow \genalg_{\mathcal{G}}(1^\lambda)}{S_\mathcal{A}^\mathtt{sIm}} + \prss{\kt \leftarrow \genalg_{\mathcal{F}}(1^\lambda)}{S_\mathcal{A}^\mathtt{Eq}} \\
\leq 2 + \prss{\kt \leftarrow \genalg_{\mathcal{F}}(1^\lambda)}{ \mathcal{B}(k) \in H_k} + \negl(\lambda)\,,
\end{multline*}
where $H_k$ refers to the set introduced in Item~\ref{item:ahcb} of \cref{def:ntcf}.
\end{claim}
\begin{claimproof}
\Cref{claim:imtestbound} implies that
\begin{multline}
\label{eq:aux4}
\prss{\kt \leftarrow \genalg_{\mathcal{H}}(1^\lambda)}{S_\mathcal{A}^\mathtt{wIm}} + \prss{\kt \leftarrow \genalg_{\mathcal{G}}(1^\lambda)}{S_\mathcal{A}^\mathtt{sIm}} + \prss{\kt \leftarrow \genalg_{\mathcal{F}}(1^\lambda)}{S_\mathcal{A}^\mathtt{Eq}} \\
\leq 1 + \prss{\kt \leftarrow \genalg_{\mathcal{F}}(1^\lambda)}{\chkalg_\mathcal{F}(k,y,b'',x'') = 1} + \prss{\kt \leftarrow \genalg_{\mathcal{F}}(1^\lambda)}{S_\mathcal{A}^\mathtt{Eq}} + \negl(\lambda)\,.
\end{multline}
Furthermore, it holds that
\begin{multline}
\prss{\kt \leftarrow \genalg_{\mathcal{F}}(1^\lambda)}{\chkalg_\mathcal{F}(k,y,b'',x'') = 1} + \prss{\kt \leftarrow \genalg_{\mathcal{F}}(1^\lambda)}{S_\mathcal{A}^\mathtt{Eq}} \\
\leq  1 + \prss{\kt \leftarrow \genalg_{\mathcal{F}}(1^\lambda)}{
    \begin{gathered}
        \chkalg_\mathcal{F}(k,y,b'',x'') = 1 \\
        \land \; S_\mathcal{A}^\mathtt{Eq} \land (k,t_k) \in \mathcal{I}_\mathcal{F}
    \end{gathered}
} + \negl(\lambda)\,,
\end{multline}
where in the last equation we used \cref{eq:ntcf bad keys}. 
We now note that if $(k,t_k) \in \mathcal{I}_\mathcal{F}$, we have that
$\chkalg_\cF(k,y,0,x_0) = \chkalg_\cF(k,y,1,x_1) = 1$. If furthermore $\chkalg_\cF(k,y,b'',x'') = 1$, then $x = x_{b''}$ where we defined $x_{\tilde b} \coloneqq \invalg_\mathcal{F}(t_k,\tilde b,y)$, because the inversion function returns the ``correct'' preimages $x_0,x_1$ under these conditions.
Using this and the definition of $S_\mathcal{A}^\mathtt{Eq}$, we can conclude that: 
\begin{align}
&\prss{\kt \leftarrow \genalg_{\mathcal{F}}(1^\lambda)}{\chkalg_\mathcal{F}(k,y,b'',x'') = 1 \land S_\mathcal{A}^\mathtt{Eq} \land (k,t_k) \in \mathcal{I}_\mathcal{F}} \nonumber\\
\leq &\prss{\kt \leftarrow \genalg_{\mathcal{F}}(1^\lambda)}{ 
    \begin{gathered}
        c = d\cdot(J(x_0)\oplus J(x_1)) \land x = x_{b''} \land\, (x_0,x_1) \in \mathcal{R}_k \\[3pt]
         \land\, d\in \dset_{k,0,x_0}\cap \dset_{k,1,x_1} \land x_0\in\cX_0 \land x_1\in\cX_1
    \end{gathered}
} \nonumber\\
\leq &\prss{\kt \leftarrow \genalg_{\mathcal{F}}(1^\lambda)}{ 
    \begin{gathered}
        c = d\cdot(J(x_0)\oplus J(x_1)) \land x = x_{b''} \land\, (x_0,x_1) \in \mathcal{R}_k \\[3pt]
         \land\, d\in \dset_{k,0,x_0}\cap \dset_{k,1,x_1} \land x_{b''}\in\cX_{b''}
    \end{gathered}
}\;,
\end{align}
where $(y,d,c) = \cA(k,r)$, $(b'',x'') = \cE^*(k,r)$, and $x_{\tilde b} = \invalg_\cF(t_k,\tilde b,y)$.
We now define an algorithm $\mathcal{B}$ which uses the algorithms $\mathcal{A}, \mathcal{E}^*$ as subroutine: see \cref{algo:soundness b}.
\begin{algorithm}[ht]
\caption{Adversary $\cB$ breaking the AHCB property of $\cF$}
\label{algo:soundness b}
\begin{myalgo}
    \Require Key $k \in \cK_\cF$
    \State Sample $r \leftarrow R$
    \State Run $\mathcal{A}$ on the input $(k,r)$, obtaining output $(y,d,c)$
    \State Run $\mathcal{E}^*$ on the input $(k,r)$, obtaining output $(b'',x'')$
    \State Return $(b'',x'',d,c)$
\end{myalgo}
\end{algorithm}
Thus, by definition of $\mathcal{B}$ and the set $H_k$ (see Item \ref{item:ahcb} in \cref{def:ntcf}), it follows that
\begin{multline}
\label{eq:successprobB}
\prss{\kt \leftarrow \genalg_{\mathcal{F}}(1^\lambda)}{ 
    \begin{gathered}
        c = d\cdot(J(x_0)\oplus J(x_1)) \land x = x_{b''} \\[3pt]
        \land\, (x_0,x_1) \in \mathcal{R}_k \land d\in \dset_{k,0,x_0}\cap \dset_{k,1,x_1} \\[3pt]
        \land\, x_{b''}\in\cX_{b''}
    \end{gathered}
} 
 = \prss{\kt \leftarrow \genalg_{\mathcal{F}}(1^\lambda)}{\mathcal{B}(k) \in H_k}\,.
\end{multline}
Combining \crefrange{eq:aux4}{eq:successprobB} concludes the proof of the claim.
\end{claimproof}
Using \cref{claim:advsuccessbound} and the assumption on the success probability of $\cA$ being at least $5/6 + 1/q(\lambda)$, we can conclude that there exists a PPT algorithm $\cB$ such that 
\begin{equation}
    \prss{\kt \leftarrow \genalg_{\mathcal{F}}(1^\lambda)}{\mathcal{B}(k) \in H_k}
    \geq \frac{1}{2} + \frac{3}{q(\lambda)} - \negl(\lambda)\,.
\end{equation}
This constitutes a contradiction with the AHCB property (Item~\ref{item:ahcb} of \cref{def:ntcf}).
Thus, the success probability of $\mathcal{A}$ must be upper bounded by $\frac{5}{6} + \text{negl}(\lambda)$. \qedhere
\end{myproof}

\section{A Doubly Extended Extractable Trapdoor Claw-free Function Family from DDH and KEA}
\label{section:ddhntcf}

In the following, we present an explicit construction for an e$^3$NTCF  family (see \cref{def:e^3ntcf} in \cref{subsection:defe^3ntcf}), which is the cryptographic primitive required for our first single-round proof of quantumness (described in \cref{subsection:3testpoq}). We show that this function family tuple, which we denote as ($\cF_{\rm DDH}, \cG_{\rm DDH},\cH_{\rm DDH} $), fulfills the requirements of an e$^3$NTCF under the DDH assumption (\cref{asmp:ddh}) and the \tlkea~(\cref{assump:tkea}). 

First we introduce the construction of the NTCF $\mathcal{F}_{\rm DDH}$, which is closely based on the DDH-based TCF construction presented in \cite{poqbell} and prove that this function family fulfills the definition. In particular, we prove that $\mathcal{F}_{\rm DDH}$ fulfills the \emph{adaptive hard-core bit} property, which has previously only been shown for TCFs based on LWE~\cite{randomness} and (non-standard) hardness of isogenies~\cite{alamati2022candidate}. Then, we construct the corresponding trapdoor injective function family $\mathcal{G}_{\rm DDH}$ and weak extension $\mathcal{H}_{\rm DDH}$. Finally, we show how we can derive the weak extractability property of $\mathcal{H}_{\rm DDH}$ from the \tlkea.

\subsection{Noisy Trapdoor Claw-Free Family}

We use the DDH-based TCF construction from~\cite{poqbell}, which we briefly summarize.
Let $\lambda\in\N$ be a security parameter. Let $q$ be a $\lambda$-bit prime, $n = 121 \cdot\lceil \log(q) \rceil$ and $d$ an integer such that $d = \Theta(n^2)$. We define the NTCF family $\cF_{\rm DDH}$ as follows\footnote{Note that in the NTCF definition from~\cite{randomness} (see \cref{sec:ntcf}), the functions output probability distributions, i.e.~are of the form $f_{k,b}: \cX \rightarrow \cD_\cY$ (where $\cD_{\cY}$ is the set of probability distributions on a finite set $\cY$).
In our DDH based construction, the output of $f_{k,b}$ will be point distributions, so for simplicity we write $f_{k,b}: \cX \to \cY$.}.

Let $\cX = \mathbb{Z}_d^n$ and $\cY = \mathbb{G}^{n+1}$ for some multiplicative group $\mathbb{G}$. Given a key $k = (\mathbf{G},\mathbf{g})$, where $\mathbf{G} \in \mathbb{G}^{(n+1)\times n}$ and $\mathbf{g} \in \mathbb{G}^{n+1}$, we define the functions $f_{k,b}: \cX \rightarrow \cY$ as
\begin{align*}
    f_{k,0}(\mathbf{x})= \left(\textstyle\prod_j G_{i,j}^{x_j}\right)_i\,, \hspace{20pt}
    f_{k,1}(\mathbf{x})= \left(g_i\cdot\textstyle\prod_j G_{i,j}^{x_j}\right)_i\,.
\end{align*}
For a key of the form $k = (g^\mathbf{A},g^\mathbf{As})$, this is equal to
\begin{align*}
    f_{k,0}(\mathbf{x}) = g^\mathbf{Ax}\,, \hspace{20pt}
    f_{k,1}(\mathbf{x}) = g^\mathbf{A(x+s)}\,.
\end{align*}
Furthermore, we define $\chkalg_{\cF_{\rm DDH}}(k,b,\mathbf{x},\mathbf{y}) \coloneqq \delta_{\mathbf{y},f_{k,b}(\mathbf{x})}\,.$
Let $\groupgen$ be a PPT procedure that, on input $1^\lambda$, outputs an $\lambda$-bit prime $q$, a multiplicative group $\mathbb{G}$ of order $q$ and a generator $g$ of $\mathbb{G}$. We define the key generation and inversion algorithms as in \cref{alg:Gen_F_DDH,alg:inv_F}.
This construction is the same as the DDH-based TCF in \cite{poqbell}, except that for simplicity we use a uniformly distributed key matrix $\mathbf{A}$ (which has full rank with overwhelming probability), instead of specifying a procedure for obtaining a full-rank matrix. The set $\cI_{\cF_{\rm DDH}}$ would thus be given by the key-trapdoor pairs for which $\mathbf{A}$ is a full-rank matrix. 

Furthermore, the sets $\cX_b$ can be chosen as the subsets of $\cX = \Z_d^n$ such that for all $\mathbf{x} \in \cX_b$, all entries of $\mathbf{x}$ are in the range of $1-b$ to $d-b$. This ensures in particular that for all $\mathbf{x} \in \cX_b$, $\mathbf{s} \in \bits^n$, $\mathbf{x} -(-1)^b\mathbf{s} \in \cX$ and for all $\mathbf{x} \in \cX_0 \cap \cX_1$, $\mathbf{x} -(-1)^b\mathbf{s} \in \cX_{b\oplus1}$. Since $|\cX_0 \cap \cX_1|/|\cX| = (1-2/d)^n \rightarrow 1$, $c_\cF$ can be chosen to be 0.99.
\begin{algorithm}[ht]
\caption{Key Generation Procedure $\genalg_{\cF_{\rm DDH}}$}
\label{alg:Gen_F_DDH}
\begin{myalgo}
    \Require Security parameter $\lambda\in\N$.
    \State Sample $(\mathbb{G},q,g) \leftarrow \groupgen(1^\lambda)$.
    \State Sample $\mathbf{A} \leftarrow_U \mathbb{Z}_q^{(n+1)\times n}$, $\mathbf{s} \unifsample \bits^n$.
    \State Compute $g^\mathbf{A},g^\mathbf{As}$ via element-wise exponentiation.
    \State Return key $k = (g^\mathbf{A},g^\mathbf{As})$ and trapdoor $t_k = (g,\mathbf{A},\mathbf{s}).$
\end{myalgo}
\end{algorithm}
\begin{algorithm}[ht]
\caption{Inversion Algorithm $\invalg_{\cF_{\rm DDH}}$}
\label{alg:inv_F}
\begin{myalgo}
\Require Trapdoor data $t_k = (g, \mathbf{A}, \mathbf{s})$, a bit $b \in \bits$ and a value $\mathbf{y} \in \cY$.
    \Ensure $\mathbf{x}_b \in \cX$ such that $\mathbf{y} = g^{\mathbf{A}(\mathbf{x}_b+b\mathbf{s})}$ and  or message ``$\mathtt{error}$''.
    \State Compute pseudo-inverse $\mathbf{A}^{-1}$ using $\mathbf{A}$. If $\mathbf{A}$ is not full-rank, return $\mathtt{error}$.
    \State Compute $\mathbf{z} = \left(g^{-bs_i}\prod_j y_j^{A_{i j}^{-1}}\right)_{i=1}^n.$
    \State Try to find $\mathbf{x} \in \cX$ such that $\mathbf{z} = g^{\mathbf{x}}$, by brute force. If this procedure fails, return $\mathtt{error}$, else return $\mathbf{x}$. (Note that this can be done efficiently since $d = \textup{poly}(n)$.) 
\end{myalgo}
\end{algorithm}

It was shown in~\cite{poqbell} that this construction is a TCF, i.e.~that it is hard to find collisions.
However, it was not proven that this construction also satisfies the (much stronger) AHCB property.
We show this fact using a lossy sampling argument similar to the one used in the AHCB proof in~\cite{randomness}.
As the proof is somewhat technical, we defer it to \cref{app:ddh_ntcf}.
\begin{lemma} \label{lem:ddh_ahcb}
    $\cF_{\rm DDH}$ is a noisy trapdoor claw-free function family under the assumption that $\groupgen$ fulfills the DDH assumption.
\end{lemma}
\subsection{Trapdoor Injective Family}

To obtain an injective invariant family, we set $\cB_{\cG_{\rm DDH}} = \mathbb{Z}_q$ and $\cX_{\cG_{\rm DDH}} = \Z_q^n$.
For $k \in \mathbb{G}^{(n+1)\times(n+1)}$ and $b \in \cB_{\cG_{\rm DDH}}$, we define
\begin{equation*}
    g_{k,b}(\mathbf{x}) = \left(g_i^b\cdot\textstyle\prod_j G_{i,j}^{x_j}\right)_i\,,
\end{equation*}
which, for a key of the form $k = (g^\mathbf{A},g^\mathbf{u})$ ( where $\mathbf{A}\in \mathbb{Z}_q^{(n+1)\times n}$ and $\mathbf{u} \in \Z_q^{n+1}$), is equal to
\begin{equation*}
    g_{k,b}(\mathbf{x}) = g^{\mathbf{Ax}+b\mathbf{u}}\,.
\end{equation*}
$\chkalg_{\cF_{\rm DDH}}$ can be naturally extended to keys $k \in \cK_{\cG_{\rm DDH}}$ and values $b \in \cB_{\cG_{\rm DDH}}, x \in \cX_{\cG_{\rm DDH}}$, and we set $\chkalg_{\cG_{\rm DDH}} = \chkalg_{\cF_{\rm DDH}}$. 

The key generation and inversion algorithms of $\cG_{\rm DDH}$ are defined in \cref{alg:Gen_G_DDH} and \cref{alg:inv_G}, respectively.
\begin{algorithm}[ht]
\caption{Key Generation Procedure $\genalg_{\cG_{\rm DDH}}$}
\label{alg:Gen_G_DDH}
\begin{myalgo}
    \Require Security parameter $\lambda\in\N$.
    \State  Sample $(\mathbb{G},q,g) \leftarrow \groupgen(1^\lambda)$.
    \State Sample $\mathbf{\Tilde{A}} \leftarrow_U \mathbb{Z}_q^{(n+1)\times (n+1)}$.
    \State Compute $g^\mathbf{\Tilde{A}}$ via element-wise exponentiation.
    \State Return key $k = g^\mathbf{\Tilde{A}}$ and trapdoor $t_k = (g,\mathbf{\Tilde{A}}).$
\end{myalgo}
\end{algorithm}
\begin{algorithm}
\caption{Inversion Algorithm $\invalg_{\cG_{\rm DDH}}$}
\label{alg:inv_G}
\begin{myalgo}
    \Require Trapdoor data $t_k = (g, \mathbf{\Tilde{A}})$ and a value $\mathbf{y} \in \cY = \mathbb{G}^{n+1}$.
    \Ensure Tuple $(\mathbf{x},b) \in \cX \times \bits $ such that $\mathbf{y} = g^{\mathbf{\Tilde{A}}\binom{\mathbf{x}}{b}}$ or error message ``$\mathtt{error}$''.
    \State  Compute $\mathbf{\Tilde{A}}^{-1}$ using $\mathbf{\Tilde{A}}$. If $\mathbf{\Tilde{A}}$ is not full-rank, return $\mathtt{error}$.
    \State Compute $\mathbf{z} = \left(\prod_j y_j^{\Tilde{A}^{-1}_{i j}}\right)_{i=1}^{n+1}$.
    \State Try to find $(\mathbf{x},b) \in \cX\times\bits$ such that $\mathbf{z} = g^{\binom{\mathbf{x}}{b}}$, by brute force.
    If this procedure fails, return $\mathtt{error}$, else return ($\mathbf{x},b$).
    (Note that this can be done efficiently since $d = $poly$(n)$.)
\end{myalgo}
\end{algorithm}
\begin{lemma}  
\label{lemma:tif2}
 $\mathcal{G}_{\rm DDH}$ is a trapdoor injective family.
\end{lemma}
\begin{myproof}
~
\begin{myenumi}
    \item \textbf{Efficient Function Generation.} This condition holds since all steps of the key generation procedure $\genalg_{\cG_{\rm DDH}}$ are efficient. 
    \item \textbf{Disjoint Trapdoor Injective Pair.} Let $\mathcal{I}_{\cG_{\rm DDH}}$ be the set of key-trapdoor pairs such that $\mathbf{\Tilde{A}}$ is invertible. By \cref{lemma:rudm}, 
    $$
    \text{Pr}_{\kt \leftarrow \genalg_{\mathcal{G}_{\rm DDH}}(1^\lambda)}\left[(k,t_k)\notin\mathcal{I}_{\cG_{\rm DDH}}\right]= \text{negl}(\lambda)\;.$$ 
    First, note that all steps of $\invalg_{\cG_{\rm DDH}}$ are efficient. Furthermore, for all $(k,t_k) \in \mathcal{I}_{\cG_{\rm DDH}}$, it is easy to check that for all $b \in \bits$, $\mathbf{x} \in \cX$ and $\mathbf{y} \in \supp(g_{k,b}(\mathbf{x}))$ it holds that $\invalg_{\cG_{\rm DDH}}(t_k,\mathbf{y}) = (\mathbf{x},b)$. The condition of disjoint supports of $g_{k,b}(\mathbf{x})$ for all $(\mathbf{x},b) \in \cX_\cG \times \cB_\cG$ follows from the fact that $\mathbf{\Tilde{A}}$ is invertible.
    \item \textbf{Efficient Range Superposition.} This is evident by the definition of $\chkalg_{\cG_{\rm DDH}} = \chkalg_{\cF_{\rm DDH}}$ and the fact that the functions $f_{k,b}$ are efficiently computable. \qedhere
\end{myenumi} 
\end{myproof}

To show that $\cF_{\rm DDH}$ is an injective invariant family with associated trapdoor injective family $\cG_{\rm DDH}$, it remains to show that the key distributions of $\cF_{\rm DDH}$, $\cG_{\rm DDH}$ are computationally indistinguishable. This is a direct consequence of the indistinguishability of the distributions $D_0$ and $D_2$ from Lemma \ref{lemma:comp_ind2} for the case $a = n+1, b = n$. 
Thus, we can conclude:
\begin{lemma}
    $\cF_{\rm DDH}$ is an injective invariant family (with associated trapdoor injective family $\cG_{\rm DDH}$) under the assumption that $\groupgen$ fulfills the DDH assumption. 
\end{lemma}

\subsection{Weak Extension}
\label{section:ddh_h}

In this section, we present a construction for a weak extension $\cH_{\rm DDH}$ of $\cF_{\rm DDH}$, which we can show also fulfills the weak extractability property under the $(n(\lambda)+1)$-KEA.
The key generation algorithm of $\cH_{\rm DDH}$ is defined in \cref{alg:Gen_H_DDH}.
\begin{algorithm}[ht]
\caption{Key Generation Procedure $\genalg_{\cH_{\rm DDH}}$}
\label{alg:Gen_H_DDH}
\begin{myalgo}
    \Require Security parameter $\lambda\in\N$.
    \State Sample $(\mathbb{G},q,g) \leftarrow \groupgen(1^\lambda)$.
    \State Sample $\mathbf{v} \leftarrow_U \mathbb{Z}_q^{n+1}, \mathbf{u} \leftarrow_U \mathbb{Z}_q^n$ and compute $\mathbf{\Tilde{A}}= \scalebox{0.8}{$\begin{pmatrix}1 \\ \mathbf{u}\end{pmatrix}$}\mathbf{v}^T$.
    \State Compute $g^\mathbf{\Tilde{A}}$ via element-wise exponentiation.
    \State Return key $k = g^\mathbf{\Tilde{A}}$ and trapdoor $t_k = (g,\mathbf{u}).$
\end{myalgo}
\end{algorithm}
To obtain a weak extension, we set $\cB_{\cH_{\rm DDH}} = \cB_{\cG_{\rm DDH}} = \mathbb{Z}_q$ and $\cX_{\cH_{\rm DDH}} = \cX_{\cG_{\rm DDH}} = \mathbb{Z}_q^n$ and for all $k \in \mathbb{G}^{(n+1)\times(n+1)}$ and $b \in \cB_{\cH_{\rm DDH}}$ define $h_{k,b}$ as
\begin{equation*}
    h_{k,b}(\mathbf{x}) = \left(g_i^b\cdot\textstyle\prod_j G_{i,j}^{x_j}\right)_i \,,
\end{equation*} 
which, for a key of the form $k = (g^\mathbf{A},g^\mathbf{u})$, is equal to
\begin{equation*}
    h_{k,b}(\mathbf{x}) = g^{\mathbf{Ax}+b\mathbf{u}} \,.
\end{equation*}
$\chkalg_{\cF_{\rm DDH}}$ can be naturally extended to keys $k \in \cK_{\cH_{\rm DDH}}$ and values $(b,x) \in \cB_{\cH_{\rm DDH}}\times\cX_{\cH_{\rm DDH}}$ via $\chkalg_{\cF_{\rm DDH}}(k,b,\mathbf{x},\mathbf{y})\coloneqq \delta_{\mathbf{y},h_{k,b}(\mathbf{x})}$.
Furthermore, we define the image check function as in \cref{alg:ImCHK_H_DDH}. 
\begin{algorithm}[ht]
\caption{Image Check Algorithm $\imchkalg_{\cH_{\rm DDH}}$}
\label{alg:ImCHK_H_DDH}
\begin{myalgo}
\Require  Trapdoor data $t_k = (g,\mathbf{u})$ and a value $\mathbf{y} \in \cY = \mathbb{G}^{n+1}$.
\Ensure $c \in \bits$.
\State  Check if $\mathbf{y}$ is of the form $\mathbf{y} = \scalebox{0.65}{$\begin{pmatrix}y_1\\y_1^\mathbf{u}\end{pmatrix}$}$.
\end{myalgo}
\end{algorithm}

We show this definition of $\cH_{\rm DDH}$ is a weak extension of $\cF_{\rm DDH}$.
\begin{lemma}  
\label{lemma:wif}
 $\mathcal{H}_{\rm DDH}$ is a weak extension of  $\mathcal{F}_{\rm DDH}$ (as defined in \cref{def:weakinv}) under the assumption that $\groupgen$ fulfills the DDH assumption.
\end{lemma}
\begin{myproof}
~
\begin{myenumi}
    \item \textbf{Efficient Function Generation.} This condition holds since all steps of the key generation procedure $\genalg_{\cH_{\rm DDH}}$ are efficient. 
    \item \textbf{Efficient Range Superposition.} This is evident from the definition of $h_{k,b}$ and the definition of the functions $\chkalg_{\mathcal{F}_{\rm DDH}}$ and SAMP$_{\mathcal{F}_{\rm DDH}}$.
    \item \textbf{Image Check.} All steps of $\imchkalg_{\cH_{\rm DDH}}$ are efficient. Furthermore one can easily check that for all $(k, t_k) \in \cK_{\cH_{\rm DDH}}\times\cT_{\cH_{\rm DDH}}, \mathbf{x} \in \cX_{\cH_{\rm DDH}}, b \in \cB_{\cH_{\rm DDH}}, \imchkalg_{\cH_{\rm DDH}}(t_k,h_{k,b}(\mathbf{x})) = 1$.
    \item \textbf{Indistinguishability of Keys.} This follows from the indistinguishability of the distributions $D_0$ and $D_3$ in \cref{lemma:comp_ind2}. \qedhere
\end{myenumi} 
\end{myproof}

\subsection{Weak Extractability Property from KEA Assumption}

Using the knowledge of exponent assumption, we can show that $\cH_{\rm DDH}$ satisfies the weak extractability property from~\cref{def:weak extractability}.

\begin{lemma}[Weak Extractability]
\label{lem:ddh_weak_extract}
$\mathcal{H}_{\rm DDH}$ fulfills the weak extractability property under the assumption that $\groupgen$ fulfills the DDH assumption and the $(n(\lambda)+1)$-KEA.
\end{lemma}

\begin{proof}
Let $\mathcal{A}$ be a PPT algorithm that takes as input a key $k \in \mathbb{G}^{(n+1)\times(n+1)}$ and random coins $r$ from some distribution $R$, and outputs a point $y \in \cY = \mathbb{G}^{n+1}$. Let  $\mathcal{A}^*$ be the corresponding extractor from the $(n(\lambda)+1)$-KEA. 

First, we need to show that correlations between the auxiliary input and the first two rows of the keys do not impact the correctness of the extractor, i.e.~that:

\begin{claim}[Extraction under Auxiliary Information]
\label{claim:aux} 
Assuming that $\groupgen$ fulfills the DDH assumption, it holds for all random coins $r$ of $\cA$ that: 
\begin{align*}
 \underset{ (\mathbb{G},q, g) \leftarrow \groupgen(1^\lambda) \atop (\mathbf{u}, \mathbf{v}) \stackrel{U}{\longleftarrow} \mathbb{Z}_q^n \times \mathbb{Z}_q^{(n+1)}}{\operatorname{Pr}}
\left[\begin{array}{cc}
\left(f, \mathbf{f}\,\right) \leftarrow \mathcal{A}\left(g^{\mathbf{v}}, g^{u_1 \mathbf{v}}, ...,g^{u_n \mathbf{v}}, r\right) \\
f_1=f^{u_1} \end{array}
\land
\begin{array}{cc}
\mathbf{x} \leftarrow \mathcal{A}^*\left(g^{\mathbf{v}}, g^{u_1 \mathbf{v}}, ...,g^{u_n \mathbf{v}}, r\right) \\
 g^{\langle\mathbf{x}, \mathbf{v}\rangle} \neq f
\end{array}\right]
\leq \operatorname{negl}(\lambda) \,.
\end{align*}
\end{claim}

\begin{myproof}
We will show the result by assuming that the converse holds and constructing an adversary which contradicts \cref{lemma:comp_ind2}. Note that for clarity in presentation, we sometimes denote matrix inputs as a tuple of the matrix rows/submatrices.  
Assume there exists a polynomial $p$ such that 
\begin{multline}
\label{eq:poly}
 \underset{ (\mathbb{G},q, g) \leftarrow \groupgen(1^\lambda) \atop (\mathbf{u}, \mathbf{v}) \stackrel{U}{\longleftarrow} \mathbb{Z}_q^n \times \mathbb{Z}_q^{(n+1)}}{\operatorname{Pr}}
\left[\begin{array}{cc}
\left(f, \mathbf{f}\,\right) \leftarrow \mathcal{A}\left(g^{\mathbf{v}}, g^{u_1 \mathbf{v}}, ...,g^{u_n \mathbf{v}}, r\right) \\
f_1=f^{u_1} \end{array}
\begin{array}{cc}
\land
\end{array}
\begin{array}{cc}
\mathbf{x} \leftarrow \mathcal{A}^*\left(g^{\mathbf{v}}, g^{u_1 \mathbf{v}}, ...,g^{u_n \mathbf{v}}, r\right) \\
 g^{\langle\mathbf{x}, \mathbf{v}\rangle} \neq f
\end{array}\right] > 1/p(\lambda) \,.
\end{multline}

Consider the algorithm $\cB$ defined in \cref{alg:adv1}.
\begin{algorithm}[ht]
\caption{Adversary $\mathcal{B}$}
\label{alg:adv1}
\begin{myalgo}
    \Require $k = g^{\mathbf{M}}$ for $\mathbf{M} \in \mathbb{Z}_q^{n\times(n+1)}$
    \State Sample $r$ from the random coin distribution of $\cA$
    \State  Sample $\alpha \in \mathbb{Z}_q$ uniformly at random, and set $k' \coloneqq (k_1,k_1^\alpha, k_2,...,k_n)^T$ where $k_i$ denote the rows of $k$. 
    \State Run $\mathcal{A}$ on the input $(k',r)$, obtaining output $(f, f_1,..,f_n)$.
    \State Run $\mathcal{A}^*$ on the input $(k',r)$, obtaining output $\mathbf{x}$.
    \State Check if $f_1 = f^\alpha$ and $g^{\langle (\mathbf{M}_{1,j})_j , \mathbf{x}\rangle} \neq f$. If yes, return 0. Else return 1.
\end{myalgo}
\end{algorithm}

First, note that 
\begin{align}
&\prss{(\mathbb{G},q, g) \leftarrow \groupgen(1^\lambda)\\ (\mathbf{u}, \mathbf{v}) \unifsample \mathbb{Z}_q^{n-1} \times \mathbb{Z}_q^{n+1},\ R = (1,\mathbf{u})\mathbf{v}^T}{\mathcal{B}(g^{\mathbf{R}}) = 0} = 
\prss{ (\mathbb{G},q, g) \leftarrow \groupgen(1^\lambda) \nonumber\\ (\mathbf{u}, \mathbf{v}) \unifsample \mathbb{Z}_q^{n-1} \times \mathbb{Z}_q^{n+1}}{\mathcal{B}(g^{\mathbf{v}}, g^{\mathbf{uv}^T}) = 0} \\
\annotatesign{Def. $\cB$}{=}& \underset{\substack{ (\mathbb{G},q, g) \leftarrow \groupgen(1^\lambda) \\ (\alpha, \mathbf{u}, \mathbf{v}) \unifsample \mathbb{Z}_q \times \mathbb{Z}_q^{n-1} \times \mathbb{Z}_q^{(n+1)}
}}
{\operatorname{Pr}}
\left[
    \begin{gathered}
        \begin{array}{cc}
        \left(f, \mathbf{f}\,\right) \leftarrow \mathcal{A}\left(g^{\mathbf{v}}, g^{\alpha \mathbf{v}}, g^{\mathbf{uv}^T}, r\right) \\
        f_1=f^{\alpha} 
        \end{array}
        \land
        \begin{array}{cc}
        \mathbf{x} \leftarrow \mathcal{A}^*\left(g^{\mathbf{v}}, g^{\alpha \mathbf{v}}, g^{\mathbf{uv}^T}, r\right) \\
         g^{\langle\mathbf{x}, \mathbf{v}\rangle} \neq f 
        \end{array}
    \end{gathered}
\right]
\nonumber\\
 \annotatesign{Eq. \eqref{eq:poly}}{>}&\hspace{15pt} 1/p(\lambda)\,.
\end{align}
Furthermore, notice that
\begin{multline*}
\underset{(\mathbb{G},q, g) \leftarrow \groupgen(1^\lambda) \atop \mathbf{M} \leftarrow_U \mathbb{Z}_q^{n\times (n+1)}}{\text{Pr}}\left[\mathcal{B}(g^\mathbf{M} ) = 0\right] = \underset{(\mathbb{G},q, g) \leftarrow \groupgen(1^\lambda) \atop \mathbf{v} \leftarrow_U \mathbb{Z}_q^{n+1},\,\mathbf{M}' \leftarrow_U \mathbb{Z}_q^{(n-1)\times (n+1)}}{\text{Pr}}\left[\mathcal{B}(g^\mathbf{v},g^\mathbf{M'}) = 0\right] \\
 = \underset{ (\mathbb{G},q, g) \leftarrow \groupgen(1^\lambda) \atop (\alpha, \mathbf{v}, \mathbf{M}') \stackrel{U}{\longleftarrow} \mathbb{Z}_q \times \mathbb{Z}_q^{n+1} \times \mathbb{Z}_q^{(n-1)\times (n+1)}}{\operatorname{Pr}}
    \left[\begin{array}{cc}
    \left(f, \mathbf{f}\,\right) \leftarrow \mathcal{A}\left(g^{\mathbf{v}}, g^{\alpha \mathbf{v}}, g^{\mathbf{M}'}, r\right) \\
    f_1=f^{\alpha} \end{array}
    \begin{array}{cc}
    \land
    \end{array}
    \begin{array}{cc}
    \mathbf{x} \leftarrow \mathcal{A}^*\left(g^{\mathbf{v}}, g^{\alpha \mathbf{v}}, g^{\mathbf{M}'}, r\right) \\
     g^{\langle\mathbf{x}, \mathbf{v}\rangle} \neq f 
\end{array}\right] \\
 \overset{(\ref{eq:tkea})}{=} \text{negl}(\lambda)
\end{multline*}
where the last equation follows from the definition of the KEA and the fact that the auxiliary input on the left hand side is independent from the first two inputs.

Thus, $\cB$ can distinguish between the ensembles $\{(G,g,g^{(1,\mathbf{u})\mathbf{v}^T})| \mathbf{u} \leftarrow_U \mathbb{Z}_q^{n-1}, \mathbf{v} \leftarrow_U \mathbb{Z}_q^{n+1} \}$ and $\{(G,g,g^\mathbf{M}) | \mathbf{M} \leftarrow_U \mathbb{Z}_q^{n\times (n+1)}\}$, a contradiction to the indistinguishability of the distributions $D_2$, $D_3$ from
\cref{lemma:comp_ind2} in the case of $a = b = n$.
\end{myproof}

It follows from \cref{claim:aux} that:

\begin{equation}
\begin{split}
 \operatorname{negl}(\lambda) = 
 \underset{ (\mathbb{G},q, g) \leftarrow \groupgen(1^\lambda) \atop (\mathbf{u}, \mathbf{v}) \stackrel{U}{\longleftarrow} \mathbb{Z}_p^n \times \mathbb{Z}_p^{(n+1)}}{\operatorname{Pr}}
\left[\begin{array}{cc}
\left(f, \mathbf{f}\,\right) \leftarrow \mathcal{A}\left(g^{\mathbf{(1,u)v}^T}, r\right) \\
f_1=f^{u_1} \end{array}
\begin{array}{cc}
\land
\end{array}
\begin{array}{cc}
\mathbf{x} \leftarrow \mathcal{A}^*\left(g^{\mathbf{(1,u)v}^T}, r\right) \\
 g^{\langle\mathbf{x}, \mathbf{v}\rangle} \neq f
\end{array}\right]  \\
\geq \underset{ (\mathbb{G},q, g) \leftarrow \groupgen(1^\lambda) \atop (\mathbf{u}, \mathbf{v}) \stackrel{U}{\longleftarrow} \mathbb{Z}_p^n \times \mathbb{Z}_p^{(n+1)}}{\operatorname{Pr}}
\left[\begin{array}{cc}
\left(f, \mathbf{f}\,\right) \leftarrow \mathcal{A}\left(g^{\mathbf{(1,u)v}^T}, r\right) \\
 \forall i \in \{1,...,n\}: f_i=f^{u_i} \end{array}
\begin{array}{cc}
\land
\end{array}
\begin{array}{cc}
\mathbf{x} \leftarrow \mathcal{A}^*\left(g^{\mathbf{(1,u)v}^T}, r\right) \\
g^{\langle\mathbf{x}, \mathbf{v}\rangle} \neq f
\end{array}\right] \\
= \underset{ (\mathbb{G},q, g) \leftarrow \groupgen(1^\lambda) \atop (\mathbf{u}, \mathbf{v}) \stackrel{U}{\longleftarrow} \mathbb{Z}_p^n \times \mathbb{Z}_p^{(n+1)}}{\operatorname{Pr}}
\left[\begin{array}{cc}
\left(f, \mathbf{f}\,\right) \leftarrow \mathcal{A}\left(g^{\mathbf{(1,u)v}^T}, r\right) \\
 \forall i \in \{1,...,n\}: f_i=f^{u_i} \end{array}
\begin{array}{cc}
\land
\end{array}
\begin{array}{cc}
\mathbf{x} \leftarrow \mathcal{A}^*\left(g^{\mathbf{(1,u)v}^T}, r\right) \\
 g^{\langle\mathbf{x}, \mathbf{v}\rangle\cdot(1,u)} \neq \left(f, \mathbf{f}\,\right)
\end{array}\right] \\
= \underset{\kt \leftarrow \genalg_{\mathcal{H}_{\rm DDH}}(1^\lambda)}{\operatorname{Pr}}
\left[\begin{array}{cc}
\mathbf{y} \leftarrow \mathcal{A}\left(k, r\right) \\
 \imchkalg_{\mathcal{H}_{\rm DDH}}(t_k,\mathbf{y}) = 1 \end{array}
\begin{array}{cc}
\land
\end{array}
\begin{array}{cc}
(\mathbf{x},b) \leftarrow \mathcal{A}^*\left(k, r\right) \\
 \chkalg_{\mathcal{H}_{\rm DDH}}(k,\mathbf{y},b,\mathbf{x}) = 0
\end{array}\right] \\
\end{split}
\end{equation}
where the last step uses the definitions of $\genalg_{\cH_{\rm DDH}}$, $\chkalg_{\cH_{\rm DDH}}$ and $\imchkalg_{\cH_{\rm DDH}}$, and we relabel $\mathbf{x} \in \mathbb{Z}_q^{n+1}$ as $(\mathbf{x},b) \in \mathbb{Z}_q^n \times \mathbb{Z}_q$. 

Since this holds for all random coins $r$ of $\cA$, the lemma follows. \qed
\end{proof}

Thus, we can conclude: 

\begin{theorem}[Existence of an \eiiintcf{}]
\label{thm:e^3NTCF} The family tuple ($\mathcal{F}_{\rm DDH},\mathcal{G}_{\rm DDH},\mathcal{H}_{\rm DDH}$) is an \eiiintcf{} with completeness parameter $c_\cF = 0.99$ under the assumption that $\groupgen$ fulfills both the DDH assumption and the ($n(\lambda)+1$)-KEA.  
\end{theorem}

Using this \eiiintcf~in \protref{poq:3test}, we get a single-round proof of quantumness based on the DDH and $t$-KEA assumptions:
\begin{corollary}[Existence of Single-Round Proof of Quantumness based on DDH and KEA]
\label{corr:3testsingleroundpoq} Assuming that a group sampling procedure $\groupgen$ fulfills both the DDH assumption and the ($n(\lambda)+1$)-KEA, there exists a single-round proof of quantumness with completeness 0.99 and soundness $\frac{5}{6}$.  
\end{corollary}

\section{A Single-Round Proof of Quantumness based on Extended Extractable Trapdoor Claw-free Function Families}
\label{section:2test}

In \cref{section:ddhntcf}, we presented an explicit construction for an \eiiintcf{} (the cryptographic primitive underpinning our first single-round protocol presented in \cref{section:3test}) based on the DDH and \tlkea{} assumptions. Naturally, the question arises whether this primitive can be realized from other cryptographic assumptions, too. As it turns out, such a function family tuple can also be realized from the LWE assumption and a lattice knowledge assumption (the LK-$\epsilon$ assumption, see \cref{assumpt:lke}).

Before providing this construction, we first note the following: in the special case of an \eiiintcf{} (\cref{def:e^3ntcf}) for which $\cG = \cH$, and the image test is given by concatenating the inversion and check function, the weak and strong image tests in \protref{poq:3test} are the same.
Therefore, these two tests can be combined into one, resulting in a simplified protocol. 

Motivated by this simpler protocol, we define \eiintcf~families, which are essentially the same as \eiiintcf, except that the roles of the injective and extractable function families are combined.
To make this part of the paper more self-contained, we give a separate definition in \cref{def:e^2ntcf} as well as the simplified proof of quantumness protocol in \protref{poq:2test}.
Note that an \eiintcf~is a stronger cryptographic primitive than an \eiiintcf, i.e.~every \eiintcf~is also an \eiiintcf, but not vice versa.

Later, in \cref{section:lwentcf} we will show how an \eiintcf~can be realized from the LWE and  LK-$\epsilon$ assumption using a slightly modified variant of the LWE-based NTCF of \cite{randomness} and extending it by a suitable trapdoor injective function family $\cG$.

\subsection{Extended Extractable NTCF Families (\eiintcf)}
\label{subsection:defe^2ntcf}

An injective invariant NTCF ~\cite{mahadev2018classical} is a pair $(\cF,\cG)$ such that $\cF$ is an NTCF that is computationally indistinguishable from the associated trapdoor injective function family $\cG$ (see \cref{sec:ntcf}). An \eiintcf~is the same, except we require an additional extractability property. This is a property that the trapdoor injective function family $\cG$ should satisfy. Intuitively, this property is satisfied if the only way to produce an image $y \in \cY$ of the functions from $\cG$ is by first sampling a preimage $x \in \cX, b \in \bits$.
More precisely, the condition of \cref{eq:extractability property} states that the probability that $\cA$ produces a valid image $y$ and the extractor $\cA^*$ fails to produce the correct preimage $(b,x)$ is negligible.

\begin{definition}[Extractability Property]
\label{def:extractability}
    Let $\mathcal{G}$ be a trapdoor injective function family.
    We say that $\cG$ \emph{satisfies the extractability property} if for every PPT procedure $\cA$ that takes $k \in \cK_\cG$ as input and outputs a point $y \in \cY$, there exists a PPT extractor $\cA^*$, with input $k \in \mathcal{K}_\mathcal{G}$, $y \in \cY$ and the random coins $r$ (sampled from a distribution $R$) of $\cA$, that outputs $b \in \bits$, $x \in \cX$, such that
    \begin{equation} \label{eq:extractability property}
        \prss{\substack{
            \kt \leftarrow \genalg_\cG(1^\lambda) \\ 
            r \leftarrow R
        }}{
            y \notin \supp\big(g_{k,b}(x)\big) \land y \in \supp \big\{ g_{k,b'}(x') \big\}_{\substack{b'\in\bits, \\ x' \in \cX}}
        } = \negl(\lambda)\,,
    \end{equation}
    where $y = \cA(k,r)$ and $(b,x) = \cA^*(k,r,y)$.
\end{definition}

With this, we can define \eiintcf~families as follows.

\begin{definition}[\eiintcf]
\label{def:e^2ntcf}
    A pair of function families $(\cF, \cG)$ is called an \emph{extended extractable noisy trapdoor claw-free family pair} (\eiintcf) if the following conditions hold:
    \begin{myenumi}
        \item $\cF$ is a noisy trapdoor claw free family (see \cref{sec:ntcf});
        \item $\cF$ is injective invariant, with $\cG$ being a corresponding trapdoor injective family (see \cref{sec:ntcf});
        \item $\cG$ satisfies the extractability property in \cref{def:extractability}.
    \end{myenumi}
\end{definition}

\subsection{Single-Round Proof of Quantumness}
\label{subsection:2testpoq}

In the following, we present our second (and less general, but simpler) protocol for a single-round proof of quantumness, which is based on the use of an \eiintcf{} family introduced above. An explicit construction of an \eiintcf{} based on the LWE and LK-$\epsilon$ assumption is given in \cref{section:lwentcf}.

This protocol consists of two different tests, one \emph{image} and one \emph{equation} test. The high-level ideas behind this protocol are the same as for the general protocol: the equation test $\mathtt{Eq}$ tests the prover's ability to find a valid equation, while the image test $\mathtt{Im}$ only accepts images $y$ in the support of $g_{k,b}(x)$ for some $x\in\cX$, i.e. for which we can use the extractability property $\cG$ to conclude that the prover must have ``known'' a preimage $(b,x) \in \bits\times\cX$ to the image $y$ it has returned. By the computational indistinguishability of the two key distributions, the prover cannot tell which test is being performed, and thus has to choose a strategy that works well for both tests on average. Furthermore, the computational indistinguishability of the key distributions allows us to find a lower bound on the  probability that the prover-extractor pair produces both a valid equation and a valid preimage at the same time. 

\begin{myprotocol}
    \caption{Single-Round Proof of Quantumness}
    \label{poq:2test}
    
    \begin{myalgo}
        \Require The prover and the verifier both receive an \eiintcf{} pair $(\mathcal{F},\mathcal{G})$ and a security parameter $\lambda \in \mathbb{N}$. 
        
        \MyState The verifier samples $a \unifsample \{\eqtest, \imtest\}$.
        
        \MyIf{$a = \eqtest$}
            \MyState The verifier samples a key $(k,t_k) \leftarrow \genalg_{\mathcal{F}}(1^\lambda)$ and sends $k$ to the prover.
            \State The verifier receives tuple $(y, d, c)$, computes $x_b = \invalg_{\mathcal{F}}(t_k,b,y)$ for $b \in \bits$ using 
            \StateNoNumber trapdoor, and checks
            \begin{itemize} 
                \item if $\chkalg_\cF(k,y,0,x_0) = \chkalg_\cF(k,y,1,x_1) = 1\,,$
                \item if $(x_0,x_1) \in \cR_k\,,$
                \item if $x_0 \in \cX_0$ and $x_1 \in \cX_1\,$,
                \item if $(J(x_0) \oplus J(x_1))\cdot d = c\,,$
                \item if $d \in G_{k,0,x_0}\cap G_{k,1,x_1}\,.$
            \end{itemize}
            \StateNoNumber If all of these conditions hold, the verifier accepts.
        \MyElsIf{$a = \imtest$}
            \MyState The verifier samples $(k,t_k) \leftarrow \genalg_{\mathcal{G}}(1^\lambda)$ and sends $k$ to the prover.
            \State{The verifier receives a tuple $(y, d, c)$ from the prover, computes $(b,x) = \invalg_\mathcal{G}(t_k,y)$ using}
            \StateNoNumber the trapdoor, and checks whether $\chkalg_\mathcal{G}(k,y,b,x) = 1$. 
            If this holds, the verifier accepts.
        \EndIf
    \end{myalgo}
\end{myprotocol}

The reason why we only need one image test in this case lies in the fact that the extractability property of the \eiintcf{} (\cref{def:extractability}) is stronger than that of the \eiiintcf{} (\cref{def:weak extractability}): in particular, the extractor whose existence is required by the extractability property already returns a preimage in the domain of the function family $\cF$, which makes the strong image test from \protref{poq:3test} (whose purpose was to allow a bound on the probability that the extractor returns an preimage in the ``desired'' domain) redundant. 

Similarly, the soundness proof for \protref{poq:2test} is a simplification of the soundness proof of \protref{poq:3test}. 
For completeness, we include a full soundness proof in \cref{app:e2soundness}.

\begin{theorem}
\label{thm:poq}
Let $\lambda \in \mathbb{N}$, 
and let $(\cF,\cG)$ be an \eiintcf{} pair with paremeter $\injconstant\cF$.
We consider \protref{poq:2test} with inputs $\lambda$ and $(\cF, \cG)$.
\begin{myenumi}
    \item \textbf{Completeness.} There is a QPT prover which succeeds in the protocol with probability at least
    $$
        \frac{1 + \injconstant\cF}{2} - \negl(\lambda) \,.
    $$
    \item \textbf{Soundness.} Any PPT adversary succeeds in the protocol with probability at most
    $$
        \frac{3}{4} + \negl(\lambda) \,.
    $$
\end{myenumi}
\end{theorem}

\section{An Extended Extractable Trapdoor Claw-Free Family from LWE and a Lattice Knowledge Assumption}
\label{section:lwentcf}

In this section, we present a construction of an \eiintcf{} based on LWE (\cref{assump:lwe}). 
This construction is based on the injective invariant NTCF family constructed in~\cite{mahadev2018classical}.
We use the same construction with a slight modification in parameters.
We explain this construction and our modification to it in \cref{sec:ntcf,sec:ntcf_lwe_params}.
In this section, we focus on the extractability property of the trapdoor injective function family $\cG$, which is the main additional feature compared to the construction in~\cite{mahadev2018classical}.

\subsection{Parameter choices for NTCFs from LWE} \label{sec:ntcf_lwe_params}

In~\cite{randomness}, the authors construct an NTCF family based on the LWE assumption, and~\cite{mahadev2018classical} extended this construction to include the injective invariance property.
In our work we use the same NTCF construction, but we need to modify a parameter in the construction slightly in order to combine this NTCF with the LK-$\eps$ knowledge assumption in \cref{section:lwentcf}.
This modification has no substantial impact on the proof in~\cite{randomness}, but for completeness  we explain in this section how and why we modify this parameter and why this change is inconsequential for the reduction in~\cite{randomness}.

Concretely, we use the NTCF construction $\mathcal{F}_{\rm LWE}$ in Chapter 4 of \cite{randomness}, with a minor modification. We replace the parameter choice for $B_P$ (condition (\textbf{A.3}) in~\cite{randomness}) with the following:
\begin{equation*}
    B_P = \frac{q}{2C_Tm\sqrt{(n+1)\log(q)}} = \frac{1}{\poly(\lambda)} \cdot B_P^{\rm BCM}
\end{equation*}
 where $B_P^{\rm BCM}$ denotes the choice of $B_P$ used in \cite{randomness}.
The motivation behind this choice of $B_P$ will become apparent when considering the derivation of the extractability property from the lattice knowledge assumption. It does not impact the $\mathcal{F}_{\rm LWE}$'s property of being a NTCF (i.e. all arguments in the proof of Theorem 4.1 in~\cite{randomness} still hold with the modified parameter choice). In particular, the ratios $B_P/B_V,B_V/B_L $ are still superpolynomial in $\lambda$ (required for $f_{k,b}$, $f'_{k,b}$ to be within negligible distance, Item 3~(iii) of \cref{def:ntcf} of an NTCF). Furthermore, the only part of the proof of $\cF_{\rm LWE}$ being an NTCF (Theorem 4.1 in \cite{randomness}) where the exact value of $B_P$ matters is for the existence of a trapdoor-aided inversion function (Item 2~(i) of \cref{def:ntcf}). Naturally, reducing the value $B_P$ means that the inversion function only has to work on a smaller error radius, thus this condition is still satisfied. The only quantity affected by this change is the statistical security parameter of the construction. In particular, we still have $c_\cF = 1$ for this construction.

\subsection{Trapdoor Injective Function Family}
For completeness, we also give the definition of the corresponding trapdoor injective family $\cG_{\rm LWE}$, which is essentially the same as in~\cite{mahadev2018classical}, again with very slightly different parameters.
Specifically, we define $\mathcal{G}_{\rm LWE}$ as follows: 

\begin{myenumi}
    \item\textbf{Function Family}: Define $g_{k,b}$ by the same algebraic expression as  $f'_{k,b}$.
    \item\textbf{Check Function and State Preparation Algorithm}: $\chkalg_{\mathcal{G}_{\rm LWE}} \coloneqq \chkalg_{\mathcal{F}_{\rm LWE}}$, $\mathrm{SAMP}_{\mathcal{G}_{\rm LWE}} \coloneqq\mathrm{SAMP}_{\mathcal{F}_{\rm LWE}}\;.$ Furthermore, let $\cB_\cG = \bits$.
    \item\textbf{Key Generation Procedure}: Set $\genalg_{\mathcal{G}_{\rm LWE}}(1^\lambda)$ to be the sampling procedure $\mathrm{GenTrap}(1^{n+1}, 1^m, q)$ from Theorem 2.6 in \cite{randomness} and denote by \algofont{Invert} the corresponding inversion function. Let $\mathcal{I}_{\mathcal{G}_{\rm LWE}}$ be the corresponding set of key-trapdoor pairs such that for all $(k,t_k) \in \mathcal{I}_{\mathcal{G}_{\rm LWE}}$, $\algofont{Invert}(t_k,k\mathbf{s}+\mathbf{e}) = (\mathbf{s},\mathbf{e})$ for all $\mathbf{s} \in \mathbb{Z}_q^{n+1}$ and  $\mathbf{e} \in \mathbb{Z}_q^m$ with $\lVert \mathbf{e} \rVert \leq q/C_T\sqrt{(n+1)\log(q)}\;.$
    \item\textbf{Inversion Algorithm}: Define $\invalg_{\mathcal{G}_{\rm LWE}}(t_k,y) = (s_{n+1},\mathbf{s}')$ where $s_{n+1}$ is the last bit and $\mathbf{s}'$ the first n bits of the output $\mathbf{s}$ of $\algofont{Invert}(t_k,\mathbf{y})$.
\end{myenumi}

\begin{lemma}  
\label{lemma:tif}
 $\mathcal{G}_{\rm LWE}$ is a trapdoor injective family.
\end{lemma}

\begin{myproof}
~
\begin{enumerate}[label=\arabic*.]
    \item By Theorem 2.6 of \cite{randomness}, $\genalg_{\mathcal{G}_{\rm LWE}}(1^\lambda)$ fulfills this condition.
    \item By Theorem 2.6 of \cite{randomness}, the probability of sampling $(k,t_k) \notin \mathcal{I}_{\mathcal{G}_{\rm LWE}}$ is negligible. 
    Furthermore, the theorem imples that for our choice of $B_P$ and $g_{k,b}$, we have for all $(k,t_k) \in \mathcal{I}_{\mathcal{G}_{\rm LWE}}$, $b \in \bits$, $\mathbf{x} \in \mathcal{X}$ and $\mathbf{y} \in \supp(g_{k,b}(\mathbf{x}))$ that $\invalg_{\mathcal{G}_{\rm LWE}}(t_k,\mathbf{y}) = (b,\mathbf{x})$.
    \item This follows directly from $\mathcal{F}_{\rm LWE}$ being an NTCF family and our choice of $\mathcal{G},\chkalg_{\mathcal{G}_{\rm LWE}}$ and $\mathrm{SAMP}_{\mathcal{G}_{\rm LWE}}\;.$
\end{enumerate}
\end{myproof}

\begin{lemma} 
$\mathcal{F}_{\rm LWE}$ is an injective invariant family under the hardness assumption \emph{LWE}$_{l,q,D_{\mathbb{Z}_q,B_L}}$.
\end{lemma} 

\begin{myproof}
For $\mathcal{G} = \mathcal{G}_{\rm LWE}$, the first condition is true by construction. The second condition follows from the fact that the marginal distributions of $\genalg_{\mathcal{F}_{\rm LWE}}$ and $\genalg_{\mathcal{G}_{\rm LWE}}$ on the key $k$ are both computationally indistinguishable from the uniform distribution on $\mathbb{Z}_q^{m\times(n+1)}$. For the former, this be seen by applying Lemma 9.3 from \cite{mahadev2018classical}, which assumes the hardness of LWE$_{l,q,D_{\mathbb{Z}_q,B_L}}$, as well as Theorem 2.6 of \cite{randomness}. 
For the latter, this follows from Theorem 2.6 of \cite{randomness}. Thus, they are also computationally indistinguishable from each other.
\end{myproof}

\subsection{Extractability Property from Lattice Knowledge Assumption}

Let $(\cF_{\rm LWE}, \cG_{\rm LWE})$ be the injective invariant NTCF family described in \cref{sec:ntcf_lwe_params}.
The extractability property for this family can be derived from a lattice knowledge assumption. The particular lattice knowledge assumption that we use states (informally) that any classical algorithm which can find a point suitably close to a lattice must have ``known'' the corresponding lattice point. This notion was formalized in \cite{k_assumption} as the LK-$\epsilon$ assumption, see \cref{assumpt:lke}. 

\begin{lemma}
\label{lemma:extprop}
$\mathcal{G}_{\rm LWE}$ satisfies the extractability property under the assumption that the marginal distribution of
$\algofont{GenTrap}(1^{n+1}, 1^m, q)$ over the key $\mathbf{A}$ satisfies the LK-$\frac{1}{4}$ assumption. 
\end{lemma}

\begin{myproof}
To emphasize that the key $k$ produced by $\algofont{GenTrap}(1^{n+1}, 1^m, q)$ is a matrix, we denote it as $\mathbf{A}$ in the following. 

First, recall that by Theorem 4.6 of \cite{randomness}, for all $(\mathbf{A},t_\mathbf{A}) \in \mathcal{I}_{\mathcal{G}_{\rm LWE}}$ and all $y \in \mathbb{Z}_q^m$, if $\mathbf{y}$ is within $l_2$-distance $r(\lambda) = q/C_T\sqrt{(n+1)\log q}$ of a lattice point $\mathbf{p}$ of $\cL(\mathbf{A})$, \algofont{Invert} can correctly recover the preimage $(\mathbf{x},b)$ of the closest lattice point $\mathbf{p}$ under $\mathbf{A}$. In particular, this implies that the $l_2$-distance between 2 lattice points of $\cL(\mathbf{A})$ is at least $2r$. Thus, in particular, the $l_2$ norm of a shortest lattice vector with respect to the infinity norm is at least $2r$. Since $\lVert\cdot\rVert_2 \leq \sqrt{m} \lVert\cdot\rVert_\infty$ on $\mathbb{R}^m$, we can conclude that 
$2r \leq \sqrt{m}\lambda_\infty(\cL(\mathbf{A}))$. From this we can conclude that for our choice $B_P = r/2\sqrt{m}$, we have that for all $(\mathbf{A},t_\mathbf{A}) \in \mathcal{I}_{\mathcal{G}_{\rm LWE}}\,,$
\begin{equation}
\label{eq:bound}
    \sqrt{m}B_P \leq \frac{1}{4}\lambda_\infty(\cL(\mathbf{A}))\,.
\end{equation}
Let $\mathcal{A}$, $\mathcal{A}^*$ be classical probabilistic algorithms as in the definition of the LK-$\epsilon$ assumption, though, for the sake of convenience, let $\mathcal{A}$ output the point $\mathbf{y}$ it generates, and $\mathcal{A}^*$ outputs a preimage $(\mathbf{x},b)$ of $\mathbf{p}$ under $\mathbf{A}$ by performing Gaussian elimination, which is an efficient procedure. We denote the random coin distribution of $\mathcal{A},\mathcal{A}^*$ by $r \leftarrow R$.
Furthermore, we introduce the following events:
\begin{itemize}
    \item[\textbullet] $S_\mathcal{A}$: $\mathcal{A}$ produces a point $\mathbf{y}$ within $\frac{1}{4}\lambda_\infty(\mathcal{L}(\mathbf{A}))$ of a grid point of $\mathbf{A}$.
    \item[\textbullet] $S_{\mathcal{A}^*}$: $\mathcal{A}^*$ produces a point $\Tilde{\mathbf{x}} = (\mathbf{x},b)$ such that $\mathbf{A}\Tilde{\mathbf{x}}$ is the closest lattice point to the output $\mathbf{y}$ of $\mathcal{A}$.
    \item[\textbullet] $R_\mathcal{A}$: $\mathcal{A}$ produces a point within $d_{max} \coloneqq \sqrt{m}B_P$ of a grid point $\mathbf{p}$ of $\mathbf{A}$, and there exists a preimage $\Tilde{\mathbf{x}}$ of $\mathbf{p}$ under $\mathbf{A}$ such that $\Tilde{\mathbf{x}} \in \Z_q^n\times\{0,1\}$.
    \item[\textbullet] $R_{\mathcal{A}^*}$: $\mathcal{A}^*$ produces a point $\Tilde{\mathbf{x}} = (\mathbf{x},b)\in \Z_q^n\times\{0,1\}$ such that $\mathbf{A}\Tilde{\mathbf{x}}$ is within distance $d_{max}$ of the output $\mathbf{y}$ of $\mathcal{A}$.
\end{itemize}
We will denote the probability distribution Pr$_{(k,t_k) \leftarrow \genalg_{\mathcal{G_{\rm LWE}}}(1^\lambda), \ r \leftarrow R}$ simply as Pr from now on. 
Then, the LK-$\frac{1}{4}$ assumption states that 
\begin{equation*}
    \text{Pr}\left[\overline{S_{\mathcal{A}}}_{^*} \land S_{\mathcal{A}}\right] = \negl(\lambda)\;.
\end{equation*}
It follows that:
\begin{align*}
    \text{negl}(\lambda) &= \text{Pr}\left[\,\overline{S_{\mathcal{A}}}_{^*}\land S_{\mathcal{A}}\right] \geq \text{Pr}\left[\,\overline{S_{\mathcal{A}}}_{^*} \land S_{\mathcal{A}}\land (\mathbf{A},t_\mathbf{A}) \in \mathcal{I}_{\mathcal{G}_{\rm LWE}}\right]  \\
    &\annotatesign{Eq. \eqref{eq:bound}}{\geq}\hspace{15pt} \text{Pr}\left[\,\overline{S_{\mathcal{A}}}_{^*} \land R_{\mathcal{A}} \land (\mathbf{A},t_\mathbf{A}) \in \mathcal{I}_{\mathcal{G}_{\rm LWE}}\right] = \text{Pr}\left[\,\overline{R_{\mathcal{A}}}_{^*} \land R_{\mathcal{A}} \land (\mathbf{A},t_\mathbf{A}) \in \mathcal{I}_{\mathcal{G}_{\rm LWE}}\right]\, \\
    &\geq \text{Pr}\left[\,\overline{R_{\mathcal{A}}}_{^*} \land R_{\mathcal{A}}\right] - \text{Pr}\left[(\mathbf{A},t_\mathbf{A}) \notin \mathcal{I}_{\mathcal{G}_{\rm LWE}}\right]\,
    = \text{Pr}\left[\,\overline{R_{\mathcal{A}}}_{^*} \land R_{\mathcal{A}}\right] - \negl(\lambda)\;.
\end{align*}
Now, note that the event $\overline{R_{\mathcal{A}}}_{^*} \land R_{\mathcal{A}}$ is equivalent to the event $y \notin \supp\big(g_{k,b}(x)\big) \land y \in \supp \big\{ g_{k,b'}(x') \big\}_{\substack{b'\in\bits, x' \in \cX}}$.
Thus,
\begin{align*}
    \text{negl}(\lambda) = \text{Pr}\left[y \notin \supp\big(g_{k,b}(x)\big) \land y \in \supp \big\{ g_{k,b'}(x') \big\}_{\substack{b'\in\bits,\,x' \in \cX}}\right]\,.
\end{align*}
\end{myproof}

Together with the properties in \cref{sec:ntcf_lwe_params}, this implies the existence of an \eiintcf~based on the LWE and LK-$\eps$ assumptions:
\begin{theorem}[Existence of an \eiintcf{}]
\label{thm:e^2NTCF}
The family pair ($\mathcal{F}_{\rm LWE},\mathcal{G}_{\rm LWE}$) is an \eiintcf{} with completeness parameter $c_\cF = 1$ under the hardness assumption \emph{LWE}$_{l,q,D_{\mathbb{Z}_q,B_L}}$ and assuming that the marginal distribution of
$\algofont{GenTrap}(1^{n+1}, 1^m, q)$ over the key $\mathbf{A}$ satisfies the \emph{LK}-$\frac{1}{4}$ assumption.
\end{theorem}

Using this \eiintcf~construction in \protref{poq:2test}, we then get a single-round proof of quantumness from LWE and the LK-$\eps$ assumption.

\begin{corollary}[Existence of Single-Round Proof of Quantumness based on LWE and LK-$\eps$]
\label{corr:2testsingleroundpoq} Assuming that \emph{LWE}$_{l,q,D_{\mathbb{Z}_q,B_L}}$ is hard and that the marginal distribution of
$\algofont{GenTrap}(1^{n+1}, 1^m, q)$ over the key $\mathbf{A}$ satisfies the \emph{LK}-$\frac{1}{4}$ assumption., there exists a single-round proof of quantumness with completeness 1 and soundness $\frac{3}{4}$.  
\end{corollary}


\bibliographystyle{alpha}
\bibliography{main}

\appendix

\section{Results used in Proof of \cref{thm:poq2}}
In this appendix, we collect a number of technical results that we use in the proof of \cref{thm:poq2}.

\begin{lemma} 
\label{lemma:rank1A}
Let $\lambda$ be a security parameter, $a,b \geq 1$ integer functions of $\lambda$, $q$ a prime in the range $\left[2^{\lambda - 1},2^\lambda\right)$.
Then statistical distance of the distributions $\bR\leftarrow_U \rank_1(\mathbb{Z}_q^{a\times b})$ and $\bR\leftarrow_U \rank_{\leq 1}(\mathbb{Z}_q^{a\times b})$
is negligible in $\lambda$.
\end{lemma}

\noindent\emph{Proof}. There is only one rank 0 matrix and exponentially (in $\lambda$) many rank 1 matrices in the set $\rank_{\leq 1}(\mathbb{Z}_q^{a\times b})$, so the statistical distance between the distributions is negligible in $\lambda$.  \qed

It is easy to see that one can sample uniformly from $\rank_{\leq 1}$ by taking the outer product of two random vectors (see e.g.~\cite[page 4]{Goldwasser2010RobustnessOT}).
\begin{lemma}
\label{lemma:rank1B}
Let $q$ be a prime, $a,b$ integers.
The distributions $\bR \leftarrow_U \rank_{\leq 1}(\mathbb{Z}_q^{a\times b})$ and $R \leftarrow \{\bu\cdot \bv^T | \bu \leftarrow_U \mathbb{Z}_q^{a}, \bv \leftarrow_U \mathbb{Z}_q^{b}\}$ are identical. 
\end{lemma}

\begin{lemma}
\label{lemma:rank1C}   
Let $\lambda$ be a security parameter, $a,b \geq 1$ integer functions of $\lambda$, $q$ a prime in the range $\left[2^{\lambda - 1},2^\lambda\right)$. Then the statistical distance between the distributions  
\begin{equation*}
    D_0 = \big\{\bu\bv^T |\bu \leftarrow_U \mathbb{Z}_q^a , \bv \leftarrow_U \mathbb{Z}_q^{b}\big\}
\end{equation*}
and 
\begin{equation*}
    D_1 = \big\{(1,\bu)\bv^T |\bu \leftarrow_U \mathbb{Z}_q^{a-1} , \bv \leftarrow_U \mathbb{Z}_q^{b}\big\}
\end{equation*}
is negligible in $\lambda$.
\end{lemma}

\noindent \emph{Proof.} 
First, note that the statistical distance between the distributions $u \leftarrow_U \mathbb{Z}_q$ and $u \leftarrow_U \mathbb{Z}_q \setminus\{0\}$ is negligible in $\lambda$. Thus, the statistical distance between the distributions $D_0$ and
\begin{equation*}
    D' = \big\{\bu\bv^T |\bu \leftarrow_U (\mathbb{Z}_q \setminus\{0\})\times \mathbb{Z}_q^{a-1} , \bv \leftarrow_U \mathbb{Z}_q^{b}\big\}
\end{equation*}
is also negligible in $\lambda$. 

Furthermore, the statistical distance between $D'$ and $D_1$ is zero. Thus, $D_0$ is indistinguishable from $D_1$. \qed

\begin{corollary}[Corollary of \crefrange{lemma:rank1A}{lemma:rank1C}]
\label{corollary:rank1}
Let $\lambda$ be a security parameter, $a,b \geq 1$ integer functions of $\lambda$, $q$ a prime in the range $\left[2^{\lambda - 1},2^\lambda\right)$.
Then the distributions $\bR\leftarrow_U \rank_1(\mathbb{Z}_q^{a\times b})$, $\bR\leftarrow \big\{\bu\bv^T |\bu \leftarrow_U \mathbb{Z}_q^{a}, \bv \leftarrow_U \mathbb{Z}_q^{b}\big\}$ and  $R\leftarrow \big\{(1,\bu)\bv^T |\bu \leftarrow_U \mathbb{Z}_q^{a-1}, \bv \leftarrow_U \mathbb{Z}_q^{b}\big\}$ are computationally indistinguishable.
\end{corollary}

It is a standard combinatorial result that a uniformly random matrix  has full rank with overwhelming probability:
\begin{lemma}[Rank of uniformly distributed matrix]
\label{lemma:rudm}
Let $\lambda$ be a security parameter, q be a prime in the range $\left[2^{\lambda-1},2^\lambda\right)$ and $a = O(\lambda)$, $b$ an integer function of $\lambda$. Then 
\begin{equation*}
    \Pr_{\bM \leftarrow_U \mathbb{Z}_q^{a\times (a+b)}}\left[{\rm rank}(\bM) < a\right] = O(2^{-\lambda})\,.
\end{equation*}
\end{lemma}

\begin{corollary}[Corollary of \cref{lemma:rudm}]
\label{corr:rudm}
Let $\lambda$ be a security parameter, q be a prime in the range $\left[2^{\lambda-1},2^\lambda\right)$ and $a,b = O(\lambda)$. Then the distributions $\bR \leftarrow_U \mathbb{Z}_q^{a\times b}$ and $\bR \leftarrow_U \rank_{\min\{a,b\}}(\mathbb{Z}_q^{a\times b})$ are statistically indistinguishable.
\end{corollary}

\begin{lemma}[Lemma 4.6 from~\cite{randomness}] \label{lemma:modmat} Let $q$ be a prime, $l,n \geq 1$ integers and $\mathbf{C}\in \Z_q^{l\times n}$ a uniformly random matrix. With probability at least $1 - q^l\cdot2^{-\frac{n}{8}}$ over the choice of $\mathbf{C}$ the following holds. For a fixed $\mathbf{C}$, all $\mathbf{v}\in\Z_q^l$ and $\hat{d} \in \bits^n\setminus\{0^n\}$, the distribution of $(\hat{d}\cdot s\mod 2)$, where $s$ is uniform in $\bits^n$, conditioned on $\mathbf{Cs} = \mathbf{v}$, is within statistical distance $O(q^\frac{3l}{2}\cdot 2^{-\frac{n}{40}})$ of the uniform distribution over $\bits$.
\end{lemma} 

\begin{lemma}
\label{lemma:comp_ind2}
Let $\groupgen$ be an algorithm that takes as input a security parameter $\lambda$ and outputs a tuple $(q,\mathbb{G},g)$, where $p$ is an $\lambda$-bit prime, $\mathbb{G}$ a cyclic group of order $q$ and $g$ a generator of $\mathbb{G}$. Let a,b be integer functions of $\lambda$, s.t. $ 40\cdot log(q) \leq a, b = O(log(q))$. Assume that $\groupgen$ fulfills the DDH assumption. Then the distributions
\begin{align*}
    D_0 &= \big\{(\mathbb{G},q,g,g^\bR,g^{\bR\bs}) | (q,\mathbb{G},g) \shortleftarrow \groupgen, \bR\shortleftarrow_U \mathbb{Z}_q^{a\times b}, \bs \shortleftarrow_U \bits^b\big\}\\
    D_1 &= \big\{(\mathbb{G},q,g,g^\bR,g^{\bR\bs}) | (q,\mathbb{G},g) \shortleftarrow \groupgen,  \bu \shortleftarrow_U \mathbb{Z}_q^a, \bv \shortleftarrow_U \mathbb{Z}_q^b, \bs \shortleftarrow_U \bits^b, \bR = \bu\bv^T\big\}\\
    D_2 &= \big\{(\mathbb{G},q,g,g^{\bR'}) | (q,\mathbb{G},g) \shortleftarrow \groupgen, \bR'\shortleftarrow_U \mathbb{Z}_q^{a\times (b+1)}\big\}\\
    D_3 &= \big\{(\mathbb{G},q,g,g^{\bR'}) | (q,\mathbb{G},g) \shortleftarrow \groupgen, \bu \shortleftarrow_U \mathbb{Z}_q^{a-1} , \bv \shortleftarrow_U \mathbb{Z}_q^{b+1}, \bR' = (1,\bu)\bv^T\big\}
\end{align*}
are computationally indistinguishable for a PPT adversary.

\end{lemma}

\begin{myproof}
First, note that due to \cref{corr:rudm} the distribution $D_0$ is computationally indistinguishable from 
\begin{equation*}
    D^{(0)} = \big\{(\mathbb{G},q,g,g^\bR,g^{\bR\bs}) | (q,\mathbb{G},g) \leftarrow \groupgen, \bR\leftarrow_U \rank_{\min(a,b)}(\mathbb{Z}_q^{a\times b}), \bs \leftarrow_U \bits^b\big\}
\end{equation*}
which in turn is indistinguishable from 
\begin{equation*}
    D^{(1)} = \big\{(\mathbb{G},q,g,g^\bR,g^{\bR\bs}) | (q,\mathbb{G},g) \leftarrow \groupgen, \bR\leftarrow_U \rank_1(\mathbb{Z}_q^{a\times b}), \bs \leftarrow_U \bits^b\big\}
\end{equation*}
due to the matrix 1-linear assumption, which holds since we assume $\groupgen$ to fulfill the DDH assumption. 

Next, it follows from \cref{corollary:rank1} that $D^{(1)}$ is computationally indistinguishable from
\begin{multline*}
    D_1 = \big\{(\mathbb{G},q,g,g^\bR,g^{\bR\bs}) | (q,\mathbb{G},g) \leftarrow \groupgen, \bR = \bu\bv^T \text{ for } \bu \leftarrow_U \mathbb{Z}_q^a,
    \bv \leftarrow_U \mathbb{Z}_q^b, \bs \leftarrow_U \bits^b\big\} \,.
\end{multline*}
Now, note that \cref{lemma:modmat} implies in particular that for $\bv \leftarrow_U \mathbb{Z}_q^b$, $\bs \leftarrow_U \bits^b$ the distribution of $\bv^T\bs$ is within statistical distance at most $q^{1/2}\cdot2^{-b/40} = O(2^{-\lambda})$ from the uniform distribution over $\mathbb{Z}_q$. Thus, $D_1$ is computationally indistinguishable from 
\begin{align*}
    D^{(2)} &= \big\{(\mathbb{G},q,g,g^\bR,g^{z\cdot \bu}) | (q,\mathbb{G},g) \leftarrow \groupgen, \bu \leftarrow_U \mathbb{Z}_q^a, \bv \leftarrow_U \mathbb{Z}_q^b, z \leftarrow_U \mathbb{Z}_q, \bR = \bu\bv^T\big\} \\
    &= \big\{(\mathbb{G},q,g,g^{\bR'}) | (q,\mathbb{G},g) \leftarrow \groupgen, \bu \leftarrow_U \mathbb{Z}_q^a, \bv' \leftarrow_U \mathbb{Z}_q^{(b+1)}, \bR' = \bu\bv'^T\big\} \,.
\end{align*}
Again using \cref{corollary:rank1}, we can conclude that $D^{(2)}$ is indistinguishable from 
\begin{equation*}
    D^{(3)} = \big\{(\mathbb{G},q,g,g^{\bR'}) | (q,\mathbb{G},g) \leftarrow \groupgen, \bR' \leftarrow_U \rank_1(\mathbb{Z}_q^{a\times (b+1)})\big\} \,,
\end{equation*}
which is in turn (due to the matrix 1-linear assumption, which follows from DDH) indistinguishable from 
\begin{equation*}
    D^{(4)} = \big\{(\mathbb{G},q,g,g^{\bR'}) | (q,\mathbb{G},g) \leftarrow \groupgen, \bR' \leftarrow_U \rank_{\rm min\{a,b+1\}}(\mathbb{Z}_q^{a\times (b+1)})\big\} \,.
\end{equation*}
Furthermore, it follows from \cref{corr:rudm} that $D^{(4)}$ is indistinguishable from $D_2$, and from \cref{corollary:rank1} that $D_1$ is indistinguishable from $D_3$. Thus, all these distributions are computationally indistinguishable from each other. 
\end{myproof}

\section{Omitted proofs in \cref{section:ddhntcf}}
\label{app:ddh_ntcf}
In this section, we provide the proof of a key lemma in \cref{section:ddhntcf}, the AHCB property (\cref{app:ddh_ahcb_proof}).

\subsection{Proof of \cref{lem:ddh_ahcb}} \label{app:ddh_ahcb_proof}
\begin{myproof}[\cref{lem:ddh_ahcb}]
It is straightforward to see that $\cF_{\rm DDH}$ fulfills the first three conditions of \cref{def:ntcf}:
\begin{myenumi}
    \item \textbf{Efficient Function Generation.} This condition holds since all steps of the key generation procedure $\genalg_{\cF_{\rm DDH}}$ are efficient. 
    \item \textbf{Trapdoor Injective Pair.} Let $\mathcal{I}_{\cF_{\rm DDH}}$ be the set of key-trapdoor pairs such that $\mathbf{A}$ is full-rank (and thus in particular injective). By \cref{lemma:rudm} this set is negligible in $\lambda$. 

    Now, note that for all pairs $(k,t_k) \in \mathcal{I}_{\cF_{\rm DDH}}$:
    \begin{enumerate}
        \item \emph{Trapdoor}: Note that all steps of the inversion algorithm $\invalg_{\cF_{\rm DDH}}$ are efficient. Since $\mathbf{A}$ is full-rank, its pseudo-inverse exists, and it is easy to check that $\invalg_{\cF_{\rm DDH}}$ recovers the correct preimage for all $b \in \bits$, $\mathbf{x} \in \cX$ and $\mathbf{y} \in $supp$(f_{k,b}(\mathbf{x}))$.
        \item \emph{Injective Pair:} This follows from $\mathbf{A}$ being full-rank and from the same argumentation as in the proof of Theorem 2 in the supplementary information of \cite{poqbell}. As they mention, it is possible to achieve a value of $c_\cF$ arbitrarily close to 1. In particular, if we set $d = n^2$, we have already for $\lambda_{c_\cF} = 1$ that $c_\cF = 0.99$.
    \end{enumerate}
    \item \textbf{Efficient Range Superposition.} \hfill
    \begin{enumerate}
    \item Follows from $f_{k,b} = f'_{k,b}$.
    \item Evident from definition of $\chkalg_{\cF_{\rm DDH}}$.
    \item Follows from $f_{k,b} = f'_{k,b}$ and from the fact that the functions $f_{k,b}$ are efficiently computable classically.
    \end{enumerate}
    \item \textbf{Adaptive Hardcore Bit.}
    The proof of the adaptive hardcore bit property is slightly more involved, so we show it separately as \cref{lemma:ahcbmain} below.
\end{myenumi}
\end{myproof}

To prove the adaptive hardcore bit property of $\cF_{\rm DDH}$, we adapt the corresponding proof for $\cF_{\rm LWE}$ from \cite{randomness}.
The core of their proof is~\cite[Lemma 4.4]{randomness}. We prove a statement that is analogous to~\cite[Lemma 4.4]{randomness} in \cref{lemma:ahcbmain} using a similar lossy sampling argument as in~\cite{randomness} and \cref{lemma:modmat}~(\cite[Lemma 4.6]{randomness}). This result is then used to show \cref{lemma:ahcbside}, which is analogous to~\cite[Lemma 4.3]{randomness}. In contrast to the LWE case, the proof does not end here, since we need to take some technicalities into account due to the fact that not all elements of the domain $\cX$ are part of a claw. A restricted-domain version of \cref{lemma:ahcbside} is proven in  \cref{lemma:ahcbresdom}, and we show in \cref{lemma:setcorr} that this statement is equivalent to the AHCB property we require for $\cF_{\rm DDH}$.

\begin{lemma}[Analog to Lemma 4.4 from \cite{randomness}]
\label{lemma:ahcbmain}
    Let $\groupgen$ be an algorithm that takes as input a security parameter $\lambda$ and outputs a tuple $(q,\mathbb{G},g)$, where $q$ is an $\lambda$-bit prime (as in, q is in the range $\left[2^{\lambda - 1},2^\lambda\right)$), $\mathbb{G}$ a cyclic group of order $q$ and $g$ a generator of $\mathbb{G}$. Assume that $\groupgen$ fulfills the DDH assumption. Let $n$ be an integer function such that $121\cdot$\emph{log}$(q)\leq n = \cO(\lambda)$ and $w = n\lceil$\emph{log}$(q)\rceil$.
    
    Let $\mathcal{A}$ be a PPT algorithm that takes as input $(q,\mathbb{G},g)$ as well as an element of $\mathbb{G}^{(n+1)\times(n+1)}$ and has outputs in $\bits\times\mathbb{Z}_q^{n+1}\times\bits^w\times\bits$. Let $\delta$ be the indicator function and $I_{b,\mathbf{x}}$, $\hat{G}_{\mathbf{s}_{b\oplus1},b,\mathbf{x}}$ are defined as in Lemma 4.4 in \cite{randomness}.

    Then the distributions
    \begin{equation*}
    \begin{split}
        D_0 = \Big\{k = (\mathbb{G},q,g,g^\mathbf{A},g^{\mathbf{As}}), (b,\mathbf{x},\mathbf{d},c) \leftarrow \mathcal{A}(k),I_{b,\mathbf{x}}(\mathbf{d})\cdot \mathbf{s} \mod 2|(q,\mathbb{G},g) \leftarrow \groupgen, \\
        \mathbf{A} \leftarrow_U \mathbb{Z}_q^{(n+1)\times n},\  \mathbf{s}\leftarrow_U \bits^n\Big\}
    \end{split}
    \end{equation*}
    and 
    \begin{equation*}
    \begin{split}
        D_1 = \Big\{k = (\mathbb{G},q,g,g^\mathbf{A},g^{\mathbf{As}}), (b,\mathbf{x},\mathbf{d},c) \leftarrow \mathcal{A}(k),(\delta_{d \in \hat{G}_{\mathbf{s}_{b\oplus1},b,\mathbf{x}}}r)\oplus I_{b,\mathbf{x}}(\mathbf{d})\cdot \mathbf{s} \mod 2 \\| (q,\mathbb{G},g) \leftarrow \groupgen, 
        \mathbf{A} \leftarrow_U \mathbb{Z}_q^{(n+1)\times n}, \ \mathbf{s}\leftarrow_U \bits^n\Big\}
    \end{split}
    \end{equation*}
    are computationally indistinguishable.
\end{lemma}

\begin{myproof}
First, note that by \cref{lemma:comp_ind2} (for the case $a = n+1,\,b = n$) $D_0$ is indistinguishable from  
\begin{equation*}    
\begin{split}
D^{(1)} = \Big\{k = (\mathbb{G},q,g,g^\mathbf{A},g^{\mathbf{As}}), (b,\mathbf{x},\mathbf{d},c) \leftarrow \mathcal{A}(k),I_{b,\mathbf{x}}(d)\cdot \mathbf{s} \mod 2|(q,\mathbb{G},g) \leftarrow \groupgen,\\
\mathbf{A} = \mathbf{u}\mathbf{v}^T \text{ for }\mathbf{u} \leftarrow_U \mathbb{Z}_q^{n+1},
\mathbf{v}  \leftarrow_U \mathbb{Z}_q^n, 
\mathbf{s}\leftarrow_U \bits^n\Big\}\;
\end{split}
\end{equation*}

Now, consider the following distribution: 
\begin{equation*} 
\begin{split}
D^{(2)} = \Big\{k = (\mathbb{G},q,g,g^\mathbf{A},g^{\mathbf{As}}), (b,\mathbf{x},\mathbf{d},c) \leftarrow \mathcal{A}(k),(\delta_{\mathbf{d} \in \hat{G}_{\mathbf{s}_{b\oplus1},b,\mathbf{x}}}r)\oplus I_{b,\mathbf{x}}(\mathbf{d})\cdot \mathbf{s} \mod 2|\\(q,\mathbb{G},g) \leftarrow \groupgen, \mathbf{A} = \mathbf{uv}^T \text{ for }\mathbf{u} \leftarrow_U \mathbb{Z}_q^{n+1}, \mathbf{v} \leftarrow_U \mathbb{Z}_q^n, \mathbf{s}\leftarrow_U \bits^n\Big\}\;
\end{split}
\end{equation*}

To show the indistinguishability of $D^{(1)}$ and $D^{(2)}$, we first define the relevant random variables. We denote by $I_{b,\mathbf{x}}(\mathbf{d})_b, \mathbf{s}_b$ the first ($b=0$) and last $(b=1)$ half of the vectors $I_{b,\mathbf{x}}(\mathbf{d})$ and $\mathbf{s}$ and denote by $C = f(Y)$ whenever a random variable $C$ is deterministic given $Y$: 
\begin{equation*}
\begin{array}{llll}
X = (\mathbb{G},q,g,\mathbf{u},\mathbf{v},\mathbf{v}^T\mathbf{s},\text{coins}\left[\cA\right]) \\
B = (k,b,\mathbf{x},\mathbf{d},c) = f(X) \\
S_1 = \mathbf{s}_b = f(X,S)\\
S_2 = \mathbf{s}_{b\oplus1} =  f(B,S) = f(X,S)\\
\end{array}
\qquad
\begin{array}{llll}
A_1 = I_{b,\mathbf{x}}(\mathbf{d})_b\cdot \mathbf{s}_b \mod 2 \\
A_2 = I_{b,\mathbf{x}}(\mathbf{d})_{b\oplus1}\cdot \mathbf{s}_{b\oplus1} \mod 2  = f(X,S_2)\\
A = A_1 \oplus A_2\\ 
A' = (\delta_{d \in \hat{G}_{\mathbf{s}_{b\oplus1},b,\mathbf{x}}}r)\oplus I_{b,\mathbf{x}}(\mathbf{d})\cdot \mathbf{s} \mod 2 = f(X)\cdot r 
\end{array}
\end{equation*} 
Note that given $(X,S_2)$, the random variable $B$ is deterministic, and thus the values  $I_{b,\mathbf{x}}(\mathbf{d})_b$ and $\mathbf{v}^T\mathbf{s}_b$ are also fixed. Thus, we can conclude using \cref{lemma:modmat} that given $(X,S_2)$, $A_1$ is within statistical distance $\cO(q^{\frac{3}{2}}\cdot 2^{-n/80}) = \negl(\lambda)$ of the uniform distribution over $\bits$ (with probability $1-\negl(\lambda)$ over the distribution of $\mathbf{v}$, and thus also that of $(X,S_2)$). Thus, in particular, since $A_2$ is constant given $(X,S_2)$, $A$ is within statistical distance $\cO(q^{\frac{3}{2}}\cdot 2^{-n/80}) = \negl(\lambda)$ of the uniform distribution over $\bits$ (with probability $1-\negl(\lambda)$ over $X$). For the values of $(X,S_2)$ where $\mathbf{d} \notin \hat{G}_{s_{b\oplus1},b,\mathbf{x}}$, $A$ and $A'$ are identical and thus the statistical distance is 0. For all other values of $(X,S_2)$, $A' = r \oplus A$ for r drawn uniformly at random and thus $A'$ is distributed uniformly at random and therefore within negligible statistical distance of $A$. Thus, given $(X,S_2)$, $A'$ is within negligible statistical distance of $A$. Since $B$ is a deterministic function given $X$, the joint distribution of $(A,B)$ is within negligible statistical distance of $(A',B)$ given $(X,S_2)$.

Denote by $P_{AB}$, $P_{A'B}$  the joint probability distribution of $(A,B)$ and $(A',B)$ respectively. Note that they correspond to the distributions $D^{(1)},D^{(2)}$.

Then, we have that 
\begin{align*}
    d(P_{AB},P_{A'B}) &= \frac{1}{2}\sum_{a,b} |P_{AB}(a,b)-P_{A'B}(a,b)| \\
    &= \frac{1}{2}\sum_{a,b}|\sum_{\bx,\bs_2}(P_{ABXS_2}(a,b,\bx,\bs_2)-P_{A'BXS_2}(a,b,\bx,\bs_2))| \\
    &\leq \sum_{\bx,\bs_2} P_{XS_2}(\bx,\bs_2)\frac{1}{2}\sum_{a,b}|(P_{AB|XS_2}(a,b|\bx,\bs_2)-P_{A'B|XS_2}(a,b|\bx,\bs_2))| \\
    &\leq (1 - \negl(\lambda))\cdot\negl(\lambda) + \negl(\lambda)\cdot 1 = \negl(\lambda)\;.
\end{align*}

Thus, $D^{(1)},D^{(2)}$ are statistically indistinguishable from one another. 

Finally, we can use the same line of reasoning as between $D_0$ and $D^{(1)}$ to conclude that $D^{(2)}$ is statistically indistinguishable from $D_1$. 
\end{myproof}

Since the proof of Lemma 4.4 from Lemma 4.3 in \cite{randomness} does not depend on the exact key distribution, and \cref{lemma:ahcbmain} is analogous to Lemma 4.3 in \cite{randomness}, the exact same proof as in \cite{randomness} shows that our \cref{lemma:ahcbmain} implies the following AHCB property.
\begin{lemma}[Analog to Lemma 4.3 from \cite{randomness}]
\label{lemma:ahcbside}
Let $\groupgen$ be a PPT algorithm that takes as input a security parameter $\lambda \in \mathbb{N}$ and outputs a tuple $(q,\mathbb{G},g)$, where $q$ is an $\lambda$-bit prime, $\mathbb{G}$ a cyclic group of order $q$ and $g$ a generator of $\mathbb{G}$,  and $n: \mathbb{N}\rightarrow \mathbb{N}$ is such that $121\cdot$\emph{log}$(q)\leq n = \cO(\lambda)$. Let $\genalg_k$ be the algorithm that samples $(q,\mathbb{G},g) \leftarrow \groupgen$ as well as $\mathbf{A} \leftarrow_U \rank_n(\mathbb{Z}_q^{(n+1)\times n})$ and $\mathbf{s} \leftarrow_U\bits^n$ and outputs $k = (q,\mathbb{G},g,g^\mathbf{A},g^\mathbf{As})$ and $t_k = (\mathbf{A},\mathbf{s})$. Let $w = n\lceil$\emph{log}$(q)\rceil$. Let $J: \Z_q^n \rightarrow \{0,1\}^w$ be such that $J(\bx)$ returns the binary representation of $\bx \in \Z_q^n$.

Assume that $\groupgen$ fulfills the DDH assumption.  Let $\mathbf{s} \in \bits^n$ and
\begin{align*}
    H_\mathbf{s} &= \big\{(b,\mathbf{x},\mathbf{d},\mathbf{d}\cdot(J(\mathbf{x})\oplus J(\mathbf{x} - (-1)^b\mathbf{s}))) | b \in \bits, \mathbf{x} \in \mathbb{Z}_q^n, \mathbf{d} \in \hat{G}_{\mathbf{s}_{b\oplus1},b,\mathbf{x}}\big\} \;, \\
    \overline{H}_\mathbf{s} &= \big\{(b,\mathbf{x},\mathbf{d},c) | (b,\mathbf{x},\mathbf{d},c\oplus 1) \in H_\mathbf{s}\big\}\;.
\end{align*} 
Then for any PPT algorithm $\mathcal{A}$ that takes as input $(q,\mathbb{G},g)$ as well as an element of $\mathbb{G}^{(n+1)\times(n+1)}$ and has outputs in $\bits\times\mathbb{Z}_q^{n+1}\times\bits^w\times\bits$, it holds that
\begin{equation*}
\Big|\Pr_{\kt\leftarrow \genalg_{k}(1^{\lambda})}\left[\mathcal{A}(k) \in H_\mathbf{s}\right] - \Pr_{\kt\leftarrow \genalg_{k}(1^{\lambda})}\left[\mathcal{A}(k) \in\overline{H}_\mathbf{s}\right]\Big| = \negl(\lambda)\,.
\end{equation*}
\end{lemma}

In contrast to the LWE case, the domain $\cX$ is a strict subset of $\Z_q^n$. Thus not all elements of the domain $\cX$ are part of a claw, and in particular, the above statement is not equivalent to the AHCB property required for $\cF_{\rm DDH}$. Thus, we need a restricted-domain version of \cref{lemma:ahcbside}:

\begin{lemma}[AHCB for restricted domain]
\label{lemma:ahcbresdom}
Consider the same setting as in \cref{lemma:ahcbside}.
Let 
\begin{align*}
    H'_\mathbf{s} &= \big\{(b,\mathbf{x},\mathbf{d},\mathbf{d}\cdot(J(\mathbf{x})\oplus J(\mathbf{x} - (-1)^b\mathbf{s}))) | b \in \bits, \mathbf{x} \in \cX_b, \mathbf{d} \in \hat{G}_{\mathbf{s}_{b\oplus1},b,\mathbf{x}}\big\} \;, \\
    \overline{H}'_\mathbf{s} &= \big\{(b,\mathbf{x},\mathbf{d},c) | (b,\mathbf{x},\mathbf{d},c\oplus 1) \in H_\mathbf{s}\big\}\;.
\end{align*} 
Then for any PPT algorithm $\mathcal{A}$ that takes as input $(q,\mathbb{G},g)$ as well as an element of $\mathbb{G}^{(n+1)\times(n+1)}$ and has outputs in $\bits\times\mathbb{Z}_q^{n+1}\times\bits^w\times\bits$, it holds that
\begin{equation*}
\Big|\Pr_{\kt\leftarrow \genalg_{k}(1^{\lambda})}\left[\mathcal{A}(k) \in H'_\mathbf{s}\right] - \Pr_{\kt\leftarrow \genalg_{k}(1^{\lambda})}\left[\mathcal{A}(k) \in\overline{H}'_\mathbf{s}\right]\Big| = \negl(\lambda)\,.
\end{equation*}
\end{lemma}

\begin{myproof}
Assume there exists a PPT algorithm $\cA$ such that 
\begin{equation*}
\Big|\Pr_{\kt\leftarrow \genalg_{k}(1^{\lambda})}\left[\mathcal{A}(k) \in H'_\mathbf{s}\right] - \Pr_{\kt\leftarrow \genalg_{k}(1^{\lambda})}\left[\mathcal{A}(k) \in\overline{H}'_\mathbf{s}\right]\Big| \,>\, 1/p(\lambda)\;.
\end{equation*}
for some polynomial $p: \mathbb{N} \rightarrow \mathbb{R}_+$. We will now show that this would imply the existence of an efficient adversary $\cB$ who contradicts \cref{lemma:ahcbside}. Consider the algorithm $\cB$ as defined in \cref{alg:adv2}.
\begin{algorithm}
\caption{Adversary $\cB$ contradicting \cref{lemma:ahcbside}.}
\label{alg:adv2}
\begin{myalgo}
\Require Key $k \in \cK_{\cF_{\rm DDH}}.$
\State Run $\cA$ on the input $k$, obtaining output $(b,\mathbf{x},\mathbf{d},c).$
\State Check if $\mathbf{x} \in \cX_b$. If yes, output $(b,\mathbf{x},\mathbf{d},c)$. Else, sample $c' \leftarrow_U \bits$ and output $(b,\mathbf{x},\mathbf{d},c')$. 
\end{myalgo}
\end{algorithm}
Note that $\cB$ is a PPT procedure since $\cA$ is PPT and the condition $x \in \cX_b$ is efficiently checkable by construction of the sets $\cX_b$. Let $H_\mathbf{s}, \overline{H}_\mathbf{s}$ be the sets defined as in \cref{lemma:ahcbside}. Then it holds that: 
\begin{align*}
&\Big|\Pr_{\kt\leftarrow \genalg_{k}(1^{\lambda})}\left[\mathcal{B}(k) \in H_\mathbf{s}\right] - \Pr_{\kt\leftarrow \genalg_{k}(1^{\lambda})}[\mathcal{B}(k) \in\overline{H}_\mathbf{s}]\Big| \\
 &\geq \Big|\Pr_{\kt\leftarrow \genalg_{k}(1^{\lambda})}\left[\mathcal{B}(k) = (b,\mathbf{x},\mathbf{d},c) \in H'_\mathbf{s}\right]  - \Pr_{\kt\leftarrow \genalg_{k}(1^{\lambda})}[\mathcal{B}(k) = (b,\mathbf{x},\mathbf{d},c) \in\overline{H}'_\mathbf{s}]\Big| \\
 &-  \Big|\Pr_{\kt\leftarrow \genalg_{k}(1^{\lambda})}\left[\mathcal{B}(k) = (b,\mathbf{x},\mathbf{d},c) \in H_\mathbf{s} \land \mathbf{x} \notin \cX_b\right] \\ &- \Pr_{\kt\leftarrow \genalg_{k}(1^{\lambda})}\left[\mathcal{B}(k) = (b,\mathbf{x},\mathbf{d},c) \in\overline{H}_\mathbf{s} \land \mathbf{x} \notin \cX_b\right]\Big| \\
 &=   \Big|\Pr_{\kt\leftarrow \genalg_{k}(1^{\lambda})}\left[\mathcal{A}(k) \in H'_\mathbf{s}\right]  - \Pr_{\kt\leftarrow \genalg_{k}(1^{\lambda})}\left[\mathcal{A}(k) \in\overline{H}'_\mathbf{s}\right]\Big| > 1/p(\lambda)\;,
\end{align*}
where the first step uses a generic triangle inequality and the second step follows from the definition of $\cB$. This constitutes a contradiction to \cref{lemma:ahcbside}.
\end{myproof}

The proof is concluded by noting that the statement made in \cref{lemma:ahcbresdom} is equivalent to the AHCB property we require for $\cF_{\rm DDH}$.
\begin{lemma}
\label{lemma:setcorr}
    Consider the same setting as in \cref{lemma:ahcbresdom}, let $\cR_k$ be the set of claw pairs of $\cF_{\rm DDH}$ and $H_k, \overline{H}_k$ the sets as defined in \cref{def:ntcf}. Then it holds that 
    \begin{align*}
    H'_\mathbf{s} = H_k, \hspace{10 pt} \overline{H}'_\mathbf{s} = \overline{H}_k.
\end{align*}
\end{lemma}
\begin{myproof}
    First, note that by construction, for any key-trapdoor pair $(k,t_k) \in \cK_{\cF_{\rm DDH}}$, we have that for all $\mathbf{x}_0,\mathbf{x}_1 \in \cX$, $(\mathbf{x}_0,\mathbf{x}_1) \in \cR_k$ iff 
    $\mathbf{x}_0 = \mathbf{x}_1 - \mathbf{s}$ (where $\mathbf{s}$ is the corresponding entry in $t_k$). 
    Secondly, by construction of the sets $\cX_b$, we have that for all keys $k$ (with associated secret $\mathbf{s} \in \bits^n$) and for every $\mathbf{x}_b \in \cX_b$ it holds that  $\mathbf{x}_{b\oplus1}\coloneqq\mathbf{x}-(-1)^b\mathbf{s} \in \cX$ and thus $(\mathbf{x}_0,\mathbf{x}_1) \in \cR_k$. Using these two facts one can easily conclude the statement by checking both inclusions. 
\end{myproof}

\section{Proof of \cref{thm:poq}} \label{app:e2soundness}
In this appendix, we provide the proof of \cref{thm:poq}.
This proof is essentially a slightly simplified version of the proof of \cref{thm:poq2}, but we include a full proof for completeness.

\begin{myproof}[\cref{thm:poq}]
~
\begin{myenumi}

    \item \textbf{Completeness.}
Consider the same QPT prover as in the proof of \cref{thm:poq2}, see \cref{alg:successfulqpt}. \hfill

\begin{myenumii}
\item Case $a = \mathtt{Im}$: 
Since $\mathcal{F}$ is injective invariant (\cref{def:injinv}), SAMP$_\mathcal{F} =$ SAMP$_\mathcal{G}$ and thus the state that $P$ prepares in \cref{eq:successful quantum prover state} is
$$
    \frac{1}{\sqrt{2|\cX|}} \sum_{\substack{b \in \bits \\ x\in \cX, y\in \cY}} \sqrt{(g_{k,b}(x))(y)}\ket{b}_\sB \ket{x}_\sX \ket{y}_\sY\,.
$$
Thus, for any $y$ that $P$ gets from the computational basis measurement of $\sY$, there exist $(b',x')$ such that $y \in \supp(g_{k,b'}(x'))$. 
Using condition~\ref{item:disjoint trapdoor injective pair} of the \cref{def:trapdoor injective function family} of $\mathcal{G}$, it follows that if $(k,t_k) \in \mathcal{I}_G$, 
then the verifier will obtain $\invalg_\cG(t_k, y) = (b',x')$, and then $\chkalg_\cG(k,y,b',x') = 1$, so that $P$ succeeds in the protocol.
The failure probability of $P$ in the case $a = \imtest$ is thus at most
$$
    \prss{\kt \leftarrow \genalg_\cG(1^\lambda)}{ (k,t_k)\notin \cI_\cG} = \negl(\lambda)\,.
$$

\item Case $a = \mathtt{Eq}$: Since this test is the same as in \protref{poq:3test}, by the same argumentation as in the proof of \cref{thm:poq2}, it follows that the success probability of $P$ in the equation test is lower bounded by $\injconstant\cF - \negl(\lambda)$.
\end{myenumii}

Thus, the overall success probability of $P$ in the protocol is $\frac{1 + \injconstant\cF}{2} - \negl(\lambda)$.

\item \textbf{Soundness:} Suppose that there exists a PPT adversary $\mathcal{A}$ that succeeds in the protocol with probability at least $\frac{3}{4} + \frac{1}{q(\lambda)}$ for some polynomial $q: \mathbb{N} \rightarrow \mathbb{R}_+$. 
We will use the extractability property of $\mathcal{G}$ to construct a PPT algorithm $\mathcal{B}$ which contradicts the AHCB property of $\mathcal{F}$ (item~\ref{item:ahcb} of \cref{def:ntcf}).

Let the distribution of the random coins of $\cA$ be $R$. 
For ease of notation, we refrain from writing the coin distribution $r\leftarrow R$ of $\mathcal{A}$ explicitly in the proof, but note that all probabilities in the proof are defined on average over this random coin distribution.
Let $\mathcal{A}^*$ be the extractor corresponding to $\cA$ from the extractability of $\cG$ (\cref{def:extractability}).
Denote by $S_\mathcal{A}^{\hspace{1pt}a}$ the event that $\mathcal{A}$ produces an output which passes the test of case $a \in \{\mathtt{Im},\mathtt{Eq}\}$. We begin by relating the probability of $\cA$, $\cA^*$ producing an image-preimage pair $y, (x,b)$ and the success probability of $\cA$ in the image test.
\begin{claim}
\label{claim:imtestbound2}
It holds that
\begin{equation*}
 \prss{\kt \leftarrow \genalg_{\mathcal{G}}(1^\lambda)}{S_\cA^{\imtest}} 
 \leq
 \prss{\kt \leftarrow \genalg_{\mathcal{G}}(1^\lambda)}{\chkalg_\cF(k,y,b,x) = 1}
 + \negl(\lambda)\,.
\end{equation*}
where $y$ is obtained from $(y,d,c) = \mathcal{A}(k,r)$, $(b,x) = \mathcal{A}^*(k,r)$ and $r$ are the random coins of $\mathcal{A}$.  
\end{claim}
\begin{claimproof}
Using that $\cA^*$ is the extractor associated to $\cA$, it follows by definition that
\begin{equation}
\label{eq:2test aux 1}
    \prss{\kt \leftarrow \genalg_\cG(1^\lambda)}{
        y \notin \supp\big(g_{k,b}(x)\big) \land y \in \supp \big\{ g_{k,\tilde{b}}(\tilde{x}) \big\}_{\substack{\tilde{b}\in\bits, \\ \tilde{x} \in \cX}}
    } = \negl(\lambda)\,.
\end{equation}
Let $(b',x') = \invalg_\cG(t_k,y)$.
Note that by definition of the $\chkalg_\cG$ algorithm (see Item~\ref{item:trapdoor injective function family, efficient range superposition} of \cref{def:trapdoor injective function family}), we have that $\chkalg_\cG(k,y,b,x) = 0 \land \chkalg_\cG(k,y,b',x') = 1$ is equivalent to $y \notin \supp\big(g_{k,b}(x)\big) \land y \in \supp\big(g_{k,b'}(x')\big)$, which implies the event appearing in \cref{eq:2test aux 1}.
Thus, we also have
\begin{equation}
\label{eq:exttoim}
    \prss{\kt \leftarrow \genalg_\cG(1^\lambda)}{
        \chkalg_\cG(k,y,b,x) = 0 \land \chkalg_\cG(k,y,b',x') = 1
    } = \negl(\lambda)\,.
\end{equation}
The event $S_\cA^\imtest$ corresponds to the event where $\cA$ returns $(y,d,c)$ such that $\chkalg_\cG(k,y,b',x') = 1$.
Therefore,
\begin{align}
\label{eq:chk0}
    &\prss{\substack{\kt \leftarrow \genalg_\cG(1^\lambda)}}{\chkalg_\cG(k,y,b,x) = 0 \land S_\cA^\imtest} \nonumber\\
    &\annotatesign{Def. $S_\cA^\imtest$}{=}\ \prss{\substack{\kt \leftarrow \genalg_\cG(1^\lambda)}}{\chkalg_\cG(k,y,b,x) = 0 \land \chkalg_\cG(k,y,b',x') = 1 }
    \hspace{10pt} \annotatesign{Eq.~\eqref{eq:exttoim}}{=}\ \hspace{10pt} \negl(\lambda)\,.
\end{align}

This allows us to upper bound the success probability of $\cA$ in the $\imtest$ case:
\begin{multline*}
    \prss{\kt \leftarrow \genalg_\cG(1^\lambda)}{S_\cA^\imtest} \\
    = \prss{\substack{\kt \leftarrow \genalg_\cG(1^\lambda) }}{\chkalg_\cG(k,y,b,x) = 0 \land S_\cA^\imtest}
    + \prss{\substack{\kt \leftarrow \genalg_\cG(1^\lambda)}}{\chkalg_\cG(k,y,b,x) = 1 \land S_\cA^\imtest} \nonumber\\
    \annotatesign{Eq.~\ref{eq:chk0}}{\leq}\ \hspace{10pt} \negl(\lambda) + 
    \prss{\kt \leftarrow \genalg_{\mathcal{G}}(1^\lambda)}{\chkalg_\cG(k,y,b,x) = 1}\,.
\end{multline*}
Recalling that $\chkalg_\mathcal{G}=\chkalg_\mathcal{F}$ by the definition of injective invariance (see \cref{def:injinv}) concludes the proof of the claim.
\end{claimproof}
We note that $\chkalg_\mathcal{F}$ does not use a trapdoor and is poly-time, $\cA$, $\cA^*$ are poly-time, and thus the concatenation with $\chkalg_\mathcal{F}$ is also poly-time.
Using that the key distributions of $\cF$ and $\cG$ are computationally indistinguishable (Item~\ref{item:inj inv key indist} of \cref{def:injinv}), \cref{claim:imtestbound2} implies that also
\begin{equation}
\label{eq:bridge3}
    \prss{\kt \leftarrow \genalg_\cG(1^\lambda)}{S_\cA^\imtest}
    \leq
    \prss{\kt \leftarrow \genalg_\cF(1^\lambda)}{\chkalg_\mathcal{F}(k,y,b,x) = 1}
    + \negl(\lambda)\,.
\end{equation}
We now proceed to lower-bound the probability that the prover-extractor pair (which we combine to a single algorithm $\cB$) holds both a preimage and a valid equation by the success probability of $\cA$ in the protocol.
\begin{claim}
\label{claim:advsuccessbound2}
There exists a PPT algorithm $\cB$ that takes a key $k\in \cK_\cF$ as input and produces outputs in $\bits\times\cX\times\bits^w\times\bits$ such that
\begin{multline*}
\prss{\kt \leftarrow \genalg_\cH(1^\lambda)}{S_\mathcal{A}^\mathtt{Im}}  + \prss{\kt \leftarrow \genalg_{\mathcal{F}}(1^\lambda)}{S_\mathcal{A}^\mathtt{Eq}}
\leq 1 + \prss{\kt \leftarrow \genalg_{\mathcal{F}}(1^\lambda)}{ \mathcal{B}(k) \in H_k} + \negl(\lambda)\,,
\end{multline*}
where $H_k$ refers to the set introduced in Item~\ref{item:ahcb} of \cref{def:ntcf}.   
\end{claim}
\begin{claimproof}
First, note that \cref{eq:bridge3} implies that
\begin{align}
\label{eq:anotherprobbound}
&\prss{\kt \leftarrow \genalg_{\mathcal{G}}(1^\lambda)}{S_\mathcal{A}^\mathtt{Im}} + \prss{\kt \leftarrow \genalg_{\mathcal{F}}(1^\lambda)}{S_\mathcal{A}^\mathtt{Eq}} \nonumber\\
\leq &\prss{\kt \leftarrow \genalg_{\mathcal{F}}(1^\lambda)}{\chkalg_\mathcal{F}(k,y,b,x) = 1} + \prss{\kt \leftarrow \genalg_{\mathcal{F}}(1^\lambda)}{S_\mathcal{A}^\mathtt{Eq}} + \negl(\lambda)\,.
\end{align}
Furthermore, it holds that: 
\begin{align}
&\prss{\kt \leftarrow \genalg_{\mathcal{F}}(1^\lambda)}{\chkalg_\mathcal{F}(k,y,b,x) = 1} + \prss{\kt \leftarrow \genalg_{\mathcal{F}}(1^\lambda)}{S_\mathcal{A}^\mathtt{Eq}} \nonumber\\
\leq&\ 1 + \prss{\kt \leftarrow \genalg_{\mathcal{F}}(1^\lambda)}{\chkalg_\mathcal{F}(k,y,b,x) = 1 \land S_\mathcal{A}^\mathtt{Eq}} \nonumber\\
\leq&\ 1 + \prss{\kt \leftarrow \genalg_{\mathcal{F}}(1^\lambda)}{\chkalg_\mathcal{F}(k,y,b,x) = 1 \land S_\mathcal{A}^\mathtt{Eq} \land (k,t_k) \in \mathcal{I}_\mathcal{G}} + \negl(\lambda)\,,
\end{align}
where in the last inequality, we used \cref{eq:ntcf bad keys}.
We now note that if $(k,t_k) \in \mathcal{I}_\mathcal{F}$, if $\chkalg_\cF(k,y,0,x_0) = \chkalg_\cF(k,y,1,x_1) = 1$ and $\chkalg_\cF(k,y,b,x) = 1$ (where $x_{\tilde b} = \invalg_\mathcal{F}(t_k,\tilde b,y)$) then $x = x_b$, because the inversion function returns the correct preimages $x_0,x_1$ under these conditions.
Inserting the definition of $S_\mathcal{A}^\mathtt{Eq}$, we can conclude that: 
\begin{align}
&\prss{\kt \leftarrow \genalg_{\mathcal{F}}(1^\lambda)}{\chkalg_\mathcal{F}(k,y,b,x) = 1 \land S_\mathcal{A}^\mathtt{Eq} \land (k,t_k) \in \mathcal{I}_\mathcal{G}} \\
\leq &\prss{\kt \leftarrow \genalg_{\mathcal{F}}(1^\lambda)}{
    \begin{gathered}
        c = d\cdot(J(x_0)\oplus J(x_1)) \land x = x_b \land
    (x_0,x_1) \in \mathcal{R}_k \\
         \land\, x_0 \in \cX_0 \land x_1 \in \cX_1 \land d\in \dset_{k,0,x_0}\cap \dset_{k,1,x_1}
    \end{gathered}
} \\
\leq &\prss{\kt \leftarrow \genalg_{\mathcal{F}}(1^\lambda)}{
    \begin{gathered}
        c = d\cdot(J(x_0)\oplus J(x_1)) \land x = x_b \land
    (x_0,x_1) \in \mathcal{R}_k\\
         \land\, x_b \in \cX_b \land d\in \dset_{k,0,x_0}\cap \dset_{k,1,x_1}
    \end{gathered}
}\,.
\end{align}
where $(y,d,c) = \cA(k,r)$, $(b,x) = \cA^*$, and $x_{\tilde b} = \invalg_\cF(t_k,\tilde b,y)$.
Now, we define in \cref{algo:soundness2} an adversary $\mathcal{B}$ which uses the algorithms $\mathcal{A}$ and $\mathcal{A}^*$.
\begin{algorithm}[ht]
\caption{Adversary $\cB$ breaking the AHCB property of $\cF$}
\label{algo:soundness2}
\begin{myalgo}
    \Require Key $k \in \cK_\cF$.
    \State Sample $r$ from the random coin distribution $R$ of $\cA$.
    \State Run $\mathcal{A}$ on the input $(k,r)$, obtaining the output $(y,d,c)$.
    \State Run $\mathcal{A}^*$ on the input $(k,r,y)$, obtaining the output $(b,x)$.
    \State Return $(b,x,d,c)$.
\end{myalgo}
\end{algorithm}
By definition of $H_k$ (see Item~\ref{item:ahcb} of \cref{def:ntcf}), it follows that: 
\begin{multline}
\label{eq:advsuccessprob}
\prss{\kt \leftarrow \genalg_{\mathcal{F}}(1^\lambda)}{
    \begin{gathered}
        c = d\cdot(J(x_0)\oplus J(x_1)) \land x = x_b \\
        \land
    (x_0,x_1) \in \mathcal{R}_k \land x_b \in \cX_b \land d\in \dset_{k,0,x_0}\cap \dset_{k,1,x_1}
    \end{gathered}
} \\
= \prss{\kt \leftarrow \genalg_{\mathcal{F}}(1^\lambda)}{ \mathcal{B}(k) \in H_k}\,.
\end{multline}
Combining \crefrange{eq:anotherprobbound}{eq:advsuccessprob} concludes the proof of the claim.
\end{claimproof}
Using \cref{claim:advsuccessbound2} and the assumption on the success probability of $\cA$ being at least $3/4 + 1/q(\lambda)$, we can conclude that there exists a PPT algorithm $\cB$ such that
\begin{equation*}
\prss{\kt \leftarrow \genalg_{\mathcal{F}}(1^\lambda)}{ \mathcal{B}(k) \in H_k} \geq \frac{1}{2} + \frac{2}{q(\lambda)} - \negl(\lambda)\,.
\end{equation*}
This constitutes a contradiction with the AHCB property (see Item~\ref{item:ahcb} of \cref{def:ntcf}). 
Therefore, the success probability of $\mathcal{A}$ must be upper-bounded by $\frac{3}{4} + \negl(\lambda)$. \qedhere
\end{myenumi}
\end{myproof}

\end{document}